\documentclass[12pt,a4paper,fleqn]{article}

\usepackage{lscape,exscale,amsthm, enumerate,adjustbox}
\usepackage[intlimits]{amsmath}
\usepackage{rawfonts}
\usepackage{latexsym}
\usepackage[cp850]{inputenc}
\usepackage{epsfig}
\usepackage{bbm}
\usepackage{enumitem}
\usepackage{float}
\usepackage{ulem}
\usepackage{xcolor, multirow}
\usepackage[pagewise]{lineno}

\usepackage[pass,letterpaper]{geometry}

\usepackage{xr} 

\RequirePackage[OT1]{fontenc}
\RequirePackage{amsthm,amsmath}

\usepackage[authoryear]{natbib}
\usepackage{setspace}

\newcommand{\bol}[1]{\mbox{\boldmath$#1$}}

\newcommand{\bSigma}{\bol{\Sigma}}

\newcommand{\bm}{\bol{\mu}}

\newcommand{\bb}{\mathbf{b}}

\newcommand{\bx}{\mathbf{X}}

\newcommand{\bQ}{\mathbf{Q}}

\newcommand{\by}{\mathbf{Y}}

\newcommand{\bP}{\mathbf{P}}

\newcommand{\bR}{\mathbf{R}}

\newcommand{\bw}{\mathbf{w}}

\newcommand{\bi}{\mathbf{1}}

\newcommand{\bI}{\mathbf{I}}

\newcommand{\bxi}{\boldsymbol{\xi}}
\newcommand{\bD}{\mathbf{D}}
\newcommand{\eps}{\pmb{\varepsilon}}

\newcommand{\bV}{\mathbf{V}}

\newcommand{\bS}{\mathbf{S}}

\newcommand{\sx}{\bar{\mathbf{x}}}

\newcommand{\sy}{\bar{\mathbf{y}}}
\newcommand{\btheta}{\boldsymbol{\theta}}
\newcommand{\bTheta}{\boldsymbol{\Theta}}

\newcommand{\bU}{\mathbf{U}}

\newcommand{\tbV}{\tilde{\mathbf{V}}}

\usepackage{stackrel}
\numberwithin{equation}{section}
\theoremstyle{plain}
\newtheorem{theorem}{Theorem}[section]

\newtheorem{proposition}{Proposition}[section]
\newtheorem{lemma}{Lemma}[section]
\newtheorem{corollary}{Corollary}[section]

\newtheorem{remark}{Remark}
\newtheorem*{remark*}{Remark}

\usepackage{ulem}

\begin{document}

\begin{center}
\vspace*{0.3cm} \noindent {\bf \large Optimal shrinkage-based portfolio selection in high dimensions}
\\

\vspace{0.5cm} \noindent {\sc  Taras Bodnar$^{a}$, Yarema Okhrin$^{b}$ and Nestor Parolya$^{c,}$\footnote{Corresponding author. E-mail address: n.parolya@tudelft.nl}}\\
\vspace{0.5cm} {\it \footnotesize $^a$ Department of Mathematics, Stockholm University, Roslagsv\"{a}gen 101, SE-10691 Stockholm, Sweden}\\
{\it \footnotesize $^b$ Department of Statistics, University of Augsburg, Universit\"{a}tsstr. 16, D-86159 Augsburg, Germany}  \\
{\it \footnotesize $^c$ Department of Applied Mathematics, Delft University of Technology, Mekelweg 4,
2628 CD Delft, The Netherlands}
 \end{center}

\begin{abstract}
In this paper we estimate the mean-variance portfolio in the high-dimensional case using the recent results from the theory of random matrices. We construct a linear shrinkage estimator which is distribution-free and is optimal in the sense of maximizing with probability $1$ the asymptotic out-of-sample expected utility, i.e., mean-variance objective function for different values of risk aversion coefficient which in particular leads to the maximization of the out-of-sample expected utility and to the minimization of the out-of-sample variance.
 One of the main features of our estimator is the inclusion of the estimation risk related to the sample mean vector into the high-dimensional portfolio optimization. The asymptotic properties of the new estimator are investigated when the number of assets $p$ and the sample size $n$ tend simultaneously to infinity such that $p/n \rightarrow c\in (0,+\infty)$. The results are obtained under weak assumptions imposed on the distribution of the asset returns, namely the existence of the $4+\varepsilon$ moments is only required.
 Thereafter we perform numerical and empirical studies where the small- and large-sample behavior of the derived estimator is investigated. The suggested estimator shows significant improvements over the existent approaches including the nonlinear shrinkage estimator and the three-fund portfolio rule, especially when the portfolio dimension is larger than the sample size. Moreover, it is robust to deviations from normality.
\end{abstract}

\noindent JEL Classification: G11, C13, C14, C58, C65\\
\noindent {\it Keywords}: expected utility portfolio, large-dimensional asymptotics, covariance matrix estimation, random matrix theory.

\section{Introduction}

In the seminal paper of \cite{mark1952} the author suggests to determine the optimal composition of a portfolio of financial assets by  minimizing the portfolio variance assuming that the expected portfolio return attains some prespecified fixed value. By varying this value we obtain the whole efficient frontier in the mean-standard deviation space. Despite of its simplicity, this approach justifies the advantages of diversification and is a standard technique and benchmark in asset management. Equivalently (see, \cite{tobin1958}, \cite{bodnar2013}) we can obtain the same portfolios by maximizing the expected quadratic utility (EU) with the optimization problem given by
\begin{equation}\label{MV}
\bw^\prime\bm_n-\dfrac{\gamma}{2}\bw^{\prime}\bSigma_n\bw\rightarrow max ~~\text{subject to}~~\bw^\prime\bi_p=1\,,
\end{equation}
where $\bw=(\omega_1,\ldots,\omega_p)^\prime$ is the vector of portfolio weights, $\bi_p$ is the $p$-dimensional vector of ones,  $\bm_n$ and $\bSigma_n$ are the $p$-dimensional mean vector and the $p\times p$ covariance matrix of asset returns, respectively. The quantity $\gamma>0$ determines the investor's behavior towards risk.
It must be noted that the maximization of the mean-variance objective function (\ref{MV}) is equivalent to the maximization of the exponential utility (CARA) function under the assumption of normality of the asset returns. In this case $\gamma$ equals the investor's absolute risk aversion coefficient (see, e.g., \cite{pratt1964}).

The solution of the optimization problem (\ref{MV}) is well known and it is given by
\begin{equation}\label{EUp}
\bw_{EU}=\bw_{GMV}+\gamma^{-1}\bQ_n\bm_n\,,
\end{equation}
where
\begin{equation}\label{Q}
\bQ_n=\bSigma_n^{-1}-\dfrac{\bSigma_n^{-1}\bi_p\bi_p^\prime\bSigma_n^{-1}}{\bi_p^\prime\bSigma_n^{-1}\bi_p}
\end{equation}
and
\begin{equation}\label{GMV}
\bw_{GMV}=\dfrac{\bSigma_n^{-1}\bi_p}{\bi_p^\prime\bSigma_n^{-1}\bi_p}
\end{equation}
is the vector of the weights of the global minimum variance (GMV) portfolio. By changing the risk-aversion  coefficient $\gamma \in (0,\infty)$ we obtain the set of optimal portfolios. \cite{merton1972} proved that this set is a parabola in the mean-variance (R-V) space (cf. \cite{bodnar2009}) given by
\begin{equation}\label{eff_intro}
(R-R_{GMV})^2=s(V-V_{GMV}),
\end{equation}
where
\begin{equation}\label{RV_GMV_intro}
R_{GMV}=\dfrac{\bm_n^\prime\bSigma_n^{-1}\bi_p}{\bi_p^\prime\bSigma_n^{-1}\bi_p}
\quad \text{and} \quad
V_{GMV}=\dfrac{1}{\bi_p^\prime\bSigma_n^{-1}\bi_p}
\end{equation}
are the expected return and the variance of the GMV portfolio, and
\begin{equation}\label{s_intro}
s=\bm_n^\prime \bQ_n \bm_n
\end{equation}
is its slope parameter. The quantity $s$ is always non-negative since $\bQ_n$ is a positive semidefinite matrix. Moreover, when $s$ is equal to zero, then the efficient frontier degenerates into a straight line with the GMV portfolio being the only optimal portfolio.

In practice, however, the above  mentioned approach of constructing an optimal portfolio frequently shows  poor out-of-sample performance in terms of various performance measures. Even naive portfolio strategies, e.g., equally weighted portfolio (see, \cite{demiguel2009}), often outperform the mean-variance strategy.  One of the reasons is the estimation risk. The unknown parameters $\bm_n$ and $\bSigma_n$ have to be estimated using historical data on asset returns. This results in the ''plug-in'' estimator of the EU portfolio (\ref{EUp}) which is a traditional and simple way to evaluate the portfolio in practice. This estimator is constructed by replacing  the mean vector $\bm_n$ and the covariance matrix $\bSigma_n$ with their sample counterparts in (\ref{EUp}). \cite{okhrin2006} derive the expectation and the variance of the sample portfolio weights under the assumption that the asset returns follow a multivariate normal distribution, whereas \cite{bodnar2011} obtain the exact finite-sample distribution. Recently, \cite{bodnarmp2015}  extended these results to the case $n < p$.

The estimation  of the parameters has a negative impact on the performance of the asset allocation strategy. This is noted in a series of papers with \cite{merton1980}, \cite{bestgrauer1991}, \cite{chopraz1993} among others. Several approaches have arisen to reduce the consequences of the estimation risk. One strand of research opts for the Bayesian framework and using appropriate priors takes the estimation risk into account already while building the portfolio. The second strand relies on the shrinkage techniques and is related to the method exploited in this paper. A straightforward way to improve the properties of the estimators for $\bm_n$ and $\bSigma_n$ is to use the shrinkage approach (see, \cite{jorion1986}, \cite{lw2004}). Alternatively, one may apply the shrinkage estimation to the portfolio weights directly. \cite{golokh2007} consider the multivariate shrinkage estimator by shrinking the portfolios with and without the riskless asset to an arbitrary static portfolio. A similar technique is used by \cite{frahm2010}, who construct a feasible shrinkage estimator for the GMV portfolio which dominates the traditional one. At last, \cite{bodnar2014} suggest  a shrinkage estimator for the GMV portfolio which is feasible even for the singular sample covariance matrix.

An important issue nowadays is, however, the asset allocation for large portfolios. The sample estimators work well only in the case when the number of assets $p$ is fixed and substantially smaller than the sample size $n$.  This case is known as the standard asymptotics in statistics (see, \cite{lecam2000}). Under this asymptotics the traditional sample estimator is a consistent estimator for the EU portfolio. But what happens when the dimension $p$ and the sample size $n$ are comparable of size, say $p=900$ and $n=1000$?  Technically, here we are in the situation when both the number of assets $p$ and the sample size $n$ tend to infinity. In the case when $p/n$ tends to some concentration ratio $c>0$ this asymptotics is known as high-dimensional asymptotics or ``Kolmogorov'' asymptotics (see, e.g., \cite{baisil2010}). If $c$ is close to one the sample covariance matrix tends to be close to a singular one and when $c>1$ it becomes singular.
Thus it is very unstable and tends to under- or overestimate the true parameters for $c$ smaller but close to 1 (see, \cite{baishi2011}). As a result, the sample estimator of the EU portfolio behaves badly in this case both from the theoretical and practical points of view (see, e.g., \cite{karoui2010, RubioMestrePalomar2012}). For $c>1$ the inverse sample covariance matrix does not exist and the portfolio cannot be constructed in the traditional way.

Taking the above mentioned information into account the aim of the paper is to construct a feasible and simple shrinkage estimator of the EU portfolio which is optimal in an asymptotic sense and is additionally distribution-free. The estimator is developed using the fast growing branch of probability theory, namely random matrix theory. The main result of this theory is proved by \cite{marpas1967} and further extended under very general conditions by \cite{silverstein1995}. Now it is called Mar$\breve{\text{c}}$enko-Pastur equation. Its importance arises in many areas of science because it shows how the true covariance matrix and its sample estimator are connected asymptotically. Knowing this  we can build suitable estimators for high-dimensional quantities which depend on $\bSigma_n$. In our case this refers to the shrinkage intensities. Note however, that the optimal shrinkage intensity depends again on the unknown characteristics of the asset returns. To overcome this problem we derive consistent estimators for specific functions (quadratic and bilinear forms) of the inverse sample covariance matrix and mean vector. Furthermore, we succeed to provide consistent estimators for the optimal shrinkage intensities too. Additional advantage of our approach is the simultaneous treatment of estimation risks of both the covariance matrix and the mean vector. In particular we contribute to the existent literature (see, \cite{ledoit2017nonlinear}) by weakening the assumption imposed on the mean vector of the asset returns.

It is worth mentioning that there are clear links between the subject of the paper and classical methods in statistical signal processing. The data generating process considered in the paper encompasses a broad range of system configurations described by the general vector channel model. Moreover, as for the aforementioned mean-variance portfolio optimization problem, usual linear filtering schemes solving typical signal waveform estimation and detection problems in signal array processing and wireless communications are based on the estimation of the unknown population covariance matrix. Famous example is the equivalence of the GMV portfolio to the so-called Capon or minimum variance distortionless response (MVDR) beamformer (see, \cite{verdu1998, vantrees2002}).

The rest of paper is organized as follows. In the next section, we construct a shrinkage estimator for the optimal portfolio weights obtained by shrinking the EU portfolio weights to an arbitrary target portfolio. The oracle shrinkage intensity and the corresponding feasible bona-fide estimators for $c<1$ and $c>1$ are established as well. The derived results are evaluated in Section 3 in extensive simulation and empirical studies. All proofs are moved to the Appendix presented in the supplementary material.

\section{Optimal shrinkage estimator of mean-variance portfolio}\label{sec:main}

Let $\by_n=(\mathbf{y}_1,\mathbf{y}_2,...,\mathbf{y}_n)$ be the $p \times n$ data matrix which consists of $n$ vectors of the returns on $p\equiv p(n)$ assets. Let $E(\mathbf{y}_i)=\bm_n$ and $Cov(\mathbf{y}_i)=\bSigma_n$ for $i \in 1,...,n$. We assume that $p/n\rightarrow c\in (0, +\infty)$ as $n\rightarrow\infty$. This type of limiting behavior is known as "the large dimensional asymptotics" or "Kolmogorov asymptotics". In this case the traditional sample estimators perform poorly or even very poorly and tend to over/underestimate the unknown parameters of the asset returns, e.g., the mean vector and the covariance matrix.

Throughout the paper it is assumed that there exists a $p\times n$ random matrix $\bx_n$ which consists of independent and identically distributed (i.i.d.) real random variables with zero mean and unit variance such that
\begin{equation}\label{obs}
 \by_n= \bm_n \bi_n^\prime + \bSigma_n^{\frac{1}{2}}\bx_n \,.
 \end{equation}
It must be noted that the observation matrix $\by_n$ has dependent rows but independent columns. Broadly speaking, this means that we allow arbitrary cross-sectional correlations of the asset returns but assume their independence over time. Although this assumption looks quite restrictive for financial applications, there exist stronger results from random matrix theory which show that the model can be extended to (weakly) dependent variables by demanding more complicated conditions on the elements of $\by_n$ (see, \cite{baizhou2008}) or by controlling the number of dependent entries as dimension increases (see, \cite{huipan2010}, \cite{friesen2013}, \cite{weietal2016}).  Although our findings can still be used when weak serial dependence structure is present between the observation vectors, like in the case of uncorrelated GARCH (generalized autoregressive conditional heteroscedastic) processes or similar ones (see, e.g., the simulation study in \cite{bodnar2021sampling}), we suspect substantial changes in the analytical expressions stated in the theorems for strongly correlated observation vectors, like in the case of VAR (vector autoregressive) processes. In such situations, the estimator will depend on the autocorrelation matrices of the underlying stochastic model and the theoretical results of the paper must be adjusted correspondingly. This interesting and important topic is not treated in the paper and is left for future research.

Nevertheless, if the entries of matrix $\by_n$ are weakly dependent or so called $m$-dependent,  this will only make the proofs more technical, but leave the results unchanged. For that reason we assume independent in time asset returns only to simplify the proofs of the main theorems and make them as transparent as possible. The three assumptions which are used throughout the paper are the following:

\begin{enumerate}

\item[(A1)] The covariance matrix of the asset returns $\bSigma_n$ is a nonrandom $p$-dimensional positive definite matrix.

\item[(A2)] The elements of the matrix $\bx_n$ have uniformly bounded $4+\varepsilon$ moments for some $\varepsilon>0$.

\item[(A3)] The efficient frontier is asymptotically a non-degenerate object, i.e. for its slope parameter it holds that $s=\bm_n^\prime\bQ_n\bm_n>0$ uniformly in $p$.

\end{enumerate}

All of these regularity assumptions are general enough to fit many real world situations. The assumption (A1) together with (\ref{obs}) are usual for financial and statistical problems and they impose no strong restrictions. The assumption (A2) is a technical one. Although we demand the existence of moments of order a bit higher than four, this is solely due to the fact that the almost sure convergence is employed in the formulation of the theoretical results. In case of the convergence in probability the existence of exactly the fourth moment is sufficient. Indeed, it can be easily shown that this extra $\varepsilon$ follows from the Borel-Cantelli lemma (see \cite{rubmes2011}[Proof of Lemma 4]). The assumption (A3) has an important financial interpretation. It ensures that the efficient frontier is a parabola in the mean-variance space as defined in \eqref{eff_intro} and it does not degenerate into a line parallel to the variance axis (cf., \cite{bodnar2010unbiased}). In the latter case, the only optimal portfolio is the GMV portfolio \eqref{GMV}, a special case of the EU portfolio \eqref{EUp} with $\gamma=\infty$, and its shrinkage estimators have already been developed in \citet{frahm2010} and \citet{bodnar2014}. The assumption (A3) can be tested in practice by using Theorem 1 of \cite{bodnar2021statistical}.

The sample covariance matrix is given by
\begin{equation}\label{samplecov}
 \bS_n=\dfrac{1}{n}\by_n(\bI_n-\frac{1}{n}\bi_n\bi_n^\prime)\by_n^{\prime}=\dfrac{1}{n}\bSigma_n^{\frac{1}{2}}\bx_n(\bI_n-\frac{1}{n}\bi_n\bi_n^\prime)\bx_n^{\prime}\bSigma_n^{\frac{1}{2}}\,,
 \end{equation}
where the symbol $\bI_n$ stands for the $n$-dimensional identity matrix. The sample mean vector becomes
\begin{equation}
 \sy_n= \dfrac{1}{n}\by_n\bi_n=\bm_n+\bSigma_n^{\frac{1}{2}}\sx_n~~\text{with}~\sx_n=\dfrac{1}{n}\bx_n\bi_n\,.
\end{equation}

\subsection{Oracle estimator. Case $c<1$}

In this section we consider the optimal shrinkage estimator for the EU portfolio weights presented in the introduction by finding the shrinkage parameter $\alpha$ and fixing some target portfolio $\mathbf{b}$.

The resulting estimator for $c<1$ is given by
\begin{equation}\label{gse}
  \hat{\bw}_{GSE}=\alpha_n\hat{\bw}_S+(1-\alpha_n)\mathbf{b}~~\text{with}~\mathbf{b}^\prime\bi_p=1\,,
\end{equation}
where the vector $\hat{\bw}_S$ is the sample estimator of the EU portfolio given in (\ref{EUp}), namely
\begin{equation}
\hat{\bw}_S=\dfrac{\bS_n^{-1}\bi_p}{\bi_p^\prime\bS_n^{-1}\bi_p}+\gamma^{-1}\hat{\bQ}_n\sy_n\,
\end{equation}
with
\begin{equation}\label{hatQ}
\hat{\bQ}_n=\bS_n^{-1}-\dfrac{\bS_n^{-1}\bi_p\bi_p^\prime\bS_n^{-1}}{\bi_p^\prime\bS_n^{-1}\bi_p}\,.
\end{equation}
The target portfolio $\mathbf{b}\in\mathbbm{R}^p$ is a given nonrandom (or random, but independent of $\by_n$) vector with $\mathbf{b}^\prime\bi_p=1$.
 No assumption is imposed on the shrinkage intensity $\alpha_n$ which is the object of our interest.

The aim is now to find the optimal shrinkage intensity for a given nonrandom target portfolio $\mathbf{b}$. For that reason we introduce a unified mean-variance objective function in order to calibrate the shrinkage intensity $\alpha_n$. Consider the following optimization problem
\begin{equation}\label{utility}
U(\beta)=\hat{\bw}_{GSE}^\prime(\alpha_n)\bm_n-\dfrac{\beta}{2}\hat{\bw}_{GSE}^\prime(\alpha_n)\bSigma_n\hat{\bw}_{GSE}(\alpha_n)\longrightarrow max~~
\text{with respect to}~~\alpha_n\,.
\end{equation}
Obviously, the mean-variance objectives (\ref{MV}) and (\ref{utility}) coincide if $\beta=\gamma$. Other special values of $\beta$ which lead to widely used out-of-sample performance measures we summarize in the following proposition
\begin{proposition}[Calibration criteria]\label{calibr}
  The optimization problem (\ref{utility}) is equivalent to
  \begin{itemize}
  \item[(i)] maximization of the mean-variance objective (\ref{MV}) if $\beta=\gamma$,
  \item[(ii)] minimization of the out-of-sample variance $\hat{\bw}_{GSE}^\prime(\alpha_n)\bSigma_n\hat{\bw}_{GSE}(\alpha_n)$ if $\beta\to\infty$,
   \end{itemize}
\end{proposition}

The proof of Proposition \ref{calibr} follows from the fact that all optimal mean-variance portfolios can be obtained by maximizing the expected quadratic utility function with a specific risk aversion coefficient. As a result, the global minimum variance portfolio is a partial solution of the optimization problem \eqref{MV}.
The presentation of the calibration criterion (\ref{utility}) provides an elegant way how to find the optimal shrinkage intensity $\alpha_n=\alpha_n(\beta)$ in a unified manner for several popular out-of-sample loss functions and compare them just by changing the parameter $\beta$.
In Section 2.3, we provide consistent estimates of these quantities under high-dimensional asymptotic regime $p/n\to c >0$ for $(p, n)\to \infty$.

It is worth mentioning that the coefficient $\beta$ has an interesting interpretation from statistical point of view. While coefficient $\gamma$ controls for investor attitude towards financial risk ("in-sample risk"), the parameter $\beta$ stays for controlling the estimation risk ("out-of-sample risk"). This implies that even the mean-variance investor with arbitrary $\gamma>0$ could choose $\beta\to\infty$ if she/he is interested, for example, in the minimization of the out-of-sample variance of the estimated portfolio.

The unified calibration criterion (\ref{utility}) can be rewritten as
{\small
\begin{eqnarray}\label{minvar}
&&U(\beta)=\alpha_n\hat{\bw}_S^\prime\bm_n+(1-\alpha_n)\mathbf{b}^\prime\bm_n-\dfrac{\beta}{2}\left(\alpha^2_n\hat{\bw}_S^\prime\bSigma_n\hat{\bw}_S
+2\alpha_n(1-\alpha_n)\mathbf{b}^\prime\bSigma_n\hat{\bw}_S+(1-\alpha_n)^2\mathbf{b}^\prime\bSigma_n\mathbf{b}\right)\rightarrow max\nonumber\\
&&\text{with respect to}~~\alpha_n\,.
\end{eqnarray}
}

Next, taking the derivative of $U$ with respect to $\alpha_n$ and setting it equal to zero we get
{\small
\begin{eqnarray}\label{der2}
&&\dfrac{\partial U}{\partial\alpha_n}=(\hat{\bw}_S-\mathbf{b})^\prime\bm_n-\beta\left(\alpha_n\hat{\bw}_S^\prime\bSigma_n\hat{\bw}_S+(1-2\alpha_n)\mathbf{b}^\prime\bSigma_n\hat{\bw}_S-(1-\alpha_n)\mathbf{b}^\prime\bSigma_n\mathbf{b}\right) \overset{!}{=}0\nonumber\,.
\end{eqnarray}
}
 From the last equation it is easy to find the optimal shrinkage intensity $\alpha_n^*$ given by
 \begin{equation}\label{alfa}
 \alpha_n^*=\beta^{-1}\dfrac{(\hat{\bw}_S-\mathbf{b})^\prime(\bm_n-\beta\bSigma_n\mathbf{b})}{(\hat{\bw}_S-\mathbf{b})^\prime\bSigma_n(\hat{\bw}_S-\mathbf{b})}\,.
 \end{equation}
To ensure that $\alpha_n^*$ is the unique maximizer of (\ref{utility}) the second derivative of $U$ must be negative, which is always fulfilled. Indeed, it follows from the positive definitiveness of the matrix $\bSigma_n$, namely
\begin{equation}\label{secder}
\dfrac{\partial^2 U}{\partial\alpha^2_n}=-\beta(\hat{\bw}_S-\mathbf{b})^\prime\bSigma_n(\hat{\bw}_S-\mathbf{b})<0\,.
\end{equation}

In the next theorem we derive the asymptotic properties of the optimal shrinkage intensity $\alpha^{*}_n$ under large-dimensional asymptotics.

\begin{theorem}\label{th1}
Assume (A1)-(A3). Then  it holds that
\begin{equation*}
\left|\alpha^*_n-\alpha^*\right|\stackrel{a.s.}{\longrightarrow}0 ~~ \text{for}
~~\dfrac{p}{n}\rightarrow c\in(0, 1)~~ \text{as} ~~n\rightarrow\infty
\end{equation*}
with
\begin{equation}\label{ainfty}
\alpha^*=\beta^{-1}\dfrac{(R_{GMV}-R_b)\left(1+\dfrac{\beta/\gamma}{1-c}\right)+\beta(V_b-V_{GMV})+\dfrac{\gamma^{-1}}{1-c}s}{\dfrac{1}{1-c}V_{GMV}-2\left(V_{GMV} + \frac{\gamma^{-1}}{1-c}(R_b-R_{GMV})\right)+\gamma^{-2}\left(\dfrac{s}{(1-c)^3}+\dfrac{c}{(1-c)^3}\right)+V_b},
\end{equation}
where the parameters of the efficient frontier $R_{GMV}$, $V_{GMV}$ and $s$ are given in \eqref{RV_GMV_intro} and \eqref{s_intro},
respectively. The quantities $R_b=\mathbf{b}^\prime\bm_n$ and $V_b=\mathbf{b}^\prime\bSigma_n\mathbf{b}$ denote the expected return and the variance of the target portfolio $\mathbf{b}$.
\end{theorem}

Next, we assess the performance of the classical estimator of the portfolio weights $\hat{\bw}_S$ and the optimal shrinkage weights $\hat{\bw}_{GSE}$. As a measure of performance we consider the relative increase in the utility of the portfolio return compared to the portfolio based on true parameters of asset returns. The results are summarized in the following corollary.

\begin{corollary}\label{cor1}
(a) Let $U_{EU}$ and $U_S$ be the mean-variance objectives in (\ref{MV}) for the true EU portfolio and its traditional estimator. Then under the assumptions of Theorem \ref{th1}, the relative loss of the traditional estimator of the EU portfolio is given by
\begin{equation}\label{eq:lossS}
L_S = \frac{U_{EU}-U_S}{U_{EU}} \stackrel{a.s.}{\longrightarrow} \frac{\frac{\gamma}{2} \left(\frac{1}{1-c}-1\right)\cdot V_{GMV} +  \gamma^{-1} \left({\frac{1}{2}}-\frac{1}{(1-c)} + \frac{1}{2(1-c)^3}\right) \cdot s + \frac{1}{2\gamma}\cdot \frac{c}{(1-c)^{ 3 } } } {R_{GMV}+\frac{1}{2}\gamma^{-1}\cdot s - \frac{\gamma}{2} V_{GMV}}
\end{equation}
for $\frac{p}{n}\rightarrow c\in(0,1)$ as $n\rightarrow \infty$.

(b) Let $U_{GSE}$ be the expected quadratic utility for  optimal shrinkage estimator of the EU portfolio. Under the assumptions of Theorem \ref{th1}, the relative loss of the optimal shrinkage estimator is given by
\begin{equation}\label{eq:lossS}
L_{GSE} = \frac{U_{EU}-U_{GSE}}{U_{EU}} \stackrel{a.s.}{\longrightarrow} (\alpha^*)^2 L_S + (1-\alpha^*)^2 L_{\bb} +\alpha^*(1-\alpha^*)\frac{c}{1-c}
\frac{R_b-R_{GMV}-\gamma^{-1} s}{U_{EU}}
\end{equation}
for $\frac{p}{n}\rightarrow c\in(0,1)$ as $n\rightarrow \infty$ with $L_{\bb}=(U_{EU}-U_{\bb})/U_{EU}$ is the relative loss in the expected utility $U_{\bb}$ of the target portfolio $\bb$.
\end{corollary}

\subsection{Oracle estimator. Case $c>1$.}
Here, similarly as in \cite{bodnar2014}, we will use the generalized inverse of the sample covariance matrix $\bS_n$. Particularly, we use the following generalized inverse of the sample covariance matrix $\bS_n$
\begin{equation}\label{geninverse}
\bS_n^*=\bSigma_n^{-1/2}\left(\dfrac{1}{n}\bx_n\bx_n^\prime - \sx_n\sx_n^\prime\right)^+\bSigma_n^{-1/2}\,,
\end{equation}
where $'+'$ denotes the Moore-Penrose inverse. It can be shown that $\bS_n^*$ is a generalized inverse of $\bS_n$ satisfying $\bS_n^*\bS_n\bS_n^*=\bS_n^*$ and $\bS_n\bS_n^*\bS_n=\bS_n$. However, $\bS_n^*$ is not exactly equal to the Moore-Penrose inverse because it does not satisfy the conditions $(\bS_n^*\bS_n)^\prime=\bS_n^*\bS_n$ and $(\bS_n\bS_n^*)^\prime=\bS_n\bS_n^*$. In case $c<1$ the generalized inverse $\bS_n^*$ coincides with the usual inverse $\bS_n^{-1}$. Moreover, if $\bSigma_n$ is a multiple of identity matrix, then $\bS_n^*$ is equal to the Moore-Penrose inverse $\bS_n^+$. In this section, $\bS_n^*$ is used only to determine an oracle estimator for the weights of the EU portfolio. The bona fide estimator is constructed in the next section.

Thus, the oracle estimator for $c>1$ is given by
\begin{equation}\label{gse1}
\hat{\bw}^*_{GSE}=\alpha_n^+\hat{\bw}_{S^*}+(1-\alpha_n^+)\mathbf{b}~~\text{with}~\mathbf{b}^\prime\bi_p=1,
\end{equation}
where the vector $\hat{\bw}_{S^*}$ is the sample estimator of the EU portfolio given in (\ref{EUp}), namely
\begin{equation}
\hat{\bw}_{S^*} = \dfrac{\bS_n^{*}\bi_p}{\bi_p^\prime\bS_n^{*}\bi_p}+\gamma^{-1}\hat{\bQ}^*_n\sy_n\,
\end{equation}
with
\begin{equation}\label{hatQ1}
\hat{\bQ}^*_n=\bS_n^{*}-\dfrac{\bS_n^{*}\bi_p\bi_p^\prime\bS_n^{*}}{\bi_p^\prime\bS_n^{*}\bi_p}\,.
\end{equation}
Again, the shrinkage intensity $\alpha_n^+$ is the object of our interest. In order to save place we skip the optimization procedure for $\alpha_n^+$ as it is only slightly different from the case $c<1$. The optimal shrinkage intensity $\alpha_n^+$ in case $c>1$ is given by
 \begin{equation}\label{alfa1}
 \alpha_n^+=\beta^{-1}\dfrac{(\hat{\bw}_{S^*}-\mathbf{b})^\prime(\bm_n-\beta\bSigma_n\mathbf{b})}{(\hat{\bw}_{S^*}-\mathbf{b})^\prime\bSigma_n(\hat{\bw}_{S^*}-\mathbf{b})}\,.
 \end{equation}

In the next theorem we find the asymptotic equivalent quantity for $\alpha_n^+$ in the case $p/n\rightarrow c\in(1, +\infty)$ as $n\rightarrow\infty$.

\begin{theorem}\label{th2}
Assume (A1)-(A3). Then  it holds that
\begin{equation*}
\left|\alpha^+_n-\alpha^+\right| \stackrel{a.s.}{\longrightarrow} 0 ~~\text{for} ~~\dfrac{p}{n}\rightarrow c\in(1, +\infty)~~\text{as} ~~n\rightarrow\infty
\end{equation*}
with
\begin{equation}\label{ainfty2}
\alpha^+=\beta^{-1}\dfrac{(R_{GMV}-R_b)\left(1+\dfrac{\beta/\gamma}{c(c-1)}\right)+\beta(V_b-V_{GMV})+\dfrac{\gamma^{-1}}{c(c-1)}s}{\dfrac{c^2}{(c-1)}V_{GMV}-2\left(V_{GMV} + \frac{\gamma^{-1}}{c(c-1)}(R_b-R_{GMV})\right)+\dfrac{\gamma^{-2}}{(c-1)^3}(s+c^2)+V_b}\,,
\end{equation}
where $R_{GMV}$, $V_{GMV}$, $R_b$, $V_b$, and $s$ are defined in Theorem \ref{th1}.
\end{theorem}

Next, as for the case $c<1$, we provide here the expression for the relative losses.

\begin{corollary}\label{cor-th2}
(a) Let $U_{EU}$ and $U_S$ be the mean-variance objectives in (\ref{MV}) for the true EU portfolio and its traditional estimator. Then under the assumptions of Theorem \ref{th2}, the relative loss of the traditional estimator of the EU portfolio is given by
\begin{equation}\label{eq:lossS}
L_S = \frac{U_{EU}-U_S}{U_{EU}} \stackrel{a.s.}{\longrightarrow}\frac{\frac{ \gamma }{2} \left(\frac{c^2}{c-1}-1\right)\cdot V_{GMV} +  \gamma^{-1} \left(\frac{1}{2}-\frac{1}{c(c-1)} + \frac{1}{2(c-1)^3}\right) \cdot s + \frac{1}{2\gamma} \cdot \frac{c^2}{(c-1)^3}}{R_{GMV}+ \frac{1}{2}\gamma^{-1}  \cdot s - \frac{\gamma}{2} V_{GMV}}
\end{equation}
for $\frac{p}{n}\rightarrow c\in(1,+\infty)$ as $n\rightarrow \infty$.

(b) Let $U_{GSE}$ be the expected quadratic utility for the optimal shrinkage estimator of the EU portfolio. Under the assumptions of Theorem \ref{th2}, the relative loss of the optimal shrinkage estimator is given by
\begin{equation}\label{eq:lossGSE}
L_{GSE} = \frac{U_{EU}-U_{GSE}}{U_{EU}} \stackrel{a.s.}{\longrightarrow}(\alpha^+)^2 L_S + (1-\alpha^+)^2 L_{\bb}+\alpha^+(1-\alpha^+)\frac{1+c-c^2}{c(c-1)}\frac{R_b-R_{GMV}-\gamma^{-1} s}{U_{EU}}
\end{equation}
for $\frac{p}{n}\rightarrow c\in(1,+\infty)$ as $n\rightarrow \infty$ with $L_{\bb}=(U_{EU}-U_{\bb})/U_{EU}$ is the relative loss in the expected utility $U_{\bb}$ of the target portfolio $\bb$.
\end{corollary}

\subsection{Estimation of unknown parameters. Bona fide estimator}

The limiting shrinkage intensities $\alpha^*$ and $\alpha^+$ are not feasible in practice, because they depend on $R_{GMV}$, $V_{GMV}$, $s$, $R_b$, and $V_b$ which are unknown quantities. In this subsection we derive consistent estimators for  $\alpha^*$ and $\alpha^+$. These results are summarized in two propositions dealing with the cases $c\in(0, 1)$ and $c\in(1, \infty)$, respectively. The statements  follow directly from the proofs of Theorems \ref{th1} and \ref{th2} that are provided in the supplement of the paper.
\begin{proposition}\label{prop1} The consistent estimator for the limiting optimal shrinkage intensity $\alpha^*$ under large dimensional asymptotics $p/n\rightarrow c<1~ \text{as} ~n\rightarrow\infty$ is given by
\begin{equation}\label{a:bf}
 \widehat{\alpha}^*=\beta^{-1}\dfrac{(\hat{R}_{c}-\hat{R}_b)\left(1+\dfrac{\beta/\gamma}{1-p/n}\right)+\beta(\hat{V}_b-\hat{V}_{c})+\dfrac{\gamma^{-1}}{1-p/n}\hat{s}_{c}}{\dfrac{1}{1-p/n}\hat{V}_{c}-2\left(\hat{V}_{c} + \frac{\gamma^{-1}}{1-p/n}(\hat{R}_b-\hat{R}_{c})\right)+\gamma^{-2}\left(\dfrac{\hat{s}_c}{(1-p/n)^3}+\dfrac{p/n}{(1-p/n)^3}\right)+\hat{V}_b}
\end{equation}
where $\hat{R}_{c}$, $\hat{V}_{c}$, $\hat{s}_{c}$, $\hat{R}_{b}$ and $\hat{V}_{b}$ are given by
 \begin{eqnarray}
\hat{R}_{c}&=&\hat{R}_{GMV}\label{hRc}\\
\hat{V}_{c}&=&\frac{1}{1-p/n}\hat{V}_{GMV}\label{hVc}\\
\hat{s}_{c}&=&(1-p/n)\hat{s}-p/n\label{hsc}\\
 \hat{R}_b&=& \mathbf{b}^\prime\sy_n\label{hRb}\\
 \hat{V}_b&=& \mathbf{b}^\prime\bS_n\mathbf{b}\label{hVb}\,,
\end{eqnarray}
which are also ratio consistent estimators for $R_{GMV}$, $V_{GMV}$, $s$, $R_b$, and $V_b$, respectively, while $\hat{R}_{GMV}$, $\hat{V}_{GMV}$ and $\hat{s}$ are traditional plug-in estimators.
\end{proposition}
Using Proposition \ref{prop1} we can immediately construct a bona-fide estimator for the expected utility portfolio weights in case $c<1$. It holds that
\begin{equation}\label{bonafide}
\hat{\bw}_{BFGSE}=\widehat{\alpha}^*\left(\dfrac{\bS_n^{-1}\bi_p}{\bi_p^\prime\bS_n^{-1}\bi_p}+\gamma^{-1}\hat{\bQ}_n\sy_n\right)+(1-\widehat{\alpha}^*)\mathbf{b}\,
\end{equation}
with $\widehat{\alpha}^*$ given in Proposition \ref{prop1}. The expression (\ref{bonafide}) is the optimal shrinkage estimator for a given target portfolio $\mathbf{b}$ in  the sense that the shrinkage intensity $\widehat{\alpha}^*$ tends almost surely to its optimal value $\alpha^*$ for $p/n\rightarrow c\in(0, 1)$ as $n\rightarrow\infty.$

The situation is more complex in case $c>1$. Here we can present only oracle estimators for the unknown quantities $R_{GMV}$, $V_{GMV}$ and $s$.

\begin{proposition}\label{prop2}
The consistent estimator for the limiting optimal shrinkage intensity $\alpha^*$ under large dimensional asymptotics $p/n\rightarrow c>1~ \text{as} ~n\rightarrow\infty$ is given by
\begin{equation}\label{a:bf2o}
 \widehat{\alpha}^o = \beta^{-1}\dfrac{(\hat{R}^o_{c}-\hat{R}_b)\left(1+\dfrac{\beta/\gamma}{p/n(p/n-1)}\right)+\beta(\hat{V}_b-\hat{V}^o_{c})+\dfrac{\gamma^{-1}}{p/n(p/n-1)}\hat{s}^o_{c}}{\dfrac{(p/n)^2}{p/n-1}\hat{V}^o_{c}-2\left(\hat{V}^o_{c} + \frac{\gamma^{-1}}{p/n(p/n-1)}(\hat{R}_b-\hat{R}^o_{c})\right)+\dfrac{\gamma^{-2}}{(p/n-1)^3}\left(\hat{s}^o_c+(p/n)^2\right)+\hat{V}_b}\,,
\end{equation}
where $\hat{R}^o_{c}$, $\hat{V}^o_{c}$, $\hat{s}^o_{c}$ are given by
  \begin{eqnarray*}
\hat{R}^o_{c}&=&\hat{R}_{GMV}\\
\hat{V}^o_{c}&=&\frac{1}{p/n(p/n-1)}\hat{V}_{GMV}\\
\hat{s}^o_{c}&=&p/n[(p/n-1)\hat{s}-1]\,,
\end{eqnarray*}
where $\hat{R}_{GMV}$, $\hat{V}_{GMV}$ and $\hat{s}$ are the traditional plug-in estimators based on the generalized inverse $\bS_n^*$ from \eqref{geninverse} and $\hat{R}_b$ and $\hat{V}_b$ are given in (\ref{hRb}) and (\ref{hVb}), respectively.
\end{proposition}


Note that $\widehat{\alpha}^o$  from Proposition \ref{prop2} is not the bona fide estimator for the unknown shrinkage intensity $\alpha^+$, since the matrix $\bS_n^*$ depends on the unknown quantities. Thus, we propose a reasonable approximation using the application of the Moore-Penrose inverse $\bS_n^+$.
As a result, the bona fide estimators of the quantities $R_{GMV}$, $V_{GMV}$ and $s$ in case $c>1$ are approximated by
\begin{eqnarray}\label{RVs}
\hat{R}^+_{c}\approx \dfrac{\sy_n^\prime\bS_n^+\bi_p}{\bi_p^\prime\bS_n^+\bi_p},~~~
\hat{V}^+_{c}\approx \dfrac{1}{p/n(p/n-1)}\dfrac{1}{\bi_p^\prime\bS_n^+\bi_p},~~~
\hat{s}^+_c \approx p/n[(p/n-1)\sy_n^\prime\bQ_n^+\sy_n-1]\,,
\end{eqnarray}
respectively. The application of \eqref{RVs} leads to the bona fide optimal shrinkage estimator of the EU portfolio in case $c>1$ expressed as
\begin{equation}\label{bonafide+}
\hat{\bw}_{BFGSE}^+=\widehat{\alpha}^+\left(\dfrac{\bS_n^{+}\bi_p}{\bi_p^\prime\bS_n^{+}\bi_p}+\gamma^{-1}\hat{\bQ}^+_n\sy_n\right)+(1-\widehat{\alpha}^+)\mathbf{b}_n\,,
\end{equation}
with
\begin{equation}\label{a:bf2}
 \widehat{\alpha}^+ = \beta^{-1}\dfrac{(\hat{R}^+_{c}-\hat{R}_b)\left(1+\dfrac{\beta/\gamma}{p/n(p/n-1)}\right)+\beta(\hat{V}_b-\hat{V}^+_{c})+\dfrac{\gamma^{-1}}{p/n(p/n-1)}\hat{s}^+_{c}}{\dfrac{(p/n)^2}{p/n-1}\hat{V}^+_{c}-2\left(\hat{V}^+_{c} + \frac{\gamma^{-1}}{p/n(p/n-1)}(\hat{R}_b-\hat{R}^+_{c})\right)+\dfrac{\gamma^{-2}}{(p/n-1)^3}\left(\hat{s}^+_c+(p/n)^2\right)+\hat{V}_b}\,,
\end{equation}
where $\hat{R}_b$ and $\hat{V}_b$ are given in (\ref{hRb}) and (\ref{hVb}), respectively; $\bQ_n^+=\bS_n^{+}-\dfrac{\bS_n^{+}\bi\bi^\prime\bS_n^{+}}{\bi^\prime\bS_n^{+}\bi}$ and $\bS_n^{+}$ is the Moore-Penrose pseudo-inverse of the sample covariance matrix $\bS_n$.

\begin{remark}\rm
    It is easy to verify that if $\bSigma_n=\sigma^2\bI_p$ for any $\sigma>0$ the considered approximations in \eqref{RVs} become the exact ones. Next, we investigate the quality of this approximation in general case without imposing restrictions on $\bSigma_n$. This issue was studied for other quantities involving $\bS^*_n$ and $\bS_n^+$ in detail by \cite{bodpar2020}, who compare the limiting spectral distributions of $\bS^*_n$ and $\bS_n^+$ by deriving the limits for their corresponding Stieltjes transforms. It is concluded that the two inverses behave completely different in general. However, when the concentration ratio $c$ approaches $1$, then the limiting spectral distributions of both inverses $\bS^*_n$ and $\bS_n^+$ coincide independently of the structure of $\bSigma_n$. This in turn means that one should expect a good approximation quality when $c$ is not far away from $1$.

   \begin{figure}[h!]
    \centering
    \includegraphics[width=0.7\textwidth]{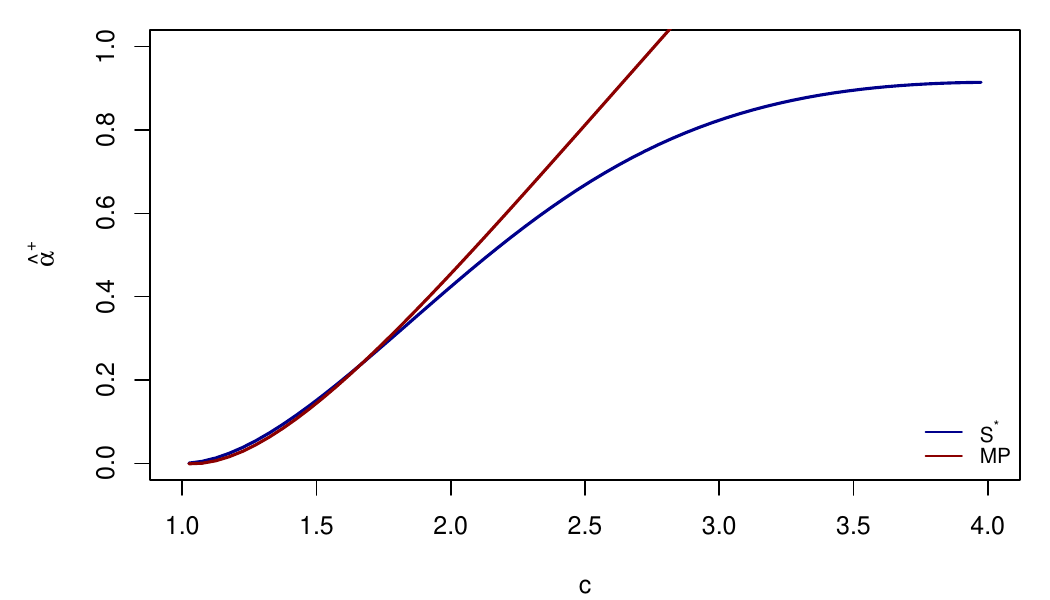}
    \caption{Estimated optimal shrinkage intensities for $\bS^*_n$ and $\bS_n^+$(MP) as function of concentration ratio $c>1$ and dimension $p=300$. }\label{alphacomp}
  \end{figure}

  \begin{figure}[h!]
    \centering
\begin{tabular}{cc}
   \hspace{-2cm} \includegraphics[width=0.55\textwidth]{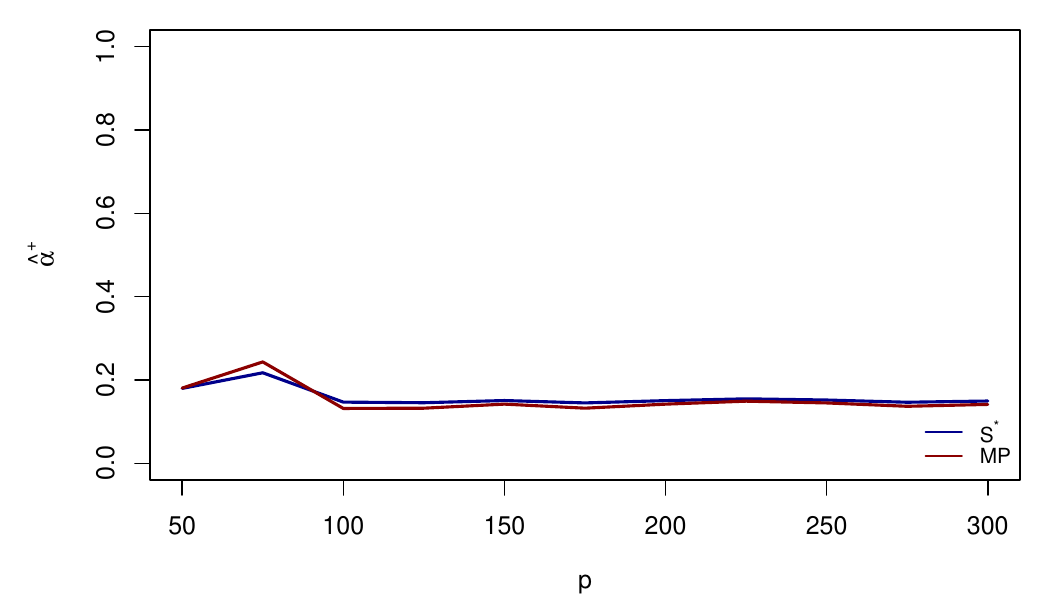} &\hspace{-0.5cm} \includegraphics[width=0.55\textwidth]{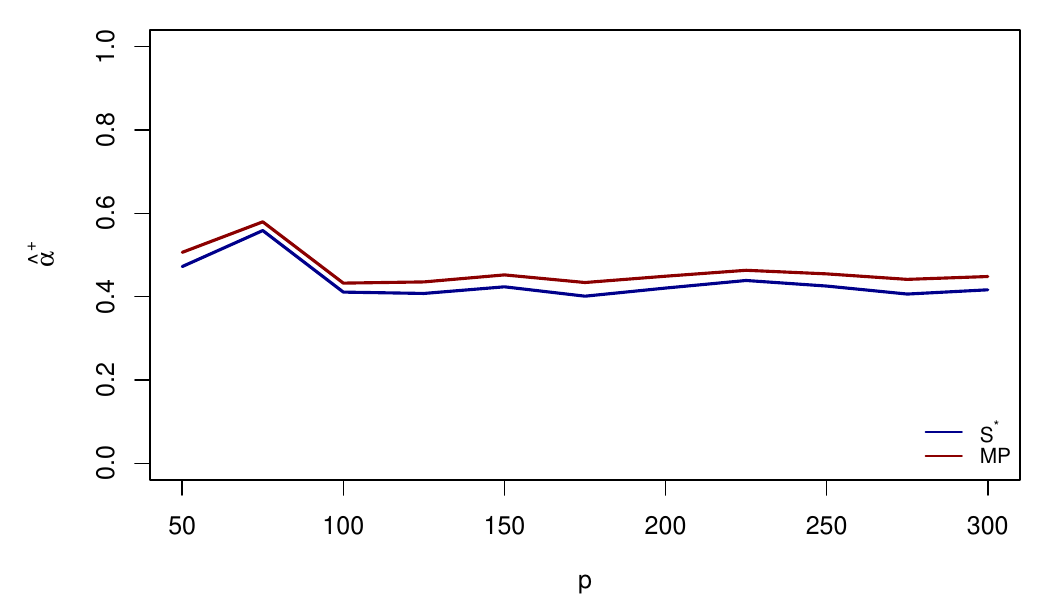}
\end{tabular}
  \caption{Estimated optimal shrinkage intensities for $\bS^*_n$ and $\bS_n^+$(MP) as function of dimension $p$ for $c=1.5$ (left) and $c=2$ (right).}\label{alphacompdim}
  \end{figure}

In Figure \ref{alphacomp} we provide a comparison between the optimal shrinkage intensities computed using different types of generalized inverses, namely Moore-Penrose inverse $\bS_n^+$ and the reflexive inverse $\bS_n^*$ from \eqref{geninverse}. The optimal shrinkage intensity is calculated by \eqref{a:bf2} in the case of $\bS_n^+$ and by using \eqref{a:bf2} where $\hat{R}^+_c$, $\hat{V}^+_c$, and $\hat{s}^+_c$ are replaced by $\hat{R}^o_c$, $\hat{V}^o_c$, and $\hat{s}^o_c$ from Proposition \ref{prop2} in the case of $\bS_n^*$. The design of the simulation study is exactly the same as the one described in Section \ref{simulations}.

We observe that the two optimal shrinkage intensities are quite identical till the breaking point $c=2$. For $c>2$ the Moore-Penrose approximation is not reliable anymore. We observe a similar behaviour for other structures of the covariance matric $\bSigma_n$ and the mean vector $\bm_n$, which indicates that the results are robust and justifies the theoretical findings of \cite{bodpar2020} for the optimal shrinkage intensity given in \eqref{a:bf2}. Figure \ref{alphacompdim} provides further numerical results related to the comparison of the two optimal shrinkage intensities. Here, we set $c=1.5$ and $c=2$ and study the robustness of the results in Figure \ref{alphacomp}  to changes in the dimension $p$ from $50$ to $300$. We see that the shrinkage intensities are very similar uniformly over $p$, independently of the chosen value of $c$.

Summarizing the above findings, we can recommend the application of the Moore-Penrose approximation for $c\le 2$, but there is no guarantee for a good performance for $c>2$. Also, we observe that the Moore-Penrose inverse gets closer to $\bS_n^*$ when the covariance matrix $\bSigma_n$ is sparse. In the empirical study of Section \ref{empirical} we consider the values of the concentration ratio $c$ bounded by $2$. The empirical results are in line with the discussion provided in this remark and, thus, $c\approx 2$ seems to be a breaking point for the approximation indeed. More theoretical treatment of this interesting phenomenon is of an independent research interest and is not within the scope of this paper.
  \end{remark}

\begin{remark}\rm
 Seemingly, we have handled two cases $c<1$ and $c>1$ (for $c \le 2$), but not $c=1$. The case $c=1$ is not easy to manage because the sample covariance matrix is theoretically invertible for $c$ equal or close to one but computationally very unstable. The reason is the smallest eigenvalue of $\bS_n$ which is numerically very close to zero. Indeed, it is well-known that the smallest eigenvalue of $\bS_n$ is of order $(1-\sqrt{p/n})^2$, which converges to zero if $p/n\to1$ and all the estimators explode (see, e.g., \cite{baiyin1993}).

In order to overcome the difficulty in a small neighborhood of $c=1$ one has a few options to proceed:
\begin{description}
\item[Tikhonov (ridge)] One of the possibilities, which has also been used in the simulation and empirical studies of Section \ref{sim-emp}, is the Tikhonov (ridge) approximation of the Moore-Penrose inverse. Indeed, one can show by the eigenvalue decomposition that 
\begin{eqnarray}\label{Tikhonov}
\lim\limits_{\delta\to0^+}(\bS_n+\delta\bI)^{-1}\bS_n(\bS_n+\delta\bI)^{-1}= \bS^+_n
\end{eqnarray}
if $c>1$. For $c<1$ the limit in \eqref{Tikhonov} trivially exists and equals $\bS_n^{-1}$. The advantage of the representation \eqref{Tikhonov} is twofold. First, one has an elegant formula which incorporates both cases $c<1$ and $c>1$, and, secondly, one stabilizes the behaviour of the inverse matrix near singularity, i.e., near $c=1$. The only question arises how to choose $\delta=o(1)$ in practice, but it seems that taking $\delta=1/p$ works well in many applications. Thus, we will employ this adjustment in the empirical study in order to have a balanced and stable estimator when $c$ is close to $1$ from both sides. Although this procedure smoothes out the estimator of the precision matrix, it does not resolve the issue when $c$ is large, i.e., $c>2$.

\item[Moore-Penrose] Yet another option would be to derive the explicit limit of \eqref{alfa1} when the Moore-Penrose inverse matrix $\bS^+_n$ is directly used in \eqref{gse1}-\eqref{hatQ1}. This procedure is highly nontrivial because $\bS_n^+$ depends in a nonlinear way  on the matrix $\bSigma_n$ and this leads to nonlinear integral equations in the high-dimensional setting. The problem becomes even more involved when we consider quadratic and bilinear forms involving the Moore-Penrose inverse. Moreover, the case of centered sample covariance matrix (the sample mean is subtracted) makes the expressions tedious and confusing (see, e.g., \cite{pan2014}). Thus, the fact that the optimal shrinkage intensity depends on the mean vector $\bm_n$ and the covariance matrix $\bSigma_n$ only via the three parameters of the efficient frontier will be lost and no closed-form formulas can be derived for $\bSigma_n\neq \sigma^2\bI$. Nevertheless, we are working on this problem in a separate project and are studying the properties of the pseudo-inverse $\bS_n^+$ in detail, especially the limiting behaviour of its eigenvectors.

\item[Double-shrinkage] A more intuitive option is to apply a double-shrinkage approach, which incorporates both the shrinkage of the estimated portfolio weights and the regularization of the sample covariance matrix. Namely, first we shrink (regularize) the sample covariance matrix and then we shrink the estimated portfolio weights built upon the regularized sample covariance matrix. This could be done by taking  either the matrix $\bS_n+\delta\bI$ for some $\delta>0$ or the one from \eqref{Tikhonov} instead of $\bS_n$. Then the limiting shrinkage intensity $\alpha$ will depend on $\delta$ and one needs to optimize the objective function \eqref{utility} over $\delta$ as well. Unfortunately, no closed-form expression for the optimal shrinkage intensity is available in this case and the theoretical results presented in Section \ref{sec:main} must be changed accordingly. Actually, this procedure would generalize all of the above mentioned ideas (including the ones presented in this paper) and would provide a data-driven regularization parameter $\delta$. This interesting and important topic is left for future research.
\end{description}

\end{remark}

\section{Simulation and empirical studies} \label{sim-emp}

In this section we illustrate the performance and the advantages of the derived results using simulated and real data. Particularly we address the estimation precision of the shrinkage coefficient and compare the bona-fide estimator with the existent approaches.

\subsection{Simulation study}\label{simulations}

For simulation purposes we select the structure of the spectrum of the covariance matrix and of the mean vector to make it consistent with the characteristics of the empirical data. Particularly, for each dimension $p$ we select the expected returns equally spread on the interval -0.3 to 0.3, capturing a typical spectrum of daily returns measured in percent. The covariance matrix has a strong impact on the properties of the shrinkage intensity and for this reason we consider several structures of its spectra. Replicating the properties of empirical data we generate covariance matrices with eigenvalues  satisfying the equation $\lambda_i = 0.1 e^{\delta c \cdot (i-1)/p}$ for $i=1,...,p$ (see, e.g., \cite{HDShOP} for implementation). Thus the smallest eigenvalue is 0.1 and by selecting appropriate values for $c$ we control the largest eigenvalue and thus the condition index of the covariance matrix. Large condition indices imply ill-conditioned covariance matrices, with the eigenvalues very sensitive to changes of the elements. We choose $\delta$ to attain the condition indices of 150, 1000 and 8000. The target portfolio weights are set equal to the weights of the equally weighted portfolio, i.e. $b_i=1/p$  for $i=1,...,p$. The calibration criteria used to determine the optimal shrinkage intensities are selected as defined in Proposition \ref{calibr}.

\begin{figure}[h!]
\begin{tabular}{cc}
  \includegraphics[width=0.5\textwidth]{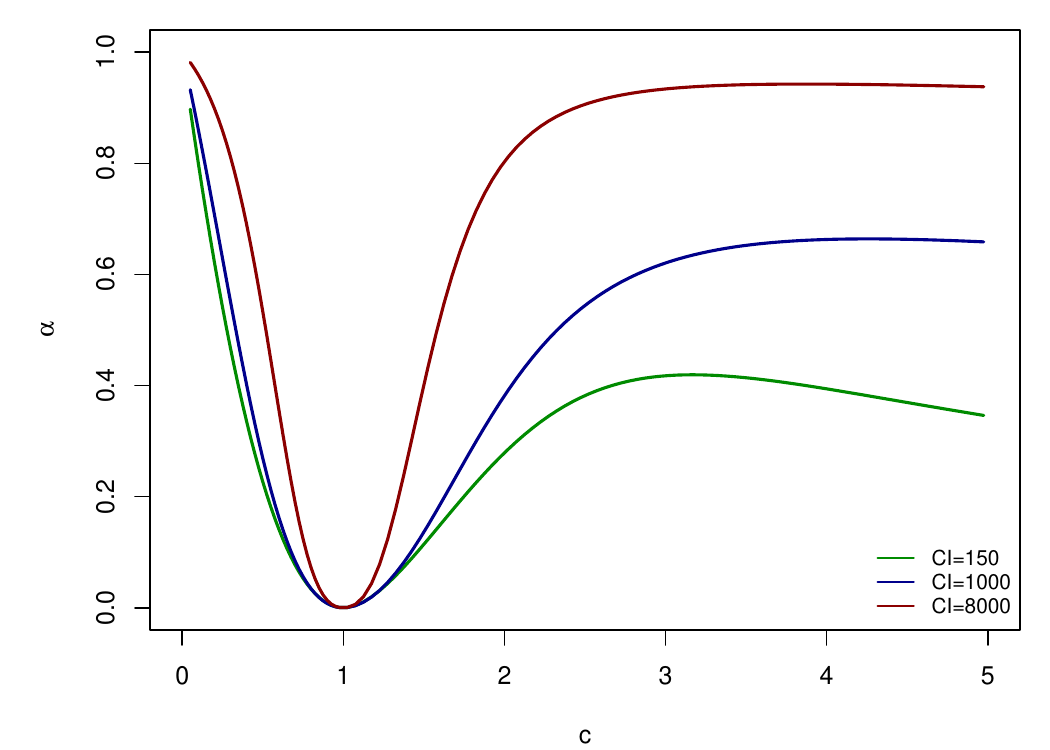}& \includegraphics[width=0.5\textwidth]{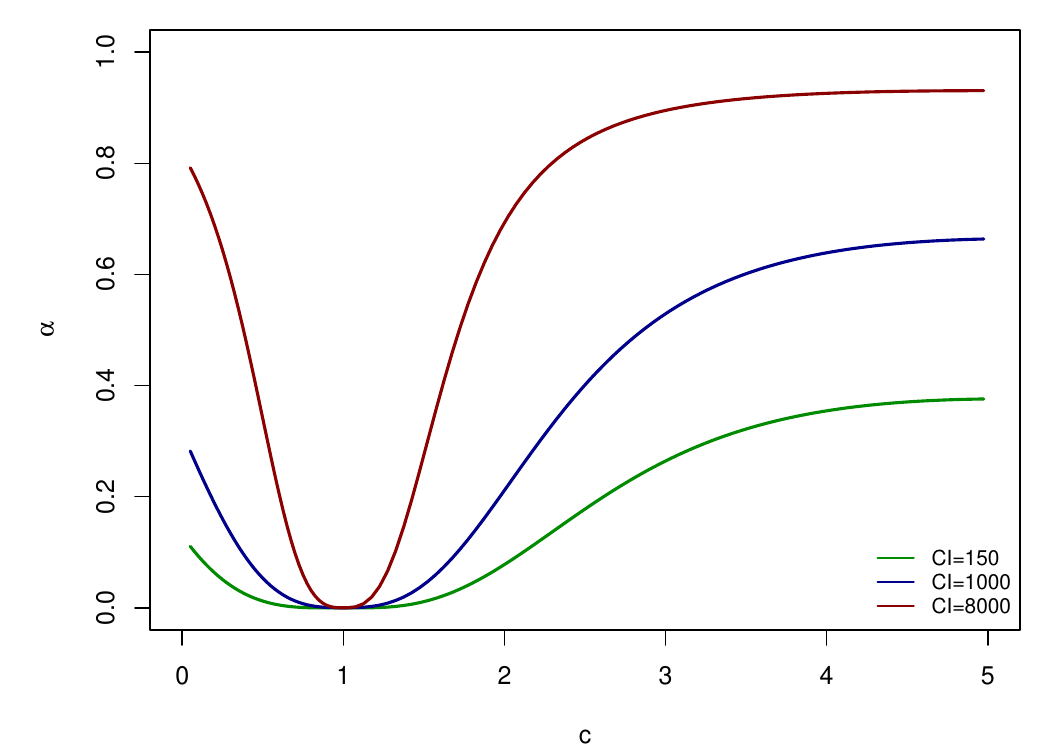}
\end{tabular}
\caption{The asymptotic optimal shrinkage intensity as a function of $c$ for the calibration criteria (i)-(ii) from Proposition \ref{calibr} (left to right).}
\label{fig:alphastar}
\end{figure}

 First, we assess the general behavior of the oracle shrinkage intensities as functions of $c$. The oracle shrinkage intensities are computed using expressions in (\ref{ainfty}) and (\ref{ainfty2}) for the cases $c<1$ and $c>1$, respectively. The parameters are computed using the true mean vector and the true covariance matrix.  The results are illustrated for different condition indices  and different calibration criteria in Figure \ref{fig:alphastar}. We observe that in all cases the shrinkage intensity falls to zero as $c\rightarrow 1_{-}$  and increases with $c$ for $c>1$. Thus if $c$ is small the shrinkage estimator puts higher weight on the traditional estimator of the portfolio weights, due to lower estimation risk. If $c$ tends to 1 the system becomes unstable because of nearly zero eigenvalues. In this case the portfolio weights collapse to the target portfolio weights. With  $c$ further increasing the shrinkage intensity increases too, implying that the pseudo-inverse covariance matrix can be evaluated in a proper way. The fraction of the sample EU portfolio increases with $c$ in this case. It is worth mentioning that at some high level of $c$ the information content in the data becomes less relevant and the shrinkage intensity starts to decrease. Note, however, that even for $p$ much larger than $n$, there is still valuable information in the sample covariance matrix leading to relatively high values of $\alpha^+$.

  Regarding the calibration criteria we observe that if the calibration criteria coincides with the expected quadratic utility (i.e. $\beta=\gamma$), then the limit shrinkage intensities are naturally higher, compared to those minimizing the out-of-sample variance. The variance of portfolio return for the equally weighted portfolio tends to be lower than that of the sample EU portfolio. Thus the shrinkage intensity weights the equally weighted portfolio more heavily. It is important to stress that the latter calibration criterion is more sensitive to the condition index.

\begin{figure}[h!]
\begin{center}
\begin{tabular}{cc}
  \includegraphics[width=0.5\textwidth]{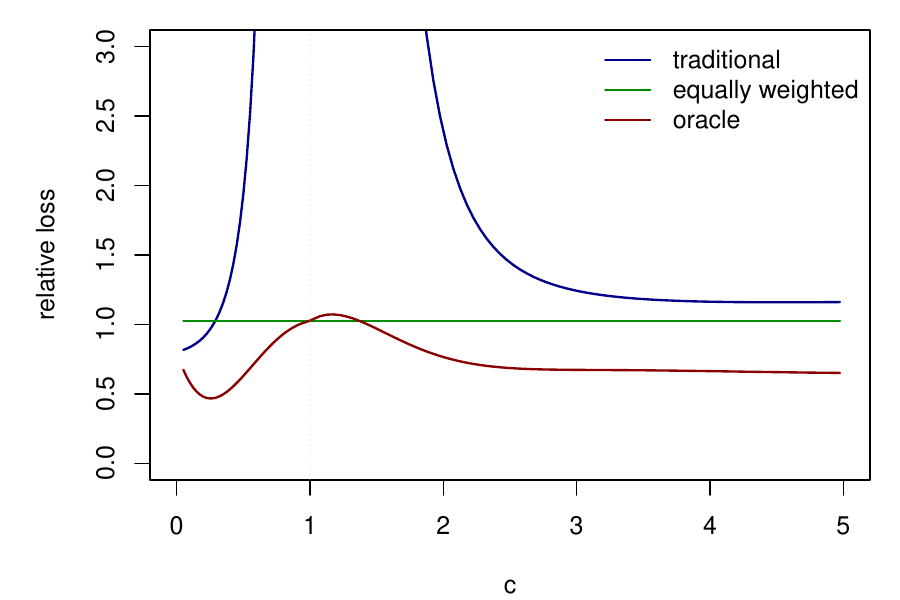} & \includegraphics[width=0.5\textwidth]{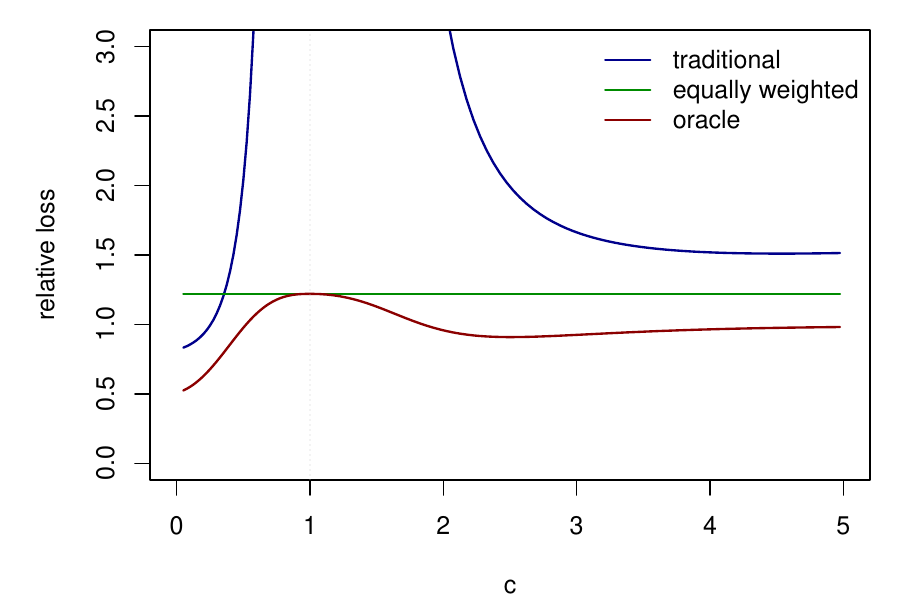}
\end{tabular}
\caption{The relative losses for the portfolios based on the optimal shrinkage estimator, the traditional estimator and the equally weighted portfolio  as a function of $c$ for the calibration criteria (i)-(ii) from Proposition \ref{calibr} (left to right).  The dimension is set to  $p=100$ and  the condition index is set to 1000.}
\label{fig:losses:p}
\end{center}
\end{figure}

In a similar fashion we analyze the relative losses of portfolios based on the traditional estimator and the oracle shrinkage estimator. As a benchmark, we take the equally weighted portfolio which is also the target portfolio of the shrinkage estimator. The relative losses as functions of $c$ for fixed $p=100$ are plotted in Figure \ref{fig:losses:p}. For $c<1$ the losses of the traditional estimator show explosive behavior and are comparable to the shrinkage-based estimators only for very small values of $c$. Thus the traditional estimator is reliable only if the sample size is considerably larger than the portfolio dimension. 
The performance of the shrinkage-based estimator is relatively stable over the whole range of $c$ and it clearly dominates both the traditional and the equally weighted benchmark in almost all of the considered cases. The losses are increasing for $c<1$ and attain the loss of the equally weighted portfolio around $c=1$. This is consistent with the results in Figure \ref{fig:alphastar}. For $c>1$ the losses decrease and remain stable for $c>3$. 

\begin{figure}[h!]
\begin{tabular}{cc}
\includegraphics[width=0.5\textwidth]{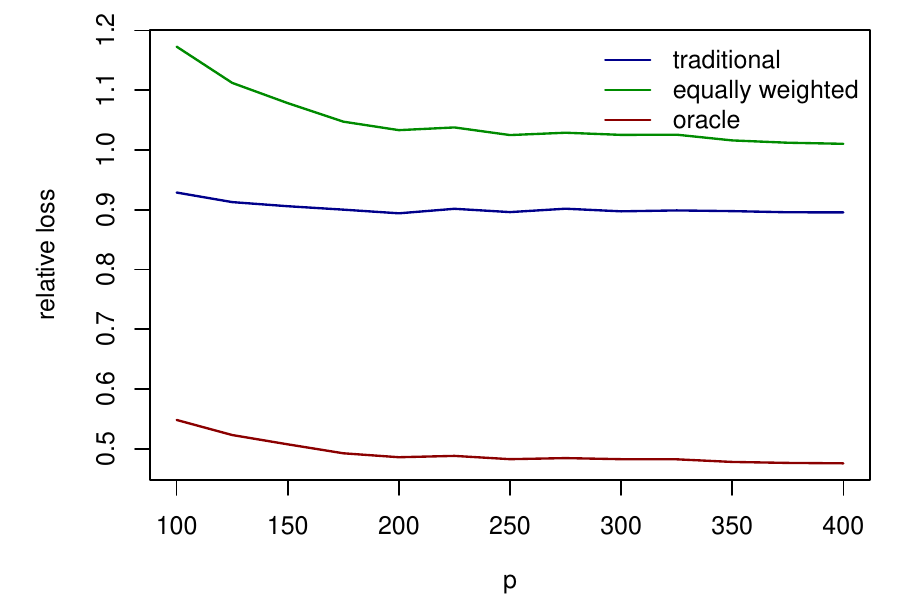} &
\includegraphics[width=0.5\textwidth]{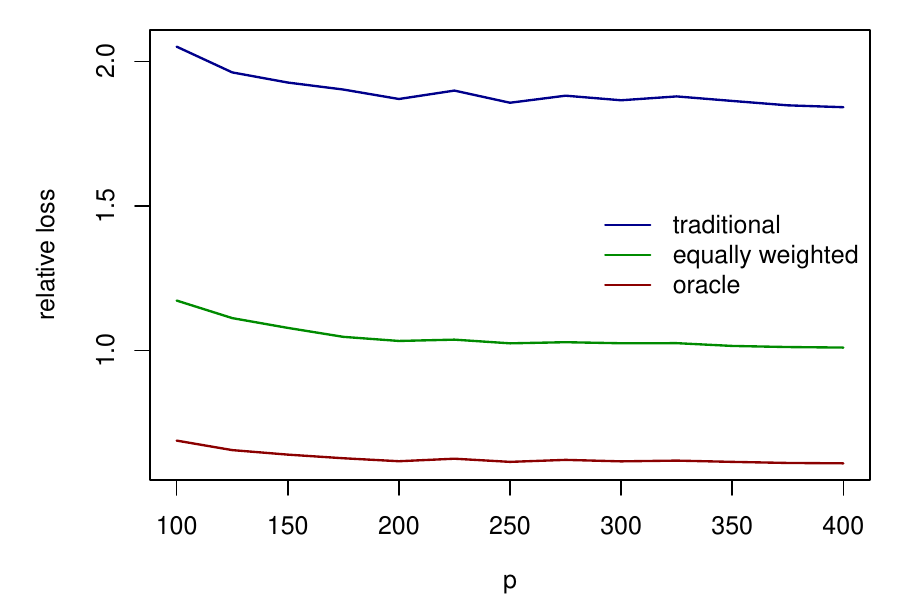} \\
\includegraphics[width=0.5\textwidth]{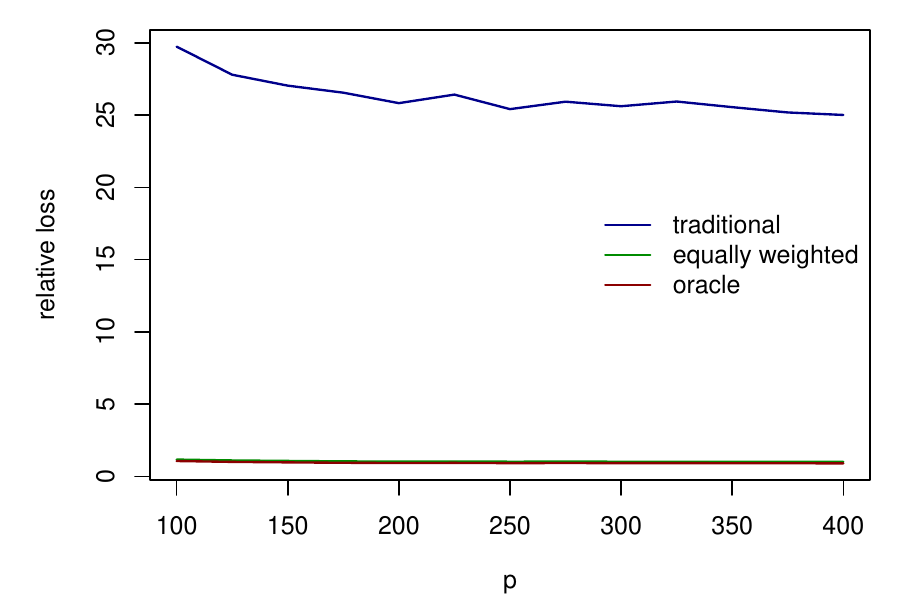} &
\includegraphics[width=0.5\textwidth]{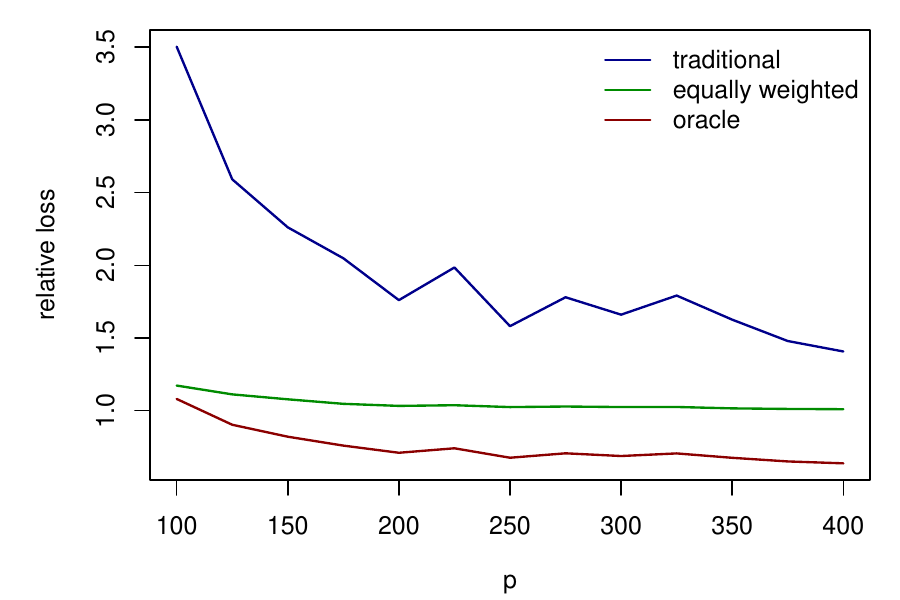}
\end{tabular}\caption{The relative losses for the portfolios based in the optimal shrinkage estimator, the traditional estimator and the equally weighted portfolio as a function of the dimension $p$ for $c=$ 0.2 (top left)), \; 0.5 (top right), \;0.8 (bottom left),  2 (bottom right). The condition index is set to 1000 and the mean-variance calibration criteria is used.}
\label{fig:losses}

\end{figure}

The behavior of losses as functions of the dimension $p$ is illustrated in Figure \ref{fig:losses}. For space reasons we provide here only the results for $\beta=\gamma$, i.e., for the first calibration criterion. The fraction $c$ is set to 0.2 (top left), 0.5 (top right), 0.8 and 2, while the condition index equals 1000. From financial perspective it is important to note that the traditional estimator outperforms the equally weighted portfolio only for small values of $c$ (in the particular setup for $c=0.2$), thus when the classical estimators are stable and robust. This is consistent with the evidence from Figure \ref{fig:losses:p}. For $c=0.8$ the losses of the traditional estimator increase dramatically and they are always considerably larger than the losses of two other considered trading strategies. As before the shrinkage-based estimator clearly beats both the traditional EU portfolio and the benchmark portfolio for all $c $ values. 
Furthermore, the performance of the shrinkage-based portfolio is stable for a wide range of dimensions, particularly for large values of $p$.

\subsection{Empirical study}\label{empirical}

 The  data used in this study cover daily returns on 395 S\&P500 constituents available for the whole period from 01.01.2000 till 23.03.2018. The investor allocates his/her wealth to the constituents with daily reallocation. We address several issues in this empirical study. First, we wish to verify the robustness of the established theoretical results for empirical data. Thus our aim is to go beyond the common practice of considering a single portfolio, but to generate a large set of different portfolios from the universe of the S\&P 500 constituents. Second, we assess the economic performance of the dynamic portfolio strategies stemming from the generalized shrinkage estimator for portfolio weights. Thus we consider several popular performance measures and test the significance in the differences between the alternative strategies. Third, the choice of the target portfolio can clearly have a substantial impact on the empirical results. For this reason, we consider several popular choices of the target. Finally, we wish to assess the dynamics of the estimated shrinkage intensities and relate their behavior to the market conditions. Next, we provide details on the setup of the empirical study.

To address the applicability of the suggested estimator in high dimensions we set $p=300$ which is  larger than a typical portfolio size in the literature. For each parameter constellation we draw randomly 1000 sets of assets from the available constituents of the S\&P500 index. This guarantees a robust assessment of the empirical results. For every set of the assets we build portfolios on each of the last 1000 trading days and compute the corresponding realized returns. Afterwards, we compute the certainty equivalent (CE), Sharpe ratio (SR), Value-at-Risk (VaR) and Expected shortfall (ES) as performance measures for each path of returns and every random portfolio. To avoid potentially skewed inferences due to outliers or asymmetries we report the 10\%-trimmed means and the medians of the CE and SR over the 1000 random portfolios. The VaR and ES are computed as lower empirical quantiles at $5\%$ and $1\%$ significance levels and are averaged over the portfolios either. For simplicity we neglect the transaction costs in the below discussion.

\subsubsection*{Target portfolios}

The target portfolio weights are the key component of the shrinkage estimator. We consider three different targets: the equally weighted portfolio and two modified global minimum-variance portfolios. The equally weighted portfolio arises if we assume that all asset returns have equal expectations, equal variances and equal correlations. The covariance matrix for the first global-minimum variance portfolio assumes different variances, but equal correlations. Thus allows for more heterogeneity compared to the equally weighted portfolio. The single correlation is computed as the average correlation for all asset pairs. The corresponding target is computed by $\mathbf{b}_{ec} =   \hat{\bSigma}_{ec}^{-1} \bi / \bi^\prime \hat{\bSigma}_{ec}^{-1}\bi$, where $\hat{\bSigma}_{ec} = diag\{\hat{\bSigma}\}^{1/2} \bR_{ec} diag\{\hat{\bSigma}\}^{1/2}$ with $diag\{\hat{\bSigma}\}$ being the diagonal of the sample covariance matrix of returns. For the second global minimum-variance portfolio we compute the covariance matrix of returns using the three-factor Fama-French model, i.e.
\[\hat{\bSigma}_{ff} = \hat{\mathbf{B}} \hat{\bSigma}_{f}\hat{\mathbf{B}}^\prime + diag\{\hat{\sigma}_{\eps_i}^2\}_{i=1,...,p},\]
where $\hat{\mathbf{B}}$ is the matrix of estimated parameters, $\hat{\bSigma}_{f}$ is the covariance matrix of the factors and $diag\{\hat{\sigma}_{\eps,i}^2\}_{i=1,...,p}$ is the diagonal matrix of residual variances.
The resulting target vector of weights is defined as $\mathbf{b}_{ff} = \hat{\bSigma}_{ff}^{-1} \bi / \bi^\prime \hat{\bSigma}_{ff}^{-1}\bi$. The latter two portfolios reduce the variation of portfolios by looking at the variance but not the mean of the underlying assets. Note, that $\mathbf{b}_{ec}$ and $\mathbf{b}_{ff}$ are stochastic by construction, since they are computed  using sample characteristics of the asset returns. Thus the theoretical results in these cases are valid only conditionally on the target vector.

\subsubsection*{Benchmark models}

To guarantee a fair assessment of the suggested estimator we consider two popular approaches as benchmarks. The first approach is based on the non-linear shrinkage estimator of the covariance matrix suggested by \cite{lw12, ledoit2017nonlinear}. The estimator relies on the spectral decomposition of the sample covariance matrix $\hat{\bSigma} = \bU \bD \bU^\prime$, but replaces the original eigenvalues by eigenvalues $\bD^*$ that minimize the Frobenius norm $\bD^*= \text{argmin}_{\bD} ||\bSigma - \bU \bD \bU^\prime||$. The solution can be approximated using a generalized version of the Mar\v{c}enko-Pastur equation. First, we consider a numerical implementation of this method proposed in \cite{lw2017} (called \textit{LWQuEST}). The direct numerical computation of the eigenvalues appears to be very demanding. Recently, \cite{lw2020} suggest an analytic expression of the nonlinear shrinkage estimator that uses a nonparametric estimator of the spectral density (called \textit{LWAnalytic}). In order to determine the optimal portfolio weights we use these two approaches as a plug-in estimator of the covariance matrix. Note that in contrary to our approach, these methods shrink a parameter of the distribution of asset returns and not the portfolio weights, which are the key object of interest in asset allocation problems. The second benchmark is an extension of the three-fund portfolio of \cite{kz07}. The optimal portfolio is a linear combination of the target portfolio, the sample global minimum-variance portfolio and the portfolio that maximizes the Sharpe ratio.\footnote{We thank Raymond Kan for providing these results in personal communication.}.

 The weights of this linear combination are determined by maximizing the expected out-of-sample utility. Note that \cite{kz07} derive the optimal portfolios assuming a finite portfolio with $p<n$. The extension of these results to the case $p/n\rightarrow c$ goes beyond the scope of this paper. For this reason we apply this estimator only for the case $p<n$. For both benchmarks we estimate the expected returns by the historical mean returns, as it is frequently done in practice.

\subsubsection*{Empirical results}

The results of the empirical study are summarized in Tables \ref{tab:emp:meanv} and \ref{tab:emp:minv} containing the results for mean-variance ($\beta=\gamma$) and minimum variance ($\beta=\infty$) calibration criteria, respectively. The top blocks of the table provides results for $c=0.2$, while the following blocks correspond to $c=$0.5, 0.8, and 2. At the beginning of each block we summarize the performance measures for the traditional estimator and the estimator based on the nonlinear-shrinkage of \cite{lw12}. These estimators do not depend on the target portfolio weights. Further we provide the results for the strategies involving the target, i.e. the suggested bona-fide shrinkage technique, the target portfolio itself and the extension of the estimator of \cite{kz07}. This is done for equally weighted portfolio, equal correlation portfolio and Fama-French portfolio as targets. Furthermore, we include results for the traditional portfolio and for the bona-fide shrinkage portfolio, where the sample covariance matrix is replaced by its regularized version based on the Tikhonov (ridge) approximation \eqref{Tikhonov} with $\delta=1/p$. The corresponding strategies are called \textit{trad ridge} and \textit{bona fide ridge}, respectively.

\begin{table}
\renewcommand{\arraystretch}{1}\scriptsize
\begin{center}
\adjustbox{scale=0.95}{%
\begin{tabular}{l|ccccccccc}
&& \multicolumn{2}{c}{CE} & \multicolumn{2}{c}{SR} &\multicolumn{2}{c}{VaR}&\multicolumn{2}{c}{ES} \\\hline
&& average  & median & average  & median  & $\alpha=0.05$ & $\alpha=0.01$ & $\alpha=0.05$ &$\alpha=0.01$ \\\hline
\multirow{16}{*}{\rotatebox[origin=c]{90}{$\bf c=0.2$}}
&trad & -0.84840 & -0.84795 & 0.10033 & 0.10041 & -0.88181 & -1.43362 & -1.27802 & -2.01774\\
&trad ridge  & -0.85096 & -0.85028 & 0.10066 & 0.1006 & -0.88193 & -1.43891 & -1.28130 & -2.02313\\
& LWQuEST  & -0.96539 & -0.96509 & \textbf{0.10310} & \textbf{0.10288} & -0.93884 & -1.52905 & -1.35071 & -2.10094\\
&LWAnalytical & -0.79142 & -0.79179 & 0.10141 & 0.10151 & -0.8544 & -1.37885 & -1.24301 & -1.97136\\
\cline{2-10}
& \multicolumn{8}{c}{equally weighted target}\\\cline{2-10}
&bona fide & -0.80701 & -0.80578 & 0.10125 & 0.10137 & -0.85637 & -1.41141 & -1.25832 & -2.00392\\
&bona fide ridge &  -0.80303 & -0.80305 & 0.10124 & 0.10136 & -0.85485 & -1.39997 & -1.25366 & -1.99908\\
&target & -1.47973 & -1.47978 & 0.05178 & 0.05183 & -1.35819 & -2.17993 & -1.90944 & -2.91308\\
&KZ & \textbf{-0.73868} & \textbf{-0.73831} & 0.10204 &0.10195 &\textbf{  -0.81524} & \textbf{-1.33615} & \textbf{-1.20824} & \textbf{-1.94859}\\\cline{2-10}
& \multicolumn{8}{c}{equal correlation target}\\\cline{2-10}
&bona fide & -0.86149 & -0.86029 & 0.09785 & 0.09803 & -0.90268 & -1.48278 & -1.30966 & -2.06512\\
&bona fide ridge & -0.86011 & -0.85989 & 0.09773 & 0.09788 & -0.90141 & -1.48209 & -1.30729 & -2.06240\\
&target & -2.13972 & -2.14348 & 0.05303 & 0.05291 & -1.46652 & -2.71956 & -2.21071 & -3.41674\\
&KZ & -0.93578 & -0.93437 & 0.09170 & 0.09159 & -0.94118 & -1.63936 & -1.40500 & -2.14364\\\cline{2-10}
& \multicolumn{8}{c}{Fama-French target}\\\cline{2-10}
&bona fide & -0.85595 & -0.85551 & 0.09504 & 0.09433 & -0.89107 & -1.45683 & -1.29601 & -2.04557\\
&bona fide ridge & -0.8524 & -0.85234 & 0.09552 & 0.09556 & -0.88859 & -1.45891 & -1.29374 & -2.04100\\
&target & -2.27656 & -2.27951 & 0.03522 & 0.03529 & -1.44902 & -2.83263 & -2.21467 & -3.40873\\
&KZ & -0.89766 & -0.89748 & 0.08611 & 0.08562 & -0.89600 & -1.58344 & -1.34739 & -2.09784\\\hline\hline
\multirow{16}{*}{\rotatebox[origin=c]{90}{$\bf c=0.5$}}
&trad & -1.56794 & -1.56702 & 0.04112 & 0.04158 & -1.25759 & -1.97710 & -1.74511 & -2.60700\\
&trad ridge & -1.47446 & -1.47309 & 0.04468 & 0.04461 & -1.2158 & -1.90656 & -1.68947 & -2.53159\\
&LWQuEST & -3.28708 & -3.28730 & 0.02557 & 0.02559 & -1.85906 & -2.93786 & -2.56534 & -3.68434 \\
&LWAnalytical & \textbf{-0.86718} & \textbf{-0.86687} & \textbf{0.06531} & \textbf{0.06548} & \textbf{-0.89736} & -1.50835 & \textbf{-1.31493} & \textbf{-2.09438}\\\cline{2-10}
& \multicolumn{8}{c}{equally weighted target}\\\cline{2-10}
&bona fide & -1.02898 & -1.02829 & 0.05067 & 0.05098 & -1.00189 & -1.67912 & -1.458300 & -2.30347\\
&bona fide ridge & -0.99158 & -0.9911 & 0.05342 & 0.05347 & -0.98209 & -1.65375 & -1.43523 & -2.27809\\
&target & -1.47854 & -1.47831 & 0.05167 & 0.05170 & -1.35780 & -2.18133 & -1.90904 & -2.91312\\
&KZ & -0.93245 & -0.93159 & 0.05903 & 0.05973 & -0.95118 & \textbf{-1.50578} & -1.34866 & -2.10714\\\cline{2-10}
& \multicolumn{8}{c}{equal correlation target}\\\cline{2-10}
&bona fide & -1.18004 & -1.17898 & 0.05198 & 0.05215 & -1.10187 & -1.82427 & -1.59237 & -2.44481\\
&bona fide ridge & -1.15576 & -1.15526 & 0.05312 & 0.05308 & -1.09212 & -1.81848 & -1.58300 & -2.43487\\
&target & -1.81699 & -1.81663 & 0.04036 & 0.04029 & -1.40816 & -2.42265 & -2.10357 & -3.33301\\
&KZ & -1.01288 & -1.01141 & 0.05760 & 0.05770 & -1.01531 & -1.57685 & -1.41783 & -2.14790\\\cline{2-10}
& \multicolumn{8}{c}{Fama-French target}\\\cline{2-10}
&bona fide & -1.22574 & -1.22594 & 0.04963 & 0.04973 & -1.11090 & -1.75940 & -1.57373 & -2.42489\\
&bona fide ridge & -1.19113 & -1.19039 & 0.05083 & 0.05064 & -1.09567 & -1.73454 & -1.55387 & -2.40117\\
&target & -1.93351 & -1.9336 & 0.04278 & 0.04265 & -1.40284 & -2.48907 & -1.99845 & -2.92676\\
&KZ & -1.01496 & -1.01416 & 0.05761 & 0.05782 & -1.00263 & -1.55246 & -1.40187 & -2.16018\\\hline\hline
\multirow{16}{*}{\rotatebox[origin=c]{90}{$\bf c=0.8$}}
&trad & -11.10781 & -11.07733 & 0.04031 & 0.03992 & -3.35917 & -5.40548 & -4.66218 & -6.70938\\
&trad ridge & -5.60256 & -5.59981 & 0.04942 & 0.04921 & -2.40138 & -3.8217 & -3.29866 & -4.71741\\
&LWQuEST & -18.80892 & -18.81717 & 0.03520 & 0.03521 & -4.51683 & -6.92530 & -6.05810 & -8.35046\\
&LWAnalytical & \textbf{-0.98453} & \textbf{-0.98438} & \textbf{0.08988} & \textbf{0.08999} & \textbf{-0.98071} & \textbf{-1.61591} & \textbf{-1.42389} & \textbf{-2.24346}\\\cline{2-10}
& \multicolumn{8}{c}{equally weighted target}\\\cline{2-10}
&bona fide & { -1.40577} & {-1.40584} & 0.05118 & 0.05134 & {-1.29895} & {-2.07527} & {1.83362} & {-2.79312}\\
&bona fide ridge &  -1.31635 & -1.31596 & {0.05238} & {0.05235} &-1.28390 & -2.02482 & -1.79171 & -2.74021\\
&target & -1.48023 & -1.47918 & 0.05170 & 0.05171 & -1.35865 & -2.18186 & -1.91045 & -2.9148\\
&KZ & -2.03472 & -2.03185 & 0.06359 & 0.06288 & -1.40796 & -2.25191 & -1.97686 & -2.94533\\\cline{2-10}
& \multicolumn{8}{c}{equal correlation target}\\\cline{2-10}
&bona fide & {-1.60399} & {-1.60241} & 0.04978 & 0.04984 & {-1.29237} & -2.39802 & -1.98752 & -3.17541\\
&bona fide ridge & -1.55829 & -1.55768 & 0.05013 & 0.05005 & -1.26800 & -2.35739 & -1.96495 & -3.16411\\
&target & -1.68942 & -1.68949 & 0.04453 & 0.04452 & -1.3124 & -2.45543 & -2.05351 & -3.35801\\
&KZ & -1.96623 & -1.96299 & 0.06550 & 0.06522 & -1.37667 & -2.20379 & -1.93809 & -2.91761\\\cline{2-10}
& \multicolumn{8}{c}{Fama-French target}\\\cline{2-10}
&bona fide & {-1.79593} & {-1.79528} & 0.04740 & 0.04760 & -1.39084 & {-2.26699} & {-1.94795} & {-2.84830}\\
&bona fide ridge & -1.69800 & -1.6972 & 0.04796 & 0.04793 & -1.34454 & -2.23851 & -1.90270 & -2.78836\\
&target & -1.83876 & -1.83798 & 0.04355 & 0.04365 & -1.37985 & -2.41411 & -1.98305 & -2.94513\\
&KZ & -2.06935 & -2.07129 & 0.06506 & 0.06515 & -1.41329 & -2.27824 & -1.98427 & -2.95639\\\hline\hline
\multirow{13}{*}{\rotatebox[origin=c]{90}{$\bf c=2.0$}}
&trad & -2.12939 & -2.13080 & 0.06414 & 0.06427 & -1.44683 & -2.38186 & -2.07522 & -3.18270\\
&trad ridge &  -2.05179 & -2.05106 & 0.06509 & 0.06514 & -1.42082 & -2.34175 & -2.03791 & -3.12989\\
&LWQuEST & -59.03298 & -58.99561 & 0.00629 & 0.00618 & -7.9549 & -12.68863 & -10.94189 & -15.31889\\
&LWAnalytical & -1.50503 & -1.50439 & 0.06355 & 0.06375 & -1.26581 & -1.95316 & -1.7726 & -2.69539\\\cline{2-10}
& \multicolumn{8}{c}{equally weighted target}\\\cline{2-10}
&bona fide & \textbf{-1.22550} & \textbf{-1.22474} & 0.05904 & 0.05907 & -1.18681 & -1.99190 & -1.72914 & -2.75302\\
&bona fide ridge & -1.30872 & -1.30891 & 0.06790 & 0.06782 & -1.14901 & -1.94775 & -1.70477 & -2.73827\\
&target & -1.47869 & -1.47779 & 0.05175 & 0.05176 & -1.35801 & -2.17908 & -1.90898 & -2.91238\\
\cline{2-10}
& \multicolumn{8}{c}{equal correlation target}\\\cline{2-10}
&bona fide & -1.22592 & -1.22761 & 0.07282 & 0.07288 & \textbf{-1.09681} & -2.09756 & -1.74639 & -2.91502\\
&bona fide ridge & -1.24075 & -1.23914 & \textbf{0.07872} & \textbf{0.07873} & -1.10084 & \textbf{-1.93814} & \textbf{-1.67112} & {-2.70638}\\
&target & -1.32850 & -1.32805 & 0.07020 & 0.07020 & -1.13496 & -2.16850 & -1.83174 & -3.08545\\
\cline{2-10}
& \multicolumn{8}{c}{Fama-French target}\\\cline{2-10}
&bona fide & -1.44490 & -1.44438 & 0.05890 & 0.05867 & -1.22790 & -2.09709 & -1.75667 &\textbf{-2.65269}\\
&bona fide ridge & -1.42024 & -1.41913 & 0.06947 & 0.06922 & -1.19613 & -1.96825 & -1.72476 & -2.66945\\
&target & -1.62084 & -1.62093 & 0.05377 & 0.05360 & -1.29152 & -2.34075 & -1.86625 & -2.83814\\[-0.4cm]
\end{tabular}
}
\end{center}
\caption{\scriptsize  Performance of traditional, bona-fide, the benchmark portfolios (LWQuEST - \cite{lw2017}, LWAnalytical - \cite{lw2020}, KZ - \cite{kz07}) and the target portfolios for the {\it mean-variance calibration} criteria. The performance measures are averaged over 1000 random portfolios of size 300. The trading period consists of 1000 days preceding 23.03.2018 and the risk aversion is set to 5. The average values are based on trimmed mean with 10\% of extreme values being dropped. The best strategies for every criteria and every values of $c$ are highlighted in bold. }\label{tab:emp:meanv}
\end{table}

\begin{table}[p!]
\renewcommand{\arraystretch}{1}\scriptsize
\begin{center}
\adjustbox{scale=0.95}{%
\begin{tabular}{l|ccccccccc}
&& \multicolumn{2}{c}{CE} & \multicolumn{2}{c}{SR} &\multicolumn{2}{c}{VaR}&\multicolumn{2}{c}{ES} \\\hline
&& average  & median & average  & median  & $\alpha=0.05$ & $\alpha=0.01$ & $\alpha=0.05$ &$\alpha=0.01$ \\\hline
\multirow{16}{*}{\rotatebox[origin=c]{90}{$\bf c=0.2$}}
&trad & -0.85309 & -0.85231 & 0.10035 & 0.10057 & -0.88250 & -1.44210 & -1.28383 & -2.02584\\
&trad ridge & -0.84923 & -0.84889 & 0.09923 & 0.09909 & -0.88174 & -1.43823 & -1.27883 & -2.01864\\
&LWQuEST & -0.96554 & -0.96474 & \textbf{0.10343} & \textbf{0.10359} & -0.93948 & -1.52539 & -1.35090 & -2.10343\\
&LWAnalytical & -0.79114 & -0.79048 & 0.1014 & 0.10161 & -0.85381 & -1.38047 & -1.24188 & -1.96952\\\cline{2-10}
& \multicolumn{8}{c}{equally weighted target}\\\cline{2-10}
&bona fide & -0.81308 & -0.81284 & 0.10123 & 0.10150 & -0.85908 & -1.41182 & -1.26178 & -2.0085\\
&bona fide ridge &  -0.80998 & -0.81034 & 0.10012 & 0.10001 & -0.85828 & -1.40687 & -1.25735 & -2.00243\\
&target & -1.47805 & -1.47784 & 0.05169 & 0.05168 & -1.35704 & -2.17938 & -1.90829 & -2.91097\\
&KZ & \textbf{-0.73828} & \textbf{-0.73822} & {0.10187} & {0.10195} & \textbf{-0.81400} & \textbf{-1.33608} & \textbf{-1.20828} & \textbf{-1.95301}\\\cline{2-10}
& \multicolumn{8}{c}{equal correlation target}\\\cline{2-10}
&bona fide & -0.85773 & -0.85717 & 0.09806 & 0.09785 & -0.89740 & -1.47224 & -1.30236 & -2.05662\\
&bona fide ridge & -0.85262 & -0.85171 & 0.09803 & 0.0981 & -0.89563 & -1.46898 & -1.29994 & -2.05272\\
&target & -2.13255 & -2.12898 & 0.05319 & 0.05312 & -1.46343 & -2.71591 & -2.20532 & -3.40780\\
&KZ & -0.93731 & -0.93588 & 0.09141 & 0.09120 & -0.94052 & -1.63821 & -1.40403 & -2.14399\\\cline{2-10}
& \multicolumn{8}{c}{Fama-French target}\\\cline{2-10}
&bona fide & { -0.85149} & { -0.85077} & 0.09663 & 0.09677 & -0.88750 & -1.45256 & -1.29064 & -2.03126\\
&bona fide ridge & -0.84949 & -0.84947 & 0.09629 & 0.09625 & -0.88760 & -1.44700 & -1.28814 & -2.0294\\
&target & -2.27063 & -2.26895 & 0.03538 & 0.03543 & -1.44480 & -2.82390 & -2.21059 & -3.40190\\
&KZ & -0.89663 & -0.89478 & 0.08661 & 0.08659 & -0.89534 & -1.58062 & -1.34627 & -2.09320\\\hline \hline
\multirow{16}{*}{\rotatebox[origin=c]{90}{$\bf c=0.5$}}
&trad & -1.56971 & -1.57001 & 0.04173 & 0.04161 & -1.26072 & -1.97665 & -1.74582 & -2.60683\\
&trad ridge &-1.47585 & -1.4756 & 0.04465 & 0.0444 & -1.21598 & -1.90785 & -1.69053 & -2.54353\\
&LWQuEST & -3.28833 & -3.29027 & 0.02581 & 0.02608 & -1.86558 & -2.93777 & -2.56735 & -3.67757\\
&LWAnalytical & \textbf{-0.86699} & \textbf{-0.86626} & \textbf{0.0654} & \textbf{0.06526} & \textbf{-0.89816} & -1.50838 & \textbf{-1.31504} & \textbf{-2.0954}\\\cline{2-10}
& \multicolumn{8}{c}{equally weighted target}\\\cline{2-10}
&bona fide & -1.0532 & -1.05298 & 0.05066 & 0.05048 & -1.01036 & -1.69687 & -1.46966 & -2.30812\\
&bona fide ridge &  -1.14637 & -1.14571 & 0.05275 & 0.05322 & -1.08067 & -1.78718 & -1.56281 & -2.41015\\
&target & -1.48082 & -1.48066 & 0.05171 & 0.05172 & -1.35847 & -2.18138 & -1.91024 & -2.91449\\
&KZ & -0.93323 & -0.93338 & 0.05939 & 0.05941 & -0.95222 & \textbf{-1.50541} & -1.34938 & -2.10436\\\cline{2-10}
& \multicolumn{8}{c}{equal correlation target}\\\cline{2-10}
&bona fide & -1.17591 & -1.17430 & 0.05196 & 0.05184 & -1.09485 & -1.80109 & -1.57693 & -2.42242\\
&bona fide ridge &-1.14637 & -1.14571 & 0.05275 & 0.05322 & -1.08067 & -1.78718 & -1.56281 & -2.41015\\
&target & -1.80886 & -1.80920 & 0.04010 & 0.04001 & -1.40596 & -2.41696 & -2.09826 & -3.32139\\
&KZ & -1.01397 & -1.01394 & 0.05762 & 0.05790 & -1.01524 & -1.58062 & -1.41661 & -2.14494\\\cline{2-10}
& \multicolumn{8}{c}{Fama-French target}\\\cline{2-10}
&bona fide & -1.22924 & -1.22845 & 0.04980 & 0.05017 & -1.11019 & -1.75950 & -1.56925 & -2.41742\\
&bona fide ridge & -1.18967 & -1.18982 & 0.05057 & 0.05064 & -1.09012 & -1.72857 & -1.54528 & -2.39013\\
&target & -1.93623 & -1.93641 & 0.04259 & 0.04257 & -1.40372 & -2.49587 & -2.00108 & -2.93538\\
&KZ & -1.01377 & -1.01328 & 0.05759 & 0.05765 & -1.00204 & -1.55270 & -1.40064 & -2.15907\\\hline\hline
\multirow{16}{*}{\rotatebox[origin=c]{90}{$\bf c=0.8$}}
&trad & -11.1587 & -11.12847 & 0.03821 & 0.03725 & -3.3777 & -5.43327 & -4.67808 & -6.73814\\
&trad ridge & -5.59307 & -5.5935 & 0.05126 & 0.05138 & -2.39789 & -3.80057 & -3.28431 & -4.69414\\
&LWQuEST & -18.84221 & -18.84773 & 0.03577 & 0.03624 & -4.50978 & -6.96099 & -6.06441 & -8.39818\\
&LWAnalytical & \textbf{-0.98496} & \textbf{-0.98436} & \textbf{0.08942} & \textbf{0.08955} & \textbf{-0.98112} & \textbf{-1.61597} &\textbf{ -1.42344} & \textbf{-2.2402}\\\cline{2-10}
& \multicolumn{8}{c}{equally weighted target}\\\cline{2-10}
&bona fide & { -1.43449} & { -1.43529} & 0.05084 & 0.0503 & { -1.30108} & { -2.08832} & { -1.83907} & { -2.79121}\\
&bona fide ridge & -1.30591 & -1.30557 & 0.05487 & 0.05487 & -1.26708 & -1.99403 & -1.77592 & -2.71366\\
&target & -1.47821 & -1.47811 & 0.05168 & 0.05169 & -1.35667 & -2.18041 & -1.90876 & -2.91209\\
&KZ & -2.03415 & -2.03269 & 0.06235 & 0.06256 & -1.41138 & -2.25474 & -1.97862 & -2.94248\\\cline{2-10}
& \multicolumn{8}{c}{equal correlation target}\\\cline{2-10}
&bona fide & { -1.60982} & { -1.60906} & 0.05073 & 0.05031 & { -1.29875} & -2.39866 & -1.98344 & -3.15048\\
&bona fide ridge &  -1.53174 & -1.52924 & 0.05174 & 0.05171 & -1.26259 & -2.35322 & -1.94643 & -3.11905\\
&target & -1.68850 & -1.68914 & 0.04481 & 0.04491 & -1.31139 & -2.45712 & -2.05434 & -3.36065\\
&KZ & -1.96089 & -1.95878 & 0.06499 & 0.06450 & -1.37460 & -2.19955 & -1.93337 & -2.90941\\\cline{2-10}
& \multicolumn{8}{c}{Fama-French target}\\\cline{2-10}
&bona fide & { -1.82002} & { -1.82289} & 0.04739 & 0.04775 & -1.40332 & { -2.25988} & { -1.95652} & { -2.86186}\\
&bona fide ridge & -1.6811 & -1.67987 & 0.04983 & 0.04971 & -1.34377 & -2.20183 & -1.89093 & -2.77691\\
&target & -1.83460 & -1.83658 & 0.04329 & 0.0432 & -1.37616 & -2.40806 & -1.98023 & -2.93814\\
&KZ & -2.07064 & -2.06718 & 0.06476 & 0.06399 & -1.41264 & -2.27641 & -1.98342 & -2.96439\\\hline\hline
\multirow{13}{*}{\rotatebox[origin=c]{90}{$\bf c=2.0$}}
&trad & -2.12727 & -2.12819 & 0.06405 & 0.0641 & -1.44373 & -2.38377 & -2.07218 & -3.17711\\
&trad ridge &-2.06068 & -2.05893 & 0.06518 & 0.06514 & -1.4236 & -2.34051 & -2.04197 & -3.13703\\
&LWQuEST & -58.92062 & -58.87816 & 0.00574 & 0.00587 & -7.96587 & -12.64654 & -10.9261 & -15.29078\\
&LWAnalytical & -1.50775 & -1.50743 & 0.064 & 0.06383 & -1.26662 & -1.9523 & -1.77553 & -2.70126\\\cline{2-10}
& \multicolumn{8}{c}{equally weighted target}\\\cline{2-10}
&bona fide & { -1.45575} & { -1.45568} & 0.05220 & 0.05226 & { -1.34997} & { -2.15551} & { -1.89422} & -2.89468\\
&bona fide ridge & -1.2937 & -1.29318 & 0.06757 & 0.06793 & -1.14349 & -1.94082 & -1.69787 & -2.73005\\
&target & -1.4788 & -1.47899 & 0.05174 & 0.05177 & -1.35821 & -2.17967 & -1.90906 & -2.91205\\
\cline{2-10}
& \multicolumn{8}{c}{equal correlation target}\\\cline{2-10}
&bona fide & { -1.31924} & { -1.32062} & { 0.07073} & { 0.07077} & { -1.13066} & { -2.15773} & { -1.82514} & { -3.07313}\\
&bona fide ridge & \textbf{-1.22331} & \textbf{-1.22256} & \textbf{0.07853} & \textbf{0.07837} & \textbf{-1.09087} & \textbf{-1.93172} & \textbf{-1.66227} & -2.69818\\
&target & -1.32790 & -1.32922 & 0.07046 & 0.07044 & -1.13449 & -2.16647 & -1.83165 & -3.08673\\
\cline{2-10}
& \multicolumn{8}{c}{Fama-French target}\\\cline{2-10}
&bona fide & { -1.60447} & { -1.60370} & 0.05413 & 0.05423 & { -1.28581} & -2.32282 & { -1.85587} & { -2.82089}\\
&bona fide ridge &-1.41559 & -1.41478 & 0.06915 & 0.06884 & -1.19404 & -1.97205 & -1.72228 & \textbf{-2.66687}\\
&target & -1.61993 & -1.62003 & 0.05373 & 0.05383 & -1.29209 & -2.33933 & -1.86513 & -2.83576\\[-0.4cm]
\end{tabular}
}
\end{center}
\caption{\scriptsize Performance of traditional, bona-fide, the benchmark portfolios (LWQuEST - \cite{lw2017}, LWAnalytical - \cite{lw2020}, KZ - \cite{kz07}) and the target portfolios for the {\it minimum variance calibration} criteria. The performance measures are averaged over 1000 random portfolios of size 300. The trading period consists of 1000 days preceding 23.03.2018  and the risk aversion is set to 5. The average values are based on trimmed mean with 10\% of extreme values being dropped. The best strategies for every criteria and every values of $c$ are highlighted in bold.}\label{tab:emp:minv}
\end{table}

First, we consider the results for the mean-variance calibration in Table \ref{tab:emp:meanv}. If the dimension is low relatively to the sample size, namely $c=0.2$, then the traditional estimator shows a good and robust performance. According to virtually all performance measures it is better than any estimator based on equal correlation and Fama-French targets. The Ledoit-Wolf estimator is better only in terms of the Sharpe ratio. The dominant strategy for this value of $c$ is the Kan-Zhou estimator with the  equally weighted target which is closely followed by the suggested bona-fide shrinkage portfolio. The ranking of the targets and the estimators slightly changes if we increase $c$ to 0.5. As expected the traditional estimators becomes worse and is dominated by every target-based portfolio. The leading estimator for this value of $c$ becomes the analytical Ledoit-Wolf estimator followed by Kan-Zhou with the equally weighted portfolio as the target. Finally, we note that the application of the regularized sample covariance matrix based on  the Tikhonov (ridge) approximation \eqref{Tikhonov} leads to minor improvements in the performance of both the traditional and the bona-fide shrinkage estimators when $c=0.2$ and $c=0.5$.

With $c=0.8$ we attain the ratio of dimension to sample size where the high-dimensional asymptotics becomes relevant and simpler estimators heavily suffer from estimation risk. The traditional estimator shows extremely poor performance, which is similar to that of the numerical Ledoit-Wolf estimator. Since the Kan-Zhou estimator does not take the increasing dimension into account, it becomes worse than the target portfolios and the bona-fide shrinkage portfolios. At the same time, the analytical Ledoit-Wolf estimator becomes dominant with bona-fide ridge estimator slightly behind.
Large improvements are observed in the performance of the traditional estimator when the ridge regularization is employed in its construction. In contrast, the application of the ridge regularization to the bona-fide shrinkage portfolios leads to minor improvements.

Finally, if we increase $c$ to $2$ the sample covariance matrix is singular and we use the generalized inverse for the traditional estimator. Now the bona-fide estimator becomes clearly dominant, while the use of the equally weighted target and of the equally correlation target leads to similar results. As mentioned above the Kan-Zhou estimator is not feasible for $c>1$. Also, the application of the regularized sample covariance matrix based on  the Tikhonov (ridge) approximation improves the performance of the traditional estimator, while it leads to slightly worse performance of the bona-fide shrinkage estimator. To this end, we note a surprisingly poor performance of both Ledoit-Wolf estimators, which appear to be worse than most of the target portfolios. Noteworthy, the LWQuEST estimator is probably suffering from some numerical instabilities, while LWAnalytical behaves still very stable. If we switch to the minimum variance calibration criteria in Table \ref{tab:emp:minv}, then the ranking of the estimators remains unchanged.

Summarizing, the suggested bona-fide shrinkage estimator is comparable to the analytical Ledoit-Wolf nonlinear shrinkage estimator for $c<1$ and becomes superior starting with $c>1$. It is dominant with respect to all performance measures.  For smaller values of $c$ the Kan-Zhou estimator outperforms the bona-fide estimator and both Ledoit-Wolf estimators, and tends to show the best performance when it is used with the equally weighted target. For intermediate values of $c<1$ the analytical Ledoit-Wolf estimator is dominant, while its poor performance for $c>1$ is potentially due to the fact, that it does not take into account the estimation risk related to the sample mean vector, which is accounted for in the bona-fide shrinkage estimator and in the Kan-Zhou estimator. The numerical Ledoit-Wolf estimator seems to show huge numerical instabities in comparison to the analytic one (see, also Remark \ref{LWassumption5} below for further discussion of this finding).

 The choice of the target is an important issue and relying on the results we can make general recommendations regarding its choice. For $c<1$ the equally weighted portfolio is the best performing standalone strategy among the three alternatives. This target also leads to the shrinkage portfolio with the best overall performance. The equally correlated target is in all cases the second best choice. If $c=2$ then the order changes and the equally correlated target becomes slightly better both as a standalone strategy and as the target for the shrinkage-based approach. Since the analysis is based on 1000 random portfolios, we can, therefore, recommend using the equally weighted target for $c<1$ and equally correlated target for $c>1$. Furthermore, we can conclude that taking the best standalone  target strategy shall lead to the best performing shrinkage-based approach.


\begin{figure}[h!]
\begin{center}
\begin{tabular}{c}
\includegraphics[width=1\textwidth]{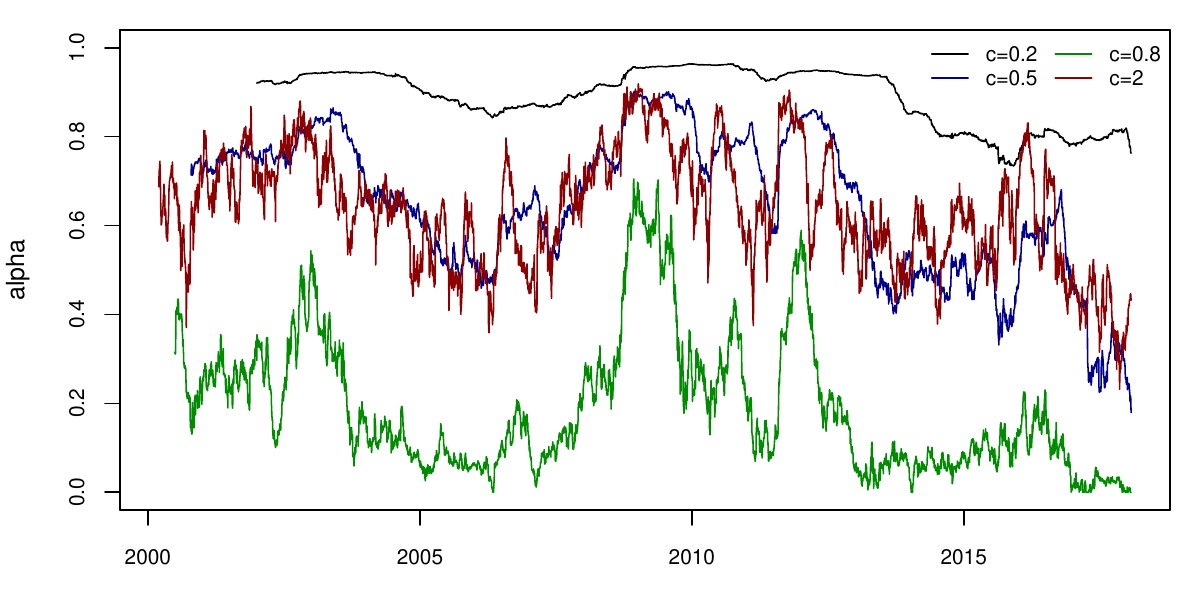}\\
\includegraphics[width=1\textwidth]{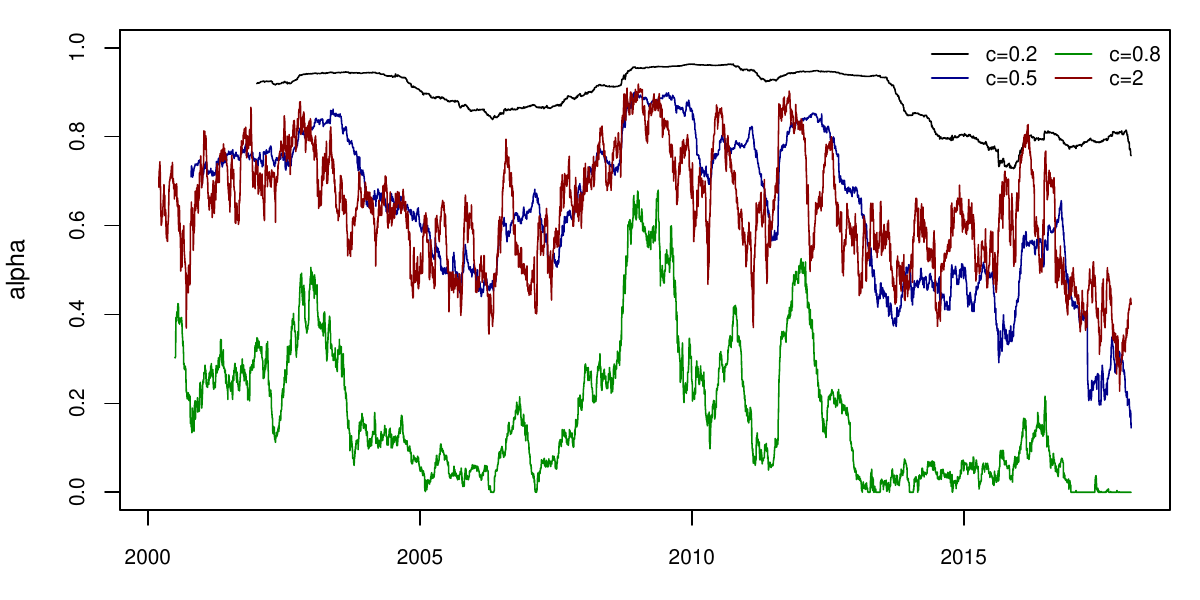}
\end{tabular}
\caption{The bona-fide shrinkage intensities for the first 100 assets (alphabetic order) using the equally weighted target portfolio and the mean-variance calibration  for  $c=0.2$, 0.5,  0.8  and 2. Above - bona fide, below - bona fide ridge (see formula \eqref{Tikhonov}).}
\label{fig:emp:alpha}
\end{center}
\end{figure}

The time series of estimated shrinkage coefficients are depicted in Figure \ref{fig:emp:alpha}. For space reasons we provide the coefficients only for the equally weighted target and the mean-variance calibration. For other parameter constellations the results are similar. The portfolio is constructed using the first alphabetically sorted assets. We observe that for small values of $c$ and thus a low estimation risk the shrinkage intensities are close to one. The behavior is very stable, but mimics the periods of high and low volatility of financial markets. Thus high volatility on financial markets causes higher shrinkage coefficients and a larger fraction of the sample EU portfolio. This can be justified by stronger effects of diversification during turmoil periods. With larger $c$ the confidence in the classical portfolio diminishes leading to a stronger preference for the equally weighted portfolio. This results in lower and more volatile shrinkage intensities. However, we observe the reverse behaviour of the estimated shrinkage intensity when $c=2$. Here, the impact of the traditional estimator in the portfolio structure increases and becomes comparable to the case of $c=0.5$. Such results are in line with our findings of the simulation study presented in Figure \ref{fig:alphastar}, where the shrinkage intensity is close to zero around $c=1$. Finally, the results obtained by employing the regularized sample covariance matrix based on  the Tikhonov (ridge) approximation leads to similar values of the shrinkage intensities independently of $c$. These are shown in the lower plot in Figure \ref{fig:emp:alpha}. This finding is in line with the values presented in Tables \ref{tab:emp:meanv} and \ref{tab:emp:minv}.

\begin{remark}\label{LWassumption5}\rm
The numerical Ledoit-Wolf estimator shows in our empirical study surprisingly poor performance, which is probably because of some numerical issues. That is why we recommend to use its new analytic version. Nevertheless, the analytical nonlinear shrinkage Ledoit-Wolf estimator still shows in our empirical study a poor performance\footnote{For $c>1$, the analytical Ledoit-Wolf estimator was initially also very unstable because the $(p-n+1)$th smallest eigenvalue was too close to zero (of order $10^{-12}$). We have corrected this issue by treating it as ``zero'' and replacing it by a specific constant. To the rest of eigenvalues the optimal nonlinear shrinkage formula was applied (see, formulas (C.4) and (C.5) in \citet[Supplement]{lw2020}). The numerical implementation of Ledoit-Wolf estimator is provided in R-package {\it HDShOP} (see, \cite{HDShOP}).} in case $c>1$, but we believe there is a specific reason for that. Indeed, in \citet[Assumption 5]{ledoit2017nonlinear} the authors assume that the sample mean vector is independently distributed of the sample covariance matrix and its distribution is rotation invariant. This assumption appears to be a characteristic property of multivariate normal distribution following \cite{lukacs1979}. Moreover, the assumption that the distribution of the sample mean vector is rotation invariant, imposes further restrictions on the data-generating model. It requires the population mean vector to be a zero vector and the population covariance matrix to be proportional to the identity matrix.
As such, it is not clear whether the Ledoit-Wolf estimator is optimal in the case of a non-zero population mean vector, i.e., in the mean-variance framework. This is also justified by the authors themselves in \citet[see Remark 5]{ledoit2017nonlinear}. It seems that this ``sample mean'' effect becomes more prominent in the case of the singular sample covariance matrix ($c>1$) and/or small risk aversion coefficient $\gamma$, i.e. when optimal portfolios lie further away from the global minimum variance portfolio on the efficient frontier. Thus, the Ledoit-Wolf estimator should be adjusted to this type of optimal portfolios before it can be efficiently used in practice when portfolio dimension is larger than the sample size, whereas the suggested bona-fide estimator for the optimal portfolio weights incorporates both the high-dimensional effects from the sample covariance matrix and the sample mean vector simultaneously.
\end{remark}

\section{Summary}

In this paper we consider the portfolio selection in the high-dimensional framework. Particularly, we assume that the number of assets $p$ and the sample size $n$ tend to infinity such that their ratio $p/n$ tends to constant $c$ where $c$ can also be larger than one, implying that we have more assets than observations. Because of the large estimation risk we suggest a shrinkage-based estimator of the portfolio weights, which shrinks the mean-variance portfolio to several target portfolios, such as the equally weighted portfolio, minimum-variance portfolio, etc. For the established shrinkage intensity we derive the limiting value which depends on $c$ and on the characteristics of the efficient frontier only. On the other side, the derived limiting expression of the shrinkage uncertainty is only an oracle value and is not feasible in practice, since it depends on unknown quantities. In order to overcome the problem,  we construct a bona-fide shrinkage estimator of the optimal portfolio weights by deducing consistent estimators of the parameters of the efficient frontier under the high-dimensional setting. As a result, a fully data-driving approach is established for constructing a practically feasible estimator for the weights of the optimal mean-variance portfolios. From the technical point of view, we rely on random matrix theory and work with the asymptotic behavior of linear and quadratic forms in the sample mean vector and in the (pseudo)-inverse sample covariance matrix. In extensive simulation and empirical studies, we evaluate the performance of established results with artificial and real data. Only if the sample size is much larger than the portfolio dimension, the traditional portfolio or the benchmark portfolio dominates the portfolio suggested in the paper.

\section*{Acknowledgment}
The authors thank Professor Christian Hansen and two anonymous Reviewers for their comments and suggestions which have improved the presentation of the paper. We also thank a Reviewer for the comment which prompted us to reconsider the proofs of the theorems and corollaries and, as a result, allowed us to simplify considerably the conditions needed for their validity.
We gratefully acknowledge the comments from the participants at the conference ''Modern Stochastics: Theory and Applications'' (Kyiv, Ukraine), the 9th Annual SoFiE Conference 2016 (Hong Kong), the Meeting of the German Statistical Society ''Statistical Week 2015'' (Hamburg, Germany) and the 7th International Conference 2014 of the ERCIM WG on Computational and Methodological Statistics (Pisa, Italy). The authors are also thankful to Prof. Raymond Kan and Prof. Michael Wolf for fruitful discussion. Taras Bodnar was partially supported by the Swedish Research Council (VR) via the project {\it Bayesian Analysis of Optimal Portfolios and Their Risk Measures}.

\section{Appendix: Proofs}

Here the proofs of the theorems are given. Recall that the sample mean vector and
the sample covariance matrix are given by
\begin{equation}\label{sy_app}
 \sy_n=\dfrac{1}{n}\by_n\bi_n
 =\bm_n+\bSigma_n^{\frac{1}{2}}\sx_n~~\text{with}~~\sx_n=\dfrac{1}{n}\bx_n\bi_n
\end{equation}
and
\begin{equation}\label{Sn_app}
 \bS_n=\dfrac{1}{n}\by_n(\bI-\frac{1}{n}\bi\bi^\prime)\by_n^{\prime}
 =\bSigma_n^{\frac{1}{2}}\bV_n\bSigma_n^{\frac{1}{2}}
 ~~\text{with}~~\bV_n=\dfrac{1}{n}\bx_n(\bI-\frac{1}{n}\bi\bi^\prime)\bx_n^{\prime}\,,
 \end{equation}
respectively. Later on, we also make use of $\tbV_n$ defined by
\begin{equation}\label{tVn}
\tbV_n=\dfrac{1}{n}\bx_n\bx_n^{\prime}
 \end{equation}
and the formula for the 1-rank update of usual inverse given by (c.f., \cite{hornjohn1985})
\begin{equation}\label{bVn_inv}
\bV_n^{-1}=(\tbV_n-\sx_n\sx_n^\prime)^{-1}=\tbV_n^{-1}+\frac{\tbV_n^{-1}\sx_n\sx_n^\prime\tbV_n^{-1}}
{1-\sx_n^\prime\tbV_n^{-1}\sx_n}
\end{equation}
as well as the formula for the 1-rank update of Moore-Penrose inverse (see, \cite{meyer1973}) expressed as
\begin{eqnarray}\label{bVn_MPinv}
\bV_n^+&=& \left(\tbV^\prime-\sx_n\sx_n^\prime\right)^+\nonumber \\
&=& \tbV_n^+-\dfrac{\tbV_n^+\sx_n\sx_n^\prime(\tbV_n^+)^2+(\tbV_n^+)^2\sx_n\sx_n^\prime
(\tbV_n^+)}{\sx_n^\prime(\tbV_n^+)^2\sx_n} +\dfrac{\sx_n^\prime(\tbV_n^+)^3\sx_n}
{(\sx_n^\prime(\tbV_n^+)^2\sx_n)^2}\tbV_n^+\sx_n\sx_n^\prime\tbV_n^+ \,.
\end{eqnarray}

First, we present an important lemma which is a special case of Theorem 1 in \cite{rubmes2011}.

\begin{lemma}\label{lem1}
Assume (A2). Let a nonrandom $p\times p$-dimensional matrix $\mathbf{\Theta}_p$ and a nonrandom $n\times n$-dimensional matrix $\bTheta_n$ possess a uniformly bounded trace norms (sum of singular values). Then it holds that
\begin{eqnarray}\label{RM2011_id}
&&\left|\text{tr}\left(\mathbf{\Theta}_p(\tbV_n-z\bI_p)^{-1}\right)
-m(z)\text{tr}\left(\mathbf{\Theta}_p\right)\right|\stackrel{a.s.}{\longrightarrow}0 \\
&&\left|\text{tr}\left(\mathbf{\Theta}_n(1/n\bx_n^\prime\bx_n-z\bI_n)^{-1}\right)
-\underline{m}(z)\text{tr}\left(\mathbf{\Theta}_n\right)\right|\stackrel{a.s.}{\longrightarrow}0
\end{eqnarray}
for $p/n\longrightarrow c \in (0, +\infty)$ as $n\rightarrow\infty$, where
\begin{equation}\label{mm}
 m(z)=(x(z)-z)^{-1}~~\text{and}~~\underline{m}(z)=-\dfrac{1-c}{z}+cm(z)
\end{equation}
with
\begin{equation}\label{RM2011_id_xz}
x(z)=\dfrac{1}{2}\left(1-c+z+\sqrt{(1-c+z)^2-4z}\right)\,.
\end{equation}
\end{lemma}

\begin{proof}[Proof of Lemma \ref{lem1}:]
The application of Theorem 1 in \cite{rubmes2011} leads to (\ref{RM2011_id}) where $x(z)$ is a unique solution in $\mathbbm{C}^+$ of the following equation
\begin{equation}\label{eq1-Lemma6_1}
\dfrac{1-x(z)}{x(z)}=\dfrac{c}{x(z)-z}\,.
\end{equation}
The two solutions of (\ref{eq1-Lemma6_1}) are given by
\begin{equation}\label{solx}
x_{1,2}(z)=\dfrac{1}{2}\left(1-c+z\pm\sqrt{(1-c+z)^2-4z}\right)\,.
\end{equation}
In order to decide which of two solutions is feasible, we note that $x_{1,2}(z)$ is the Stieltjes transform with a positive imaginary part. Thus, without loss of generality, we can take $z=1+c+i2\sqrt{c}$ and get
\begin{equation}\label{im}
\textbf{Im}\{x_{1,2}(z)\}=\textbf{Im}\left\{\dfrac{1}{2}\left(2+i2\sqrt{c}\pm i2\sqrt{2c}\right)\right\}=\textbf{Im}\left\{1+i\sqrt{c}(1\pm\sqrt{2})\right\}=\sqrt{c}\left(1\pm\sqrt{2}\right)\,,
\end{equation}
which is positive only if the sign $"+"$ is chosen. Hence, the solution is given by
\begin{equation}\label{solx_a}
x(z)=\dfrac{1}{2}\left(1-c+z+\sqrt{(1-c+z)^2-4z}\right)\,.
\end{equation}
The second assertion of the lemma follows directly from \cite{baisil2010}.
\end{proof}

 We note here that Lemma \ref{lem1} is a special case of Theorem 1 in \cite{rubmes2011}, where one has uniform convergence in the statement of the theorem. Although it is not precisely written in the statement of Theorem 1 in \cite{rubmes2011}, this observation follows from its proof on page 600 where after showing pointwise convergence Rubio and Mestre additionally proved the uniform convergence by applying Montel's theorem. In short, they first show that the random sequence of analytic functions of interest forms a normal family and, thus, by Montel's theorem there exists a subsequence of it, which converges uniformly on each compact subset of $\mathbbm{C}\setminus\mathbbm{R}^+$ to an analytic function and this one vanishes almost surely on $\mathbbm{C}\setminus\mathbbm{R}^+$. And so, the entire sequence converges uniformly to zero on every compact subset of $\mathbbm{C}\setminus\mathbbm{R}^+$.
 Furthermore, it is mentioned on page 348 of \cite{RubioMestrePalomar2012} that the convergence in Theorem 1 of \cite{rubmes2011} is in fact uniform.

Moreover, the following result (see, e.g., Theorem 1 on page 176 in \cite{ahlfors1953}), known as the Weierstrass theorem on the uniform convergence, will be used in a sequel together with Lemma \ref{lem1} in the proofs of the technical lemmas.

  \begin{theorem}[Weierstrass]\label{weierstrass}
    Suppose that $f_n(z)$ is analytic in the region $\Omega_n$, and that the sequence $\{f_n(z)\}$ converges to a limit function $f(z)$ in a region $\Omega$, uniformly on every compact subset of $\Omega$. Then $f(z)$ is analytic in $\Omega$. Moreover, $f'(z)$ converges uniformly to $f'(z)$ on every compact subset of $\Omega$.
  \end{theorem}

Because the convergence in Lemma \ref{lem1} is uniform over $z$ on every compact subset of $\mathbbm{C}\setminus\mathbbm{R}^+$, the Weierstrass theorem allows us to interchange any derivative with respect to $z$ and the limit $n\to\infty$. We will consider compact subsets, which are the small neighbourhoods of zero with $\Re(z)=0$ (without loss of generality) because all of the times we will let $z\to0$ in order to get specific limiting expressions of interest. For example, one may take $\Omega$ as a unit disk $|z|<1$ and $\Omega_n$ as a disk $|z|<\varepsilon_n$ for some $\varepsilon_n\to0$ as $n\to\infty$. The analyticity of the function $\text{tr}\left(\mathbf{\Theta}_p(\tbV_n-z\bI_p)^{-1}\right)$ follows immediately from the properties of the Stieltjes transform.


\begin{lemma}\label{lem2}
Assume (A2). Let $\boldsymbol{\theta}$ and $\boldsymbol{\xi}$ be universal nonrandom vectors with bounded Euclidean norms. Then it holds that
\begin{eqnarray}
  \left|\boldsymbol{\xi}^\prime\tbV_n^{-1}\boldsymbol{\theta}- (1-c)^{-1} \boldsymbol{\xi}^\prime \boldsymbol{\theta}\right| &\stackrel{a.s.}{\longrightarrow}& 0 \,,\label{1}\\
   \sx_n^\prime\tbV_n^{-1}\sx_n &\stackrel{a.s.}{\longrightarrow}& c  \,,\label{2}\\
 \sx_n^\prime\tbV_n^{-1}\boldsymbol{\theta}&\stackrel{a.s.}{\longrightarrow}&0\,, \label{3}\\
  \left|\boldsymbol{\xi}^\prime\tbV_n^{-2}\boldsymbol{\theta}- (1-c)^{-3} \boldsymbol{\xi}^\prime \boldsymbol{\theta}\right| &\stackrel{a.s.}{\longrightarrow}& 0 \,,\label{4}\\
   \sx_n^\prime\tbV_n^{-2}\sx_n &\stackrel{a.s.}{\longrightarrow}& \frac{c}{(1-c)}  \,,\label{5}\\
 \sx_n^\prime\tbV_n^{-2}\boldsymbol{\theta}&\stackrel{a.s.}{\longrightarrow}&0 \label{6}\end{eqnarray}
$\text{for}~ p/n\longrightarrow c \in (0,1) ~\text{as} ~n\rightarrow\infty ~~$.
\end{lemma}

\begin{proof}[Proof of Lemma \ref{lem2}:]
Since the trace norm of $\btheta\bxi^\prime$ is uniformly bounded, i.e.
\[||\btheta\bxi^\prime||_{tr}\le \sqrt{\btheta^\prime\btheta}\sqrt{\bxi^\prime\bxi} <\infty,\]
we get from Lemma \ref{lem1} that
\[|tr((\tbV_n-z\bI_p)^{-1}\boldsymbol{\theta}\boldsymbol{\xi}^\prime)-m(z)tr(\boldsymbol{\theta}\boldsymbol{\xi}^\prime)| \stackrel{a.s.}{\longrightarrow}0~~\text{for}~ p/n\rightarrow c<1~ \text{as} ~n\rightarrow\infty\]

Furthermore, the application of $m(z)\rightarrow (1-c)^{-1}$ as $z \rightarrow 0$ leads to
\begin{equation*}
|\bxi^\prime\tbV_n^{-1}\btheta-(1-c)^{-1} \bxi^\prime \btheta| \stackrel{a.s.}{\longrightarrow}0~~\text{for}~ p/n\rightarrow c<1~ \text{as} ~n\rightarrow\infty\,,
\end{equation*}
which proves \eqref{1}.

For deriving \eqref{2} we consider
\begin{eqnarray*}
\sx_n^\prime \tbV_n^{-1} \sx_n &=&
 \lim\limits_{z\rightarrow0}\text{tr}\left[\frac{1}{\sqrt{n}}\bx^\prime_n\left(\dfrac{1}{n}\bx_n\bx_n^\prime-z\bI_p\right)^{-1}
\frac{1}{\sqrt{n}}\bx_n\left(\dfrac{\bi_n\bi^\prime_n}{n}\right)\right]\\
&=& \lim\limits_{z\rightarrow0} \text{tr}\left[\left(\dfrac{\bi_n\bi^\prime_n}{n}\right)\right]+z\text{tr}\left[\left(\dfrac{1}{n}\bx_n^\prime\bx_n-z\bI_n\right)^{-1}\left(\dfrac{\bi_n\bi^\prime_n}{n}\right)\right]\,,
\end{eqnarray*}
where the last equality follows from the Woodbury formula (e.g., \cite{hornjohn1985}). The application of Lemma \ref{lem1} and Theorem \ref{weierstrass} lead to
\begin{equation*}
\text{tr}\left[\left(\dfrac{\bi_n\bi^\prime_n}{n}\right)\right]+z\text{tr}\left[\left(\dfrac{1}{n}\bx_n^\prime\bx_n-z\bI_n\right)^{-1}\left(\dfrac{\bi_n\bi^\prime_n}{n}\right)\right]
\stackrel{a.s.}{\longrightarrow}\left[1+(c-1)+czm(z)\right]\text{tr}\left[\left(\dfrac{\bi_n\bi^\prime_n}{n}\right)\right]
\end{equation*}
for $p/n\longrightarrow c<1$ as $n\rightarrow \infty$ where $m(z)$ is given by (\ref{mm}). Setting $z\rightarrow0$ and taking into account $\lim\limits_{z\rightarrow0}m(z)=\dfrac{1}{1-c}$ we get
\begin{equation*}
\sx_n^\prime \tbV_n^{-1} \sx_n\stackrel{a.s.}{\longrightarrow} 1+c-1 = c~\text{for}~\dfrac{p}{n}\longrightarrow c\in(0,1)~\text{as}~n\rightarrow\infty\,.
\end{equation*}
The result \eqref{3} was derived in \cite{pan2014} (see, p. 673 of this reference).

Next, we prove \eqref{4}. It holds that
\begin{eqnarray*}
 \bxi^\prime\tbV_n^{-2}\btheta&=& \left.\dfrac{\partial}{\partial z}\text{tr}\left[\left(\tbV_n-z\bI_p\right)^{-1}\btheta\bxi^\prime\right]\right|_{z=0}=\left.\dfrac{\partial}{\partial z} \zeta_n(z)\right|_{z=0}
 \end{eqnarray*}
where $\zeta_n(z)=\text{tr}\left[\left(\tbV_n-z\bI\right)^{-1} \btheta\bxi^\prime\right]$. From Lemma \ref{lem1} $\zeta_n(z)$ tends a.s. to $m(z)\bxi^\prime\btheta$ as $n\rightarrow\infty$. Furthermore,
{\small
\begin{equation}\label{derxx}
\left.\dfrac{\partial}{\partial z}m(z)\right|_{z=0}=\left.\dfrac{\partial}{\partial z}\dfrac{1}{ x(z)-z}\right|_{z=0}=-\left.\dfrac{x^\prime(z)-1}{ (x(z)-z)^2}\right|_{z=0}
=-\left.\dfrac{\dfrac{1}{2}\left(1-\frac{1+c-z}{\sqrt{(1-c+z)^2-4z}}\right)-1}{ (x(z)-z)^2}\right|_{z=0}
= \dfrac{1}{(1-c)^3}\,.
\end{equation}
}
Consequently, using Lemma \ref{lem1} and Theorem \ref{weierstrass} we conclude
\begin{equation*}
|\bxi^\prime\bS_n^{-2}\btheta-(1-c)^{-3}\bxi^\prime \bSigma^{-1}_n\btheta| \stackrel{a.s.}{\longrightarrow}0~~\text{for}~ p/n\rightarrow c<1~ \text{as} ~n\rightarrow\infty\,.
\end{equation*}

Let $\eta_n(z) =\sx_n^\prime(\tbV_n-z\bI)^{-1}\sx_n$ and $\bTheta_n=\left(\dfrac{\bi_n\bi^\prime_n}{n}\right)$. Then
\begin{equation*}
\sx_n^\prime\tbV_n^{-2}\sx_n =\left.\dfrac{\partial}{\partial z} \eta_n(z)\right|_{z=0}\,,
\end{equation*}
where
\begin{eqnarray*}
 \eta_n(z)& =& \text{tr}\left[\frac{1}{\sqrt{n}}\bx^\prime_n\left(\dfrac{1}{n}\bx_n\bx_n^\prime-z\bI_p\right)^{-1}\frac{1}{\sqrt{n}}\bx_n\bTheta_n\right]\\
&=& \text{tr}(\bTheta_n)+ z\text{tr}\left[ (1/n\bx_n^\prime\bx_n-z\bI_n)^{-1} \bTheta_n  \right] \stackrel{a.s.}{\longrightarrow} 1+z\underline{m}(z)=c+czm(z)\,
\end{eqnarray*}
$~~\text{for}~\dfrac{p}{n}\rightarrow c\in(0,1)~\text{as}~n\rightarrow\infty$. Hence, application of Lemma \ref{lem1} and Theorem \ref{weierstrass} reveals
\begin{equation*}
\sx_n^\prime\tbV_n^{-2}\sx_n  \stackrel{a.s.}{\longrightarrow} cm(0)+cz \left.\dfrac{\partial}{\partial z}m(z)\right|_{z=0}=\frac{c}{1-c}
\end{equation*}
$~~\text{for}~\dfrac{p}{n}\rightarrow c\in(0,1)~\text{as}~n\rightarrow\infty$.

Finally, we get
\begin{equation*}
\sx_n^\prime\tbV_n^{-2}\boldsymbol{\theta}= \left.\dfrac{\partial}{\partial z}\text{tr}\left[\sx_n^\prime\left(\tbV_n-z\bI_p\right)^{-1}\btheta\right]\right|_{z=0}  \stackrel{a.s.}{\longrightarrow} 0
\end{equation*}
$~~\text{for}~\dfrac{p}{n}\rightarrow c\in(0,1)~\text{as}~n\rightarrow\infty$.
\end{proof}

\vspace{0.5cm}
\begin{lemma}\label{lem3}
Assume (A2). Let $\boldsymbol{\theta}$ and $\boldsymbol{\xi}$ be universal nonrandom vectors with bounded Euclidean norms. Then it holds that
\begin{eqnarray}
  \left|\boldsymbol{\xi}^\prime\bV_n^{-1}\boldsymbol{\theta}- (1-c)^{-1} \boldsymbol{\xi}^\prime \boldsymbol{\theta}\right| &\stackrel{a.s.}{\longrightarrow}& 0 \,,\label{2_1}\\
   \sx_n^\prime\bV_n^{-1}\sx_n &\stackrel{a.s.}{\longrightarrow}& \frac{c}{1-c}  \,,\label{2_2}\\
 \sx_n^\prime\bV_n^{-1}\boldsymbol{\theta}&\stackrel{a.s.}{\longrightarrow}&0  \,,\label{2_3}\\
  \left|\boldsymbol{\xi}^\prime\bV_n^{-2}\boldsymbol{\theta}- (1-c)^{-3} \boldsymbol{\xi}^\prime \boldsymbol{\theta}\right| &\stackrel{a.s.}{\longrightarrow}& 0 \,,\label{2_4}\\
   \sx_n^\prime\bV_n^{-2}\sx_n &\stackrel{a.s.}{\longrightarrow}& \frac{c}{(1-c)^3}   \,,\label{2_5}\\
 \sx_n^\prime\bV_n^{-2}\boldsymbol{\theta}&\stackrel{a.s.}{\longrightarrow}&0 \label{2_6}
\end{eqnarray}
$\text{for}~ p/n\longrightarrow c \in (0,1) ~\text{as} ~n\rightarrow\infty$.
\end{lemma}

\begin{proof}[Proof of Lemma \ref{lem3}:]
From \eqref{bVn_inv} we obtain
\begin{eqnarray*}
\bxi^\prime\bV_n^{-1}\btheta&=&\bxi^\prime\tbV_n^{-1}\btheta+\frac{\bxi^\prime\tbV_n^{-1}\sx_n\sx_n^\prime\tbV_n^{-1}\btheta}
{1-\sx_n^\prime\tbV_n^{-1}\sx_n}\stackrel{a.s.}{\longrightarrow} (1-c)^{-1} \boldsymbol{\xi}^\prime \boldsymbol{\theta}
\end{eqnarray*}
$\text{for}~ p/n\longrightarrow c \in (0,1) ~\text{as} ~n\rightarrow\infty$ following \eqref{1}-\eqref{3}. Similarly, we get \eqref{2_2} and \eqref{2_3}.

In case of \eqref{2_3}, we get
\begin{eqnarray*}
\bxi^\prime\bV_n^{-2}\btheta&=&\bxi^\prime\tbV_n^{-2}\btheta+\frac{\bxi^\prime\tbV_n^{-2}\sx_n\sx_n^\prime\tbV_n^{-1}\btheta}
{1-\sx_n^\prime\tbV_n^{-1}\sx_n}+\frac{\bxi^\prime\tbV_n^{-1}\sx_n\sx_n^\prime\tbV_n^{-2}\btheta}
{1-\sx_n^\prime\tbV_n^{-1}\sx_n}\\
&+&\sx_n^\prime\tbV_n^{-2}\sx_n \frac{\bxi^\prime\tbV_n^{-1}\sx_n\sx_n^\prime\tbV_n^{-1}\btheta}
{(1-\sx_n^\prime\tbV_n^{-1}\sx_n)^2} \stackrel{a.s.}{\longrightarrow} (1-c)^{-1} \boldsymbol{\xi}^\prime \boldsymbol{\theta}
\end{eqnarray*}
$\text{for}~ p/n\longrightarrow c \in (0,1) ~\text{as} ~n\rightarrow\infty$. Similarly,
\begin{eqnarray*}
\sx_n^\prime\bV_n^{-2}\sx_n&=&\sx_n^\prime\tbV_n^{-2}\sx_n+\frac{\sx_n^\prime\tbV_n^{-2}\sx_n\sx_n^\prime\tbV_n^{-1}\sx_n}
{1-\sx_n^\prime\tbV_n^{-1}\sx_n}+\frac{\sx_n^\prime\tbV_n^{-1}\sx_n\sx_n^\prime\tbV_n^{-2}\sx_n}
{1-\sx_n^\prime\tbV_n^{-1}\sx_n}\\
&+&\sx_n^\prime\tbV_n^{-2}\sx_n \frac{\sx_n^\prime\tbV_n^{-1}\sx_n\sx_n^\prime\tbV_n^{-1}\sx_n}
{(1-\sx_n^\prime\tbV_n^{-1}\sx_n)^2}=\frac{\sx_n^\prime\tbV_n^{-2}\sx_n}{(1-\sx_n^\prime\tbV_n^{-1}\sx_n)^2} \stackrel{a.s.}{\longrightarrow} \frac{c}{(1-c)^3}
\end{eqnarray*}
and
\begin{eqnarray*}
\sx_n^\prime\bV_n^{-2}\btheta &=& \sx_n^\prime\tbV_n^{-2}\btheta+\frac{\sx_n^\prime\tbV_n^{-2}\sx_n\sx_n^\prime\tbV_n^{-1}\btheta}
{1-\sx_n^\prime\tbV_n^{-1}\sx_n}+\frac{\sx_n^\prime\tbV_n^{-1}\sx_n\sx_n^\prime\tbV_n^{-2}\btheta}
{1-\sx_n^\prime\tbV_n^{-1}\sx_n}\\
&+&\sx_n^\prime\tbV_n^{-2}\sx_n \frac{\sx_n^\prime\tbV_n^{-1}\sx_n\sx_n^\prime\tbV_n^{-1}\btheta}
{(1-\sx_n^\prime\tbV_n^{-1}\sx_n)^2} \stackrel{a.s.}{\longrightarrow} 0
\end{eqnarray*}
$\text{for}~ p/n\longrightarrow c \in (0,1) ~\text{as} ~n\rightarrow\infty$.
\end{proof}

\vspace{1cm}
\begin{lemma}\label{lem4}
Assume (A2). Let $\boldsymbol{\theta}$ and $\boldsymbol{\xi}$ be universal nonrandom vectors with bounded Euclidean norms and let $\bP_n=\bV_n^{-1}-\frac{\bV_n^{-1}\boldsymbol{\eta}\boldsymbol{\eta}^\prime\bV_n^{-1}}{\boldsymbol{\eta}^\prime\bV_n^{-1}\boldsymbol{\eta}}$ where $\boldsymbol{\eta}$ is a universal nonrandom vectors with bounded Euclidean norm. Then it holds that
\begin{eqnarray}
  \boldsymbol{\xi}^\prime\bP_n\boldsymbol{\theta} &\stackrel{a.s.}{\longrightarrow}& (1-c)^{-1}\left(\boldsymbol{\xi}^\prime \boldsymbol{\theta}-\frac{\bxi^\prime \boldsymbol{\eta}\boldsymbol{\eta}^\prime \btheta}{\boldsymbol{\eta}^\prime \boldsymbol{\eta}}\right) \,,\label{4_1}\\
   \sx_n^\prime\bP_n\sx_n &\stackrel{a.s.}{\longrightarrow}& \frac{c}{1-c}  \,,\label{4_2}\\
 \sx_n^\prime\bP_n\boldsymbol{\theta}&\stackrel{a.s.}{\longrightarrow}&0  \,,\label{4_3}\\
  \boldsymbol{\xi}^\prime\bP_n^{2}\boldsymbol{\theta} &\stackrel{a.s.}{\longrightarrow}&  (1-c)^{-3} \left(\boldsymbol{\xi}^\prime \boldsymbol{\theta}-\frac{\bxi^\prime \boldsymbol{\eta}\boldsymbol{\eta}^\prime \btheta}{\boldsymbol{\eta}^\prime \boldsymbol{\eta}}\right) \,,\label{4_4}\\
   \sx_n^\prime\bP_n^{2}\sx_n &\stackrel{a.s.}{\longrightarrow}& \frac{c}{(1-c)^3}   \,,\label{4_5}\\
 \sx_n^\prime\bP_n^{2}\boldsymbol{\theta}&\stackrel{a.s.}{\longrightarrow}&0 \label{4_6}
\end{eqnarray}
$\text{for}~ p/n\longrightarrow c \in (0,1) ~\text{as} ~n\rightarrow\infty$.
\end{lemma}

\begin{proof}[Proof of Lemma \ref{lem4}:]
It holds that
\begin{eqnarray*}
\boldsymbol{\xi}^\prime\bP_n\boldsymbol{\theta}&=&\bxi^\prime\bV_n^{-1}\btheta-\frac{\bxi^\prime\bV_n^{-1}\boldsymbol{\eta}\boldsymbol{\eta}^\prime\bV_n^{-1}\btheta}
{\boldsymbol{\eta}^\prime\bV_n^{-1}\boldsymbol{\eta}}
\stackrel{a.s.}{\longrightarrow} (1-c)^{-1} \left(\boldsymbol{\xi}^\prime \boldsymbol{\theta}-\frac{\bxi^\prime \boldsymbol{\eta}\boldsymbol{\eta}^\prime \btheta}{\boldsymbol{\eta}^\prime \boldsymbol{\eta}}\right)
\end{eqnarray*}
$\text{for}~ p/n\longrightarrow c \in (0,1) ~\text{as} ~n\rightarrow\infty$ following \eqref{2_1}. Similarly, we get
\begin{eqnarray*}
\sx_n^\prime\bP_n\sx_n&=&\sx_n^\prime\bV_n^{-1}\sx_n-\frac{\sx_n^\prime\bV_n^{-1}\boldsymbol{\eta}\boldsymbol{\eta}^\prime\bV_n^{-1}\sx_n}
{\boldsymbol{\eta}^\prime\bV_n^{-1}\boldsymbol{\eta}}
\stackrel{a.s.}{\longrightarrow} \frac{c}{1-c}
\end{eqnarray*}
and
\begin{eqnarray*}
\sx_n^\prime\bP_n\boldsymbol{\theta}&=&\sx_n^\prime\bV_n^{-1}\btheta-\frac{\sx_n^\prime\bV_n^{-1}\boldsymbol{\eta}\boldsymbol{\eta}^\prime\bV_n^{-1}\btheta}
{\boldsymbol{\eta}^\prime\bV_n^{-1}\boldsymbol{\eta}}
\stackrel{a.s.}{\longrightarrow}0
\end{eqnarray*}
$\text{for}~ p/n\longrightarrow c \in (0,1) ~\text{as} ~n\rightarrow\infty$.

The rest of the proof follows from the equality
\[\bP_n^2=\bV_n^{-2}-\frac{\bV_n^{-2}\boldsymbol{\eta}\boldsymbol{\eta}^\prime\bV_n^{-1}}{\boldsymbol{\eta}^\prime\bV_n^{-1}\boldsymbol{\eta}}
-\frac{\bV_n^{-1}\boldsymbol{\eta}\boldsymbol{\eta}^\prime\bV_n^{-2}}{\boldsymbol{\eta}^\prime\bV_n^{-1}\boldsymbol{\eta}}
+\boldsymbol{\eta}^\prime\bV_n^{-2}\boldsymbol{\eta}\frac{\bV_n^{-1}\boldsymbol{\eta}\boldsymbol{\eta}^\prime\bV_n^{-1}}{(\boldsymbol{\eta}^\prime\bV_n^{-1}\boldsymbol{\eta})^2}
\]
and Lemma \ref{lem3}.
\end{proof}

\vspace{0.5cm}

\begin{proof}[Proof of Theorem 2.1:]
Let $q_n=\max\{\bm_n^\prime\bSigma_n^{-1}\bm_n,\mathbf{b}^\prime\bSigma_n\mathbf{b}\}$. We get that $q_n>0$ uniformly in $p$, since $\bm_n^\prime\bSigma_n^{-1}\bm_n \ge s$ and $s>0$ uniformly in $p$ by Assumption (A3).

 The optimal shrinkage intensity can be rewritten in the following way
{\footnotesize
 \begin{eqnarray}\label{alfa_app}
  \alpha_n^*&=& \beta^{-1}\dfrac{\hat{\bw}^\prime_S(\bm_n-\beta\bSigma_n\mathbf{b})-\mathbf{b}^\prime(\bm_n-\beta\bSigma_n\mathbf{b})}{\hat{\bw}_S^\prime\bSigma_n\hat{\bw}_S
  -2\mathbf{b}^\prime\bSigma_n\hat{\bw}_S+\mathbf{b}^\prime\bSigma_n\mathbf{b}}\\
  &=& \beta^{-1}\dfrac{\dfrac{\bi^\prime\bS_n^{-1}(\bm_n-\beta\bSigma_n\mathbf{b})}{\bi^\prime\bS_n^{-1}\bi}
  +\gamma^{-1}\sy_n^\prime\hat{\bQ}_n(\bm_n-\beta\bSigma_n\mathbf{b})-\mathbf{b}^\prime(\bm_n-\beta\bSigma_n\mathbf{b})}
  {\dfrac{\bi^\prime\bS_n^{-1}\bSigma_n\bS_n^{-1}\bi}{(\bi^\prime\bS_n^{-1}\bi)^2}+2\gamma^{-1}\dfrac{\sy_n^\prime\hat{\bQ}_n\bSigma_n\bS_n^{-1}\bi}
  {\bi^\prime\bS_n^{-1}\bi}+\gamma^{-2}\sy_n^\prime\hat{\bQ}_n\bSigma_n\hat{\bQ}_n\sy_n-2\dfrac{\mathbf{b}^\prime\bSigma_n\bS_n^{-1}\bi}{\bi^\prime\bS_n^{-1}\bi}
  -2\gamma^{-1}\mathbf{b}^\prime\bSigma_n\hat{\bQ}_n\sy_n+\mathbf{b}^\prime\bSigma_n\mathbf{b}}\nonumber\\
&=&\beta^{-1}\frac{A_n^*}{B_n^*},\nonumber
 \end{eqnarray}
 }
where
\begin{eqnarray*}
  A_n^* &=&\frac{1}{q_n}\left(\frac{\bi^\prime\bS_n^{-1}(\bm_n-\beta\bSigma_n\mathbf{b})}{\bi^\prime\bS_n^{-1}\bi}
  +\gamma^{-1}\sy_n^\prime\hat{\bQ}_n(\bm_n-\beta\bSigma_n\mathbf{b})-\mathbf{b}^\prime(\bm_n-\beta\bSigma_n\mathbf{b})\right)\\
  &=&\frac{(\bi^\prime\bSigma_n^{-1}\bi)^{-1/2}\sqrt{\bm_n^\prime\bSigma_n^{-1}\bm_n}}{q_n} \frac{(\bi^\prime\bSigma_n^{-1}\bi)^{-1/2}(\bm_n^\prime\bSigma_n^{-1}\bm_n)^{-1/2}\bi^\prime\bS_n^{-1}\bm_n}
  {(\bi^\prime\bSigma_n^{-1}\bi)^{-1}\bi^\prime\bS_n^{-1}\bi}\\
  &-&\beta \frac{(\bi^\prime\bSigma_n^{-1}\bi)^{-1/2}\sqrt{\mathbf{b}^\prime\bSigma_n\mathbf{b}}}{q_n} \frac{(\bi^\prime\bSigma_n^{-1}\bi)^{-1/2}(\mathbf{b}^\prime\bSigma_n\mathbf{b})^{-1/2}\bi^\prime\bS_n^{-1}\bSigma_n\mathbf{b}}
  {(\bi^\prime\bSigma_n^{-1}\bi)^{-1}\bi^\prime\bS_n^{-1}\bi}\\
  &+&\gamma^{-1}\frac{\bm_n^\prime\bSigma_n^{-1}\bm_n}{q_n}\frac{\sy_n^\prime\hat{\bQ}_n\bm_n}{\bm_n^\prime\bSigma_n^{-1}\bm_n}
  -\beta\gamma^{-1}\frac{\sqrt{\bm_n^\prime\bSigma_n^{-1}\bm_n}\sqrt{\mathbf{b}^\prime\bSigma_n\mathbf{b}}}{q_n}
  \frac{\sy_n^\prime\hat{\bQ}_n\bSigma_n\mathbf{b}}{\sqrt{\bm_n^\prime\bSigma_n^{-1}\bm_n}\sqrt{\mathbf{b}^\prime\bSigma_n\mathbf{b}}}\\
  &-&\frac{\sqrt{\bm_n^\prime\bSigma_n^{-1}\bm_n}\sqrt{\mathbf{b}^\prime\bSigma_n\mathbf{b}}}{q_n}\frac{\mathbf{b}^\prime\bm_n}
  {\sqrt{\bm_n^\prime\bSigma_n^{-1}\bm_n}\sqrt{\mathbf{b}^\prime\bSigma_n\mathbf{b}}} +\beta\frac{\mathbf{b}^\prime\bSigma_n\mathbf{b}}{q_n}
\end{eqnarray*}
and
\begin{eqnarray*}
   B_n^*&=&\frac{1}{q_n}\Bigg(\dfrac{\bi^\prime\bS_n^{-1}\bSigma_n\bS_n^{-1}\bi}{(\bi^\prime\bS_n^{-1}\bi)^2}+2\gamma^{-1}\dfrac{\sy_n^\prime\hat{\bQ}_n\bSigma_n\bS_n^{-1}\bi}
  {\bi^\prime\bS_n^{-1}\bi}+\gamma^{-2}\sy_n^\prime\hat{\bQ}_n\bSigma_n\hat{\bQ}_n\sy_n\\
  &-&2\dfrac{\mathbf{b}^\prime\bSigma_n\bS_n^{-1}\bi}{\bi^\prime\bS_n^{-1}\bi}
  -2\gamma^{-1}\mathbf{b}^\prime\bSigma_n\hat{\bQ}_n\sy_n+\mathbf{b}^\prime\bSigma_n\mathbf{b} \Bigg)\\
   &=&\frac{(\bi^\prime\bSigma_n^{-1}\bi)^{-1}}{q_n}
   \frac{(\bi^\prime\bSigma_n^{-1}\bi)^{-1}\bi^\prime\bS_n^{-1}\bSigma_n\bS_n^{-1}\bi}{(\bi^\prime\bSigma_n^{-1}\bi)^{-2}(\bi^\prime\bS_n^{-1}\bi)^2}\\
   &+&2\gamma^{-1}\frac{(\bi^\prime\bSigma_n^{-1}\bi)^{-1/2}\sqrt{\bm_n^\prime\bSigma_n^{-1}\bm_n}}{q_n}
  \frac{(\bi^\prime\bSigma_n^{-1}\bi)^{-1/2}(\bm_n^\prime\bSigma_n^{-1}\bm_n)^{-1/2}\sy_n^\prime\hat{\bQ}_n\bSigma_n\bS_n^{-1}\bi}
  {(\bi^\prime\bSigma_n^{-1}\bi)^{-1}\bi^\prime\bS_n^{-1}\bi}\\
  &+&\gamma^{-2}\frac{\bm_n^\prime\bSigma_n^{-1}\bm_n}{q_n}\frac{\sy_n^\prime\hat{\bQ}_n\bSigma_n\hat{\bQ}_n\sy_n}{\bm_n^\prime\bSigma_n^{-1}\bm_n}\\
  &-&2\frac{(\bi^\prime\bSigma_n^{-1}\bi)^{-1/2}\sqrt{\mathbf{b}^\prime\bSigma_n\mathbf{b}}}{q_n}
  \frac{(\bi^\prime\bSigma_n^{-1}\bi)^{-1/2}(\mathbf{b}^\prime\bSigma_n\mathbf{b})^{-1/2}\mathbf{b}^\prime\bSigma_n\bS_n^{-1}\bi}
  {(\bi^\prime\bSigma_n^{-1}\bi)^{-1}\bi^\prime\bS_n^{-1}\bi}\\
  &-&2\gamma^{-1}\frac{\sqrt{\bm_n^\prime\bSigma_n^{-1}\bm_n}\sqrt{\mathbf{b}^\prime\bSigma_n\mathbf{b}}}{q_n}
  \frac{\mathbf{b}^\prime\bSigma_n\hat{\bQ}_n\sy_n}{\sqrt{\bm_n^\prime\bSigma_n^{-1}\bm_n}\sqrt{\mathbf{b}^\prime\bSigma_n\mathbf{b}}}
  +\frac{\mathbf{b}^\prime\bSigma_n\mathbf{b}}{q_n},
 \end{eqnarray*}

In the formulas for $A_n^*$ and $B_n^*$ the factors $q_n^{-1}\bm_n^\prime\bSigma_n^{-1}\bm_n$ and $q_n^{-1}\mathbf{b}^\prime\bSigma_n\mathbf{b}$ are bounded by one. Moreover, since $(\bi^\prime\bSigma_n^{-1}\bi)^{-1} \le \mathbf{b}^\prime\bSigma_n\mathbf{b}$ for any $\bb$ with $\bb^\prime \bi=1$, we also get $q_n^{-1}(\bi^\prime\bSigma_n^{-1}\bi)^{-1} \le 1$.

The Euclidean norms of the following vectors
\[\frac{\bSigma_n^{-1/2}\bi}{\sqrt{\bi^\prime\bSigma_n^{-1}\bi}},\quad \frac{\bSigma_n^{-1/2}\bm_n}{\sqrt{\bm_n^\prime\bSigma_n^{-1}\bm_n}}
\quad \text{and} \quad
\frac{\bSigma_n^{1/2}\mathbf{b}}{\sqrt{\mathbf{b}^\prime\bSigma_n\mathbf{b}}}\]
are all equal to one. As a result, using $\sy_n=\bm_n+\bSigma_n^{1/2}\sx_n$ and $\bS_n=\bSigma_n^{1/2}\bV_n \bSigma_n^{1/2}$ and applying Lemma \ref{lem3} we get
\begin{eqnarray*}
\frac{\bi^\prime\bS_n^{-1}\bi}{\bi^\prime\bSigma_n^{-1}\bi}&=&\frac{\bi^\prime\bSigma_n^{-1/2}\bV_n^{-1}\bSigma_n^{-1/2}\bi}{\bi^\prime\bSigma_n^{-1}\bi}
\stackrel{a.s.}{\longrightarrow}(1-c)^{-1},\\
\frac{\bi^\prime\bS_n^{-1}\bm_n}{\sqrt{\bi^\prime\bSigma_n^{-1}\bi}\sqrt{\bm_n^\prime\bSigma_n^{-1}\bm_n}}&=&
\frac{\bi^\prime\bSigma_n^{-1/2}\bV_n^{-1}\bSigma_n^{-1/2}\bm_n}{\sqrt{\bi^\prime\bSigma_n^{-1}\bi}\sqrt{\bm_n^\prime\bSigma_n^{-1}\bm_n}} \stackrel{a.s.}{\longrightarrow}(1-c)^{-1}  \frac{\bi^\prime\bSigma_n^{-1}\bm_n}{\sqrt{\bi^\prime\bSigma_n^{-1}\bi}\sqrt{\bm_n^\prime\bSigma_n^{-1}\bm_n}},\\
\frac{\bi^\prime\bS_n^{-1}\bSigma_n\mathbf{b}}{\sqrt{\bi^\prime\bSigma_n^{-1}\bi}\sqrt{\mathbf{b}^\prime\bSigma_n\mathbf{b}}}&=&
\frac{\bi^\prime\bSigma_n^{-1/2}\bV_n^{-1}\bSigma_n^{1/2}\mathbf{b}}{\sqrt{\bi^\prime\bSigma_n^{-1}\bi}\sqrt{\mathbf{b}^\prime\bSigma_n\mathbf{b}}}
\stackrel{a.s.}{\longrightarrow}(1-c)^{-1} \frac{1}{\sqrt{\bi^\prime\bSigma_n^{-1}\bi}\sqrt{\mathbf{b}^\prime\bSigma_n\mathbf{b}}},\\
\frac{\bi^\prime\bS_n^{-1}\bSigma_n\bS_n^{-1}\bi}{\bi^\prime\bSigma_n^{-1}\bi}&=&\frac{\bi^\prime\bSigma_n^{-1/2}\bV_n^{-2}\bSigma_n^{-1/2}\bi}{\bi^\prime\bSigma_n^{-1}\bi} \stackrel{a.s.}{\longrightarrow}(1-c)^{-3}
\end{eqnarray*}
$\text{for}~ p/n\longrightarrow c \in (0,1) ~\text{as} ~n\rightarrow\infty$.

Furthermore, from Lemma \ref{lem3} and \ref{lem4} using the equalities
\[\hat{\bQ}_n=\bSigma_n^{-1/2}\bP_n\bSigma_n^{-1/2} ~~\text{and}~~\bP_n\bV_n^{-1}=\bV_n^{-2}-\frac{\bV_n^{-1}\boldsymbol{\eta}\boldsymbol{\eta} ^\prime\bV_n^{-2}}
{\boldsymbol{\eta}^\prime\bV_n^{-1}\boldsymbol{\eta}}\]
with $\boldsymbol{\eta}=\bSigma_n^{-1/2}\bi/\sqrt{\bi^\prime\bSigma_n^{-1}\bi}$ we obtain with $\bm_n^\prime\bSigma_n^{-1}\bm_n>0$
\begin{eqnarray*}
\frac{\sy_n^\prime\hat{\bQ}_n\bm_n}{\bm_n^\prime\bSigma_n^{-1}\bm_n}&=&\frac{\bm_n^\prime \bSigma_n^{-1/2}\bP_n\bSigma_n^{-1/2}\bm_n}{\bm_n^\prime\bSigma_n^{-1}\bm_n}
+\frac{\sx_n^\prime\bP_n\bSigma_n^{-1/2}\bm_n}{\bm_n^\prime\bSigma_n^{-1}\bm_n}\\
&&\stackrel{a.s.}{\longrightarrow}(1-c)^{-1}\frac{\bm_n^\prime \bQ_n \bm_n}{\bm_n^\prime\bSigma_n^{-1}\bm_n},\\
\frac{\sy_n^\prime\hat{\bQ}_n\bSigma_n\mathbf{b}}{\sqrt{\bm_n^\prime\bSigma_n^{-1}\bm_n}\sqrt{\mathbf{b}^\prime\bSigma_n\mathbf{b}}}&=&\frac{\bm_n^\prime \bSigma_n^{-1/2}\bP_n\bSigma_n^{1/2}\mathbf{b}}{\sqrt{\bm_n^\prime\bSigma_n^{-1}\bm_n}\sqrt{\mathbf{b}^\prime\bSigma_n\mathbf{b}}}
+\frac{\sx_n^\prime\bP_n\bSigma_n^{1/2}\mathbf{b}}{\sqrt{\bm_n^\prime\bSigma_n^{-1}\bm_n}\sqrt{\mathbf{b}^\prime\bSigma_n\mathbf{b}}}\\
&&\stackrel{a.s.}{\longrightarrow}(1-c)^{-1}\frac{\bm_n^\prime \bQ_n\bSigma_n\mathbf{b}}{\sqrt{\bm_n^\prime\bSigma_n^{-1}\bm_n}\sqrt{\mathbf{b}^\prime\bSigma_n\mathbf{b}}},\\
\frac{\sy_n^\prime\hat{\bQ}_n\bSigma_n\bS_n^{-1}\bi}{\sqrt{\bm_n^\prime\bSigma_n^{-1}\bm_n}\sqrt{\bi^\prime\bSigma_n^{-1}\bi}} &=&\frac{\bm_n^\prime\bSigma_n^{-1/2}\bP_n\bV_n^{-1}\bSigma_n^{-1/2}\bi}{\sqrt{\bm_n^\prime\bSigma_n^{-1}\bm_n}\sqrt{\bi^\prime\bSigma_n^{-1}\bi}}
+\frac{\sx_n^\prime\bP_n\bV_n^{-1}\bSigma_n^{-1/2}\bi}{\sqrt{\bm_n^\prime\bSigma_n^{-1}\bm_n}\sqrt{\bi^\prime\bSigma_n^{-1}\bi}} \\
&&\stackrel{a.s.}{\longrightarrow}(1-c)^{-3} \frac{\bm_n^\prime\bQ_n\bi}{\sqrt{\bm_n^\prime\bSigma_n^{-1}\bm_n}\sqrt{\bi^\prime\bSigma_n^{-1}\bi}}=0,\\
\frac{\sy_n^\prime\hat{\bQ}_n\bSigma_n\hat{\bQ}_n\sy_n}{\bm_n^\prime\bSigma_n^{-1}\bm_n}&=&
\frac{\sx_n^\prime\bP_n^2\sx_n}{\bm_n^\prime\bSigma_n^{-1}\bm_n}+2\frac{\bm_n^\prime\bSigma_n^{-1/2}\bP_n^2\sx_n}{\bm_n^\prime\bSigma_n^{-1}\bm_n}
+\frac{\bm_n^\prime\bSigma_n^{-1/2}\bP_n^2\bSigma_n^{-1/2}\bm_n}{\bm_n^\prime\bSigma_n^{-1}\bm_n}\\
&&\stackrel{a.s.}{\longrightarrow}\frac{(1-c)^{-3}c + (1-c)^{-3} \bm_n^\prime\bQ_n\bm_n}{\bm_n^\prime\bSigma_n^{-1}\bm_n}
\end{eqnarray*}
$\text{for}~ p/n\longrightarrow c \in (0,1) ~\text{as} ~n\rightarrow\infty$.

Substituting the above results into the expressions of $A_n^*$ and $B_n^*$, we get that
\[|A_n^*-A^*|\stackrel{a.s.}{\longrightarrow}0
\quad \text{and} \quad
|B_n^*-B^*|\stackrel{a.s.}{\longrightarrow}0\]
with
{\footnotesize
\begin{eqnarray*}
A^* &=& \frac{(\bi^\prime\bSigma_n^{-1}\bi)^{-1/2}\sqrt{\bm_n^\prime\bSigma_n^{-1}\bm_n}}{q_n} \frac{(\bi^\prime\bSigma_n^{-1}\bi)^{-1/2}(\bm_n^\prime\bSigma_n^{-1}\bm_n)^{-1/2}(1-c)^{-1}\bi^\prime\bSigma_n^{-1}\bm_n}
  {(1-c)^{-1}}\\
  &-&\beta \frac{(\bi^\prime\bSigma_n^{-1}\bi)^{-1/2}\sqrt{\mathbf{b}^\prime\bSigma_n\mathbf{b}}}{q_n} \frac{(\bi^\prime\bSigma_n^{-1}\bi)^{-1/2}(\mathbf{b}^\prime\bSigma_n\mathbf{b})^{-1/2}(1-c)^{-1}}
  {(1-c)^{-1}}\\
  &+&\gamma^{-1}\frac{\bm_n^\prime\bSigma_n^{-1}\bm_n}{q_n}\frac{(1-c)^{-1}\bm_n^\prime\bQ_n\bm_n}{\bm_n^\prime\bSigma_n^{-1}\bm_n}
  -\beta\gamma^{-1}\frac{\sqrt{\bm_n^\prime\bSigma_n^{-1}\bm_n}\sqrt{\mathbf{b}^\prime\bSigma_n\mathbf{b}}}{q_n}
  \frac{(1-c)^{-1}\bm_n^\prime\bQ_n\bSigma_n\mathbf{b}}{\sqrt{\bm_n^\prime\bSigma_n^{-1}\bm_n}\sqrt{\mathbf{b}^\prime\bSigma_n\mathbf{b}}}\\
  &-&\frac{\sqrt{\bm_n^\prime\bSigma_n^{-1}\bm_n}\sqrt{\mathbf{b}^\prime\bSigma_n\mathbf{b}}}{q_n}\frac{\mathbf{b}^\prime\bm_n}
  {\sqrt{\bm_n^\prime\bSigma_n^{-1}\bm_n}\sqrt{\mathbf{b}^\prime\bSigma_n\mathbf{b}}} +\beta\frac{\mathbf{b}^\prime\bSigma_n\mathbf{b}}{q_n}\\
&=&\frac{1}{q_n}\left(\frac{\bi^\prime\bSigma_n^{-1}\bm_n}{\bi^\prime\bSigma_n^{-1}\bi}-\beta\frac{1}{\bi^\prime\bSigma_n^{-1}\bi}
+\frac{\gamma^{-1}}{1-c}\bm_n^\prime\bQ_n\bm_n-\frac{\gamma^{-1}\beta}{1-c}\bm_n^\prime\bQ_n\bSigma_n\mathbf{b} - \mathbf{b}^\prime\bm_n+\beta\mathbf{b}^\prime\bSigma_n\mathbf{b}\right)
\end{eqnarray*}
and
\begin{eqnarray*}
B^* &=& \frac{(\bi^\prime\bSigma_n^{-1}\bi)^{-1}}{q_n}
   \frac{(1-c)^{-3}}{(1-c)^{-2}}+2\gamma^{-1}\frac{(\bi^\prime\bSigma_n^{-1}\bi)^{-1/2}\sqrt{\bm_n^\prime\bSigma_n^{-1}\bm_n}}{q_n}
  \frac{0}{(1-c)^{-1}}\\
  &+&\gamma^{-2}\frac{\bm_n^\prime\bSigma_n^{-1}\bm_n}{q_n}\frac{(1-c)^{-3}c + (1-c)^{-3} \bm_n^\prime\bQ_n\bm_n}{\bm_n^\prime\bSigma_n^{-1}\bm_n}\\
  &-&2\frac{(\bi^\prime\bSigma_n^{-1}\bi)^{-1/2}\sqrt{\mathbf{b}^\prime\bSigma_n\mathbf{b}}}{q_n}
  \frac{(\bi^\prime\bSigma_n^{-1}\bi)^{-1/2}(\mathbf{b}^\prime\bSigma_n\mathbf{b})^{-1/2}(1-c)^{-1}}
  {(1-c)^{-1}}\\
  &-&2\gamma^{-1}\frac{\sqrt{\bm_n^\prime\bSigma_n^{-1}\bm_n}\sqrt{\mathbf{b}^\prime\bSigma_n\mathbf{b}}}{q_n}
  \frac{(1-c)^{-1}\mathbf{b}^\prime\bSigma_n\bQ_n\bm_n}{\sqrt{\bm_n^\prime\bSigma_n^{-1}\bm_n}\sqrt{\mathbf{b}^\prime\bSigma_n\mathbf{b}}}
  +\frac{\mathbf{b}^\prime\bSigma_n\mathbf{b}}{q_n}\\
&=&\frac{1}{q_n}\Bigg(\frac{1}{1-c}\frac{1}{\bi^\prime\bSigma_n^{-1}\bi}+\gamma^{-2}\left(\frac{1}{(1-c)^3}\bm_n^\prime\bQ_n\bm_n+\frac{c}{(1-c)^3}\right)
- 2\frac{1}{\bi^\prime\bSigma_n^{-1}\bi}\\
&-&2\frac{\gamma^{-1}}{1-c}\left(\mathbf{b}^\prime\bm_n
 -\frac{\bi^\prime\bSigma_n^{-1}\bm_n}{\bi^\prime\bSigma_n^{-1}\bi}\right)+\mathbf{b}^\prime\bSigma_n\mathbf{b}\Bigg).
\end{eqnarray*}
}

Let $\alpha^*=A^*/B^*$. Then,
\[|\alpha_n^*-\alpha^*|\le \left|\frac{1}{B_n^*}(A_n^*-A^*)\right|+\left|\frac{A^*}{B^*B_n^*}(B^*-B_n^*)\right|\stackrel{a.s.}{\longrightarrow}0\]
Using the notations $V_{GMV}=\dfrac{1}{\bi^\prime\bSigma_n^{-1}\bi}$, $R_{GMV}=\dfrac{\bi^\prime\bSigma_n^{-1}\bm_n}{\bi^\prime\bSigma_n^{-1}\bi}$, $s=\bm_n^\prime\bQ_n\bm_n$, $R_b=\mathbf{b}^\prime\bm_n$ and making some technical manipulations we get the statement of Theorem 2.1.\\
\end{proof}

\begin{proof}[Proof of Corollary 2.1:]
\noindent (a)
We first compute $U_S/q_n$ with $q_n=\max\{\bm_n^\prime\bSigma_n^{-1}\bm_n,\mathbf{b}^\prime\bSigma_n\mathbf{b}\}$ given  by
\begin{eqnarray*}
\frac{1}{q_n}U_S&=& \frac{1}{q_n}\hat{\bw}_S^\prime\bm_n-\dfrac{\gamma}{2}\frac{1}{q_n}\hat{\bw}_S^\prime\bSigma_n\hat{\bw}_S,
\end{eqnarray*}
where
\begin{eqnarray*}
 \frac{1}{q_n}\hat{\bw}_S^\prime\bm_n&=& \frac{1}{q_n}\frac{\bi^\prime\bS_n^{-1}\bm_n}{\bi^\prime\bS_n^{-1}\bi}+\gamma^{-1} \frac{1}{q_n}\sy_n^\prime\hat{\bQ}_n\bm_n\\
&=& \frac{(\bi^\prime\bSigma_n^{-1}\bi)^{-1/2}\sqrt{\bm_n^\prime\bSigma_n^{-1}\bm_n}}{q_n} \frac{(\bi^\prime\bSigma_n^{-1}\bi)^{-1/2}(\bm_n^\prime\bSigma_n^{-1}\bm_n)^{-1/2}\bi^\prime\bS_n^{-1}\bm_n}
  {(\bi^\prime\bSigma_n^{-1}\bi)^{-1}\bi^\prime\bS_n^{-1}\bi}\\
  &+&\gamma^{-1} \frac{\bm_n^\prime\bSigma_n^{-1}\bm_n}{q_n}\frac{\sy_n^\prime\hat{\bQ}_n\bm_n}{\bm_n^\prime\bSigma_n^{-1}\bm_n} \\
&\stackrel{a.s.}{\longrightarrow}&\frac{(\bi^\prime\bSigma_n^{-1}\bi)^{-1/2}\sqrt{\bm_n^\prime\bSigma_n^{-1}\bm_n}}{q_n} \frac{(\bi^\prime\bSigma_n^{-1}\bi)^{-1/2}(\bm_n^\prime\bSigma_n^{-1}\bm_n)^{-1/2}(1-c)^{-1}\bi^\prime\bSigma_n^{-1}\bm_n}
  {(1-c)^{-1}}\\
  &+&\gamma^{-1}(1-c)^{-1}\frac{\bm_n^\prime\bSigma_n^{-1}\bm_n}{q_n}\frac{\bm_n^\prime \bQ_n \bm_n}{\bm_n^\prime\bSigma_n^{-1}\bm_n}=\frac{1}{q_n}
  \left(R_{GMV}+\gamma^{-1}(1-c)^{-1}s\right)
\end{eqnarray*}
and
\begin{eqnarray*}
\frac{1}{q_n}\hat{\bw}_S^\prime\bSigma_n\hat{\bw}_S&=&\frac{\bi^\prime\bS_n^{-1}\bSigma_n\bS_n^{-1}\bi}{(\bi^\prime\bS_n^{-1}\bi)^2}
+2\gamma^{-1}\frac{\bi^\prime\bS_n^{-1}\bSigma_n\hat{\bQ}_n\sy_n}{\bi^\prime\bS_n^{-1}\bi}+\gamma^{-2}\sy_n^\prime\hat{\bQ}_n\bSigma_n\hat{\bQ}_n\sy_n\\
&=&\frac{(\bi^\prime\bSigma_n^{-1}\bi)^{-1}}{q_n}
   \frac{(\bi^\prime\bSigma_n^{-1}\bi)^{-1}\bi^\prime\bS_n^{-1}\bSigma_n\bS_n^{-1}\bi}{(\bi^\prime\bSigma_n^{-1}\bi)^{-2}(\bi^\prime\bS_n^{-1}\bi)^2}
   +\gamma^{-2}\frac{\bm_n^\prime\bSigma_n^{-1}\bm_n}{q_n}\frac{\sy_n^\prime\hat{\bQ}_n\bSigma_n\hat{\bQ}_n\sy_n}{\bm_n^\prime\bSigma_n^{-1}\bm_n}\\
&+& 2\gamma^{-1} \frac{(\bi^\prime\bSigma_n^{-1}\bi)^{-1/2}\sqrt{\bm_n^\prime\bSigma_n^{-1}\bm_n}}{q_n} \frac{(\bi^\prime\bSigma_n^{-1}\bi)^{-1/2}(\bm_n^\prime\bSigma_n^{-1}\bm_n)^{-1/2}\bi^\prime\bS_n^{-1}\bSigma_n\hat{\bQ}_n\sy_n}
{(\bi^\prime\bSigma_n^{-1}\bi)^{-1}\bi^\prime\bS_n^{-1}\bi} \\
&\stackrel{a.s.}{\longrightarrow}&\frac{(\bi^\prime\bSigma_n^{-1}\bi)^{-1}}{q_n}\frac{(1-c)^{-3}}{(1-c)^{-2}}
+\gamma^{-2}\frac{\bm_n^\prime\bSigma_n^{-1}\bm_n}{q_n}\frac{(1-c)^{-3}c + (1-c)^{-3} \bm_n^\prime\bQ_n\bm_n}{\bm_n^\prime\bSigma_n^{-1}\bm_n}\\
&+&2\gamma^{-1}\frac{(\bi^\prime\bSigma_n^{-1}\bi)^{-1/2}\sqrt{\bm_n^\prime\bSigma_n^{-1}\bm_n}}{q_n} \frac{0}{(1-c)^{-1}}\\
&=&\frac{1}{q_n}\left((1-c)^{-1}V_{GMV}+\gamma^{-2}\frac{c}{(1-c)^{3}}  +\gamma^{-2} (1-c)^{-3} s\right)
\end{eqnarray*}
$\text{for}~ p/n\longrightarrow c \in (0,1) ~\text{as} ~n\rightarrow\infty$.

Finally, the equality
\[\frac{1}{q_n}U_{EU}=\frac{1}{q_n}R_{GMV}+\frac{1}{2}\gamma^{-1} \frac{1}{q_n}s - \frac{\gamma}{2} \frac{1}{q_n}V_{GMV},\]
implies the statement of the first part of the corollary.

\noindent (b) It holds that
\begin{eqnarray*}
\frac{1}{q_n}U_{GSE} &=& \alpha^* \frac{1}{q_n}\hat{\bw}_S^\prime\bm_n+(1-\alpha^*)\frac{1}{q_n}\mathbf{b}^\prime\bm_n\\
&-&\frac{\gamma}{2}\left((\alpha^*)^2\frac{1}{q_n}\hat{\bw}_S^\prime\bSigma_n\hat{\bw}_S
+2\alpha^*(1-\alpha^*)\frac{1}{q_n}\mathbf{b}^\prime\bSigma_n\hat{\bw}_S+(1-\alpha^*)^2\frac{1}{q_n}\mathbf{b}^\prime\bSigma_n\mathbf{b}\right),
\end{eqnarray*}
where the asymptotic values of $\hat{\bw}_S^\prime\bm_n/ q_n$ and $\hat{\bw}_S^\prime\bSigma_n\hat{\bw}_S /q_n$ are fund in part (a) and
\begin{eqnarray*}
\frac{1}{q_n}\mathbf{b}^\prime\bSigma_n\hat{\bw}_S&=&\frac{1}{q_n}\frac{\mathbf{b}^\prime\bSigma_n\bS_n^{-1}\bi}{\bi^\prime\bS_n^{-1}\bi}
+\gamma^{-1}\frac{1}{q_n}\mathbf{b}^\prime\bSigma_n\hat{\bQ}_n\sy_n\\
&=&\frac{(\bi^\prime\bSigma_n^{-1}\bi)^{-1/2}\sqrt{\mathbf{b}^\prime\bSigma_n\mathbf{b}}}{q_n} \frac{(\bi^\prime\bSigma_n^{-1}\bi)^{-1/2}(\mathbf{b}^\prime\bSigma_n\mathbf{b})^{-1/2}\bi^\prime\bS_n^{-1}\bSigma_n\mathbf{b}}
  {(\bi^\prime\bSigma_n^{-1}\bi)^{-1}\bi^\prime\bS_n^{-1}\bi}\\
&+& \gamma^{-1}\frac{\sqrt{\bm_n^\prime\bSigma_n^{-1}\bm_n}\sqrt{\mathbf{b}^\prime\bSigma_n\mathbf{b}}}{q_n}
  \frac{\mathbf{b}^\prime\bSigma_n\hat{\bQ}_n\sy_n}{\sqrt{\bm_n^\prime\bSigma_n^{-1}\bm_n}\sqrt{\mathbf{b}^\prime\bSigma_n\mathbf{b}}} \\
&\stackrel{a.s.}{\longrightarrow}&
\frac{(\bi^\prime\bSigma_n^{-1}\bi)^{-1/2}\sqrt{\mathbf{b}^\prime\bSigma_n\mathbf{b}}}{q_n} \frac{(\bi^\prime\bSigma_n^{-1}\bi)^{-1/2}(\mathbf{b}^\prime\bSigma_n\mathbf{b})^{-1/2}(1-c)^{-1}}
  {(1-c)^{-1}}\\
&+&\gamma^{-1}\frac{\sqrt{\bm_n^\prime\bSigma_n^{-1}\bm_n}\sqrt{\mathbf{b}^\prime\bSigma_n\mathbf{b}}}{q_n}
  \frac{(1-c)^{-1}\left(\mathbf{b}^\prime\bm_n-\dfrac{\bi^\prime\bSigma_n^{-1}\bm_n}{\bi^\prime\bSigma_n^{-1}\bi}\right)}{\sqrt{\bm_n^\prime\bSigma_n^{-1}\bm_n}\sqrt{\mathbf{b}^\prime\bSigma_n\mathbf{b}}}
\\
&=&\frac{1}{q_n}\left(V_{GMV}+\gamma^{-1}(1-c)^{-1} (R_b-R_{GMV})\right).
\end{eqnarray*}
$\text{for}~ p/n\longrightarrow c \in (0,1) ~\text{as} ~n\rightarrow\infty$.

Hence,
{\footnotesize
\begin{eqnarray*}
&&\frac{1}{q_n}(U_{EU}-U_{GSE})= (\alpha^*)^2\frac{1}{q_n}U_{EU}+(1-\alpha^*)^2\frac{1}{q_n}U_{EU}+2\alpha^*(1-\alpha^*)\frac{1}{q_n}U_{EU}\\
&-&(\alpha^*)^2\frac{1}{q_n}U_{S}-\alpha^*(1-\alpha^*)\frac{1}{q_n}\hat{\bw}_S^\prime\bm_n
-(1-\alpha^*)^2\frac{1}{q_n}U_{\bb}-\alpha^*(1-\alpha^*)\frac{1}{q_n}\mathbf{b}^\prime\bm_n
+\gamma\alpha^*(1-\alpha^*)\frac{1}{q_n}\mathbf{b}^\prime\bSigma_n\hat{\bw}_S\\
&\stackrel{a.s.}{\longrightarrow}&(\alpha^*)^2\frac{1}{q_n}(U_{EU}-U_S)+(1-\alpha^*)^2\frac{1}{q_n}(U_{EU}-U_{\bb})\\
&+&\alpha^*(1-\alpha^*)\frac{1}{q_n}\left(2R_{GMV}+\gamma^{-1} s - \gamma V_{GMV}-R_{GMV}-\frac{\gamma^{-1}}{1-c} s-R_b+\gamma V_{GMV}+\frac{R_b-R_{GMV}}{1-c}\right)\\
&=&(\alpha^*)^2\frac{1}{q_n}(U_{EU}-U_S)+(1-\alpha^*)^2\frac{1}{q_n}(U_{EU}-U_{\bb})
+\alpha^*(1-\alpha^*)\frac{c}{1-c}\frac{1}{q_n}(R_b-R_{GMV}-\gamma^{-1} s).
\end{eqnarray*}
}
$\text{for}~ p/n\longrightarrow c \in (0,1) ~\text{as} ~n\rightarrow\infty$.
\end{proof}

\vspace{0.5cm}
For the proof of Theorem 2.2 we need several results about the properties of Moore-Penrose inverse which are summarized in the following three lemmas. Similarly as the proof of Lemma \ref{lem2}, we will use Lemma \ref{lem1} and Theorem \ref{weierstrass} in a sequel every time a derivative must be interchanged with the limit $n\to\infty$.

\begin{lemma}\label{lem5}
Assume (A2). Let $\boldsymbol{\theta}$ and $\boldsymbol{\xi}$ be universal nonrandom vectors with bounded Euclidean norms. Then it holds that
\begin{eqnarray}
\boldsymbol{\xi}^\prime\tbV_n^+\boldsymbol{\theta} &\stackrel{a.s.}{\longrightarrow}&  c^{-1}(c-1)^{-1} \boldsymbol{\xi}^\prime \boldsymbol{\theta} \,,\label{5_1}\\
\boldsymbol{\xi}^\prime(\tbV_n^+)^{2}\boldsymbol{\theta} &\stackrel{a.s.}{\longrightarrow}&  (c-1)^{-3} \boldsymbol{\xi}^\prime \boldsymbol{\theta} \,,\label{5_2}\\
\sx^\prime\tbV_n^+\sx_n &=&1   \,,\label{5_3}\\
\sx^\prime(\tbV_n^+)^{2}\sx_n &\stackrel{a.s.}{\longrightarrow}& \frac{c}{c-1}  \,,\label{5_4}\\
\sx^\prime(\tbV_n^+)^{3}\sx_n &\stackrel{a.s.}{\longrightarrow}& \frac{c^3}{(c-1)^3} \,,\label{5_5}\\
\sx^\prime(\tbV_n^+)^{4}\sx_n &\stackrel{a.s.}{\longrightarrow}&\frac{c^4 +c^5}{(c-1)^5} \,,\label{5_6}\\
\sx^\prime\tbV_n^+\boldsymbol{\theta} &\stackrel{a.s.}{\longrightarrow}&  0  \,,\label{5_7}\\
\sx^\prime(\tbV_n^+)^{2}\boldsymbol{\theta} &\stackrel{a.s.}{\longrightarrow}&  0  \,,\label{5_8}\\
\sx^\prime(\tbV_n^+)^{3}\boldsymbol{\theta} &\stackrel{a.s.}{\longrightarrow}&  0  \label{5_9}
\end{eqnarray}
$\text{for}~ p/n\longrightarrow c \in (1, +\infty) ~\text{as} ~n\rightarrow\infty ~~$.
\end{lemma}

\begin{proof}[Proof of Lemma \ref{lem5}:]
It holds that
\[\tbV^+=\left(\frac{1}{n}\bx_n\bx_n^\prime\right)^+=\frac{1}{\sqrt{n}}\bx_n\left(\frac{1}{n}\bx_n^\prime\bx_n\right)^{-2}\frac{1}{\sqrt{n}}\bx_n^\prime\]
and, similarly,
\[(\tbV^+)^i=\frac{1}{\sqrt{n}}\bx_n\left(\frac{1}{n}\bx_n^\prime\bx_n\right)^{-(i+1)}\frac{1}{\sqrt{n}}\bx_n^\prime ~~\text{for} ~~ i=2,3,4.\]

Let $\mathbf{\Theta}=\boldsymbol{\theta}\boldsymbol{\xi}^\prime$. It holds that
\begin{eqnarray*}
\boldsymbol{\xi}^\prime\tbV_n^+\boldsymbol{\theta}&=&\text{tr}\left[\frac{1}{\sqrt{n}}\bx_n\left(\frac{1}{n}\bx_n^\prime\bx_n\right)^{-2}\frac{1}{\sqrt{n}}\bx_n^\prime\mathbf{\Theta}\right]
=\left.\dfrac{\partial}{\partial z}\text{tr}\left[\frac{1}{\sqrt{n}}\bx_n\left(\frac{1}{n}\bx_n^\prime\bx_n-z\bI_n\right)^{-1}\frac{1}{\sqrt{n}}\bx_n^\prime\mathbf{\Theta}\right]\right|_{z=0}\,,\\
\boldsymbol{\xi}^\prime(\tbV_n^+)^{2}\boldsymbol{\theta}&=&\text{tr}\left[\frac{1}{\sqrt{n}}\bx_n\left(\frac{1}{n}\bx_n^\prime\bx_n\right)^{-3}\frac{1}{\sqrt{n}}\bx_n^\prime\mathbf{\Theta}\right]
=\frac{1}{2}\left.\dfrac{\partial^2}{\partial z^2}\text{tr}\left[\frac{1}{\sqrt{n}}\bx_n\left(\frac{1}{n}\bx_n^\prime\bx_n-z\bI_n\right)^{-1}\frac{1}{\sqrt{n}}\bx_n^\prime\mathbf{\Theta}\right]\right|_{z=0}\,.
\end{eqnarray*}

The application of Woodbury formula (matrix inversion lemma, see, e.g., \cite{hornjohn1985}),
\begin{eqnarray}\label{wf}
\frac{1}{\sqrt{n}}\bx_n\left(\frac{1}{n}\bx_n^\prime\bx_n-z\bI_n\right)^{-1}\frac{1}{\sqrt{n}}\bx_n^\prime
=\bI_p+z\left(\frac{1}{n}\bx_n\bx_n^\prime-z\bI_p\right)^{-1}
\end{eqnarray}
leads to
\begin{eqnarray*}
\boldsymbol{\xi}^\prime \tbV_n^+\boldsymbol{\theta}&=&\left.\dfrac{\partial}{\partial z}z\text{tr}\left[\left(\frac{1}{n}\bx_n\bx_n^\prime-z\bI_p\right)^{-1}\mathbf{\Theta}\right]\right|_{z=0},\\
\boldsymbol{\xi}^\prime(\tbV_n^+)^{2}\boldsymbol{\theta}&=& \frac{1}{2}\left.\dfrac{\partial^2}{\partial z^2}z\text{tr}\left[\left(\frac{1}{n}\bx_n\bx_n^\prime-z\bI_p\right)^{-1}\mathbf{\Theta}\right]\right|_{z=0}\,.
\end{eqnarray*}

From the proof of Lemma \ref{lem2} we know that the matrix $\mathbf{\Theta} $ possesses the bounded trace norm. Then the application of Lemma \ref{lem1} leads to
\begin{eqnarray*}
\boldsymbol{\xi}^\prime\tbV_n^+\boldsymbol{\theta}&\stackrel{a.s.}{\longrightarrow}&\left.\dfrac{\partial}{\partial z}\frac{z}{x(z)-z}\right|_{z=0}\boldsymbol{\xi}^\prime\boldsymbol{\theta}\,,\label{th12_eq1}\\
\boldsymbol{\xi}^\prime(\tbV_n^+)^{2}\boldsymbol{\theta}&\stackrel{a.s.}{\longrightarrow}& \frac{1}{2}\left.\dfrac{\partial^2}{\partial z^2}\frac{z}{x(z)-z}\right|_{z=0}\boldsymbol{\xi}^\prime\boldsymbol{\theta}\label{th12_eq3}
\end{eqnarray*}
for $p/n\rightarrow c>1$ as $n\rightarrow\infty$, where $x(z)$ is given in \eqref{RM2011_id_xz}.

Let us make the following notations
\begin{equation*}\label{not}
\theta(z)=\dfrac{z}{x(z)-z} \quad \text{and} \quad \phi(z)=\dfrac{x(z)-zx^\prime(z)}{z^2}\,.
\end{equation*}
Then the first and the second derivatives of $\theta(z)$ are given by
\begin{equation}\label{not_der}
\theta^{\prime}(z)=\theta^2(z)\phi(z)\quad \text{and} \quad \theta^{''}(z)=2\theta(z)\theta^\prime(z)\phi(z)+\theta^2(z)\phi^\prime(z)\,.
\end{equation}

Using L'Hopital's rule, we get
\begin{equation}\label{tetalim}
\theta(0)=\lim\limits_{z\rightarrow0}\theta(z)=\lim\limits_{z\rightarrow0}\dfrac{z}{x(z)-z}
=\lim\limits_{z\rightarrow0}\dfrac{1}{(x^\prime(z)-1)}=\dfrac{1}{\dfrac{1}{2}\left(1-\dfrac{1+c}{|1-c|}\right)-1}=-\dfrac{c-1}{c}\,,
\end{equation}
\begin{equation}\label{philim}
\phi(0)=\lim\limits_{z\rightarrow0}\phi(z)=\lim\limits_{z\rightarrow0}\dfrac{x(z)-zx^\prime(z)}{z^2}=-\dfrac{1}{2}\lim\limits_{z\rightarrow0}x^{''}(z)
=-\dfrac{1}{2}\lim\limits_{z\rightarrow0}\frac{-2c}{((1-c+z)^2-4z)^{3/2}}
=\dfrac{c}{(c-1)^3}\,,
\end{equation}
and
\begin{eqnarray}
\lim\limits_{z\rightarrow0}\phi^\prime(z)&=&-\lim\limits_{z\rightarrow0}\dfrac{2(x(z)-zx^\prime(z))+z^2x^{''}(z)}{z^2}\nonumber\\
&=&-\lim\limits_{z\rightarrow0}\dfrac{2\phi(z)+x^{''}(z)}{z}=-\lim\limits_{z\rightarrow0}(2\phi^\prime(z)+x^{'''}(z))\,,
\end{eqnarray}
which implies
\begin{equation}\label{phi1lim2}
\phi^\prime(0)=\lim\limits_{z\rightarrow0}\phi^\prime(z)
=-\dfrac{1}{3}\lim\limits_{z\rightarrow0}x^{'''}(z)=-\dfrac{1}{3}\lim\limits_{z\rightarrow0}\frac{6c(z-c-1)}{((1-c+z)^2-4z)^{5/2}}=\dfrac{2c(c+1)}{(c-1)^5}
\,.
\end{equation}

Combining (\ref{not_der}), (\ref{tetalim}), (\ref{philim}), and (\ref{phi1lim2}), we get
\begin{equation*}\label{num11a}
\theta^{'}(0)=\lim\limits_{z\rightarrow0}\theta^{'}(z)=\theta^2(0)\phi(0)=\dfrac{1}{c(c-1)}
\end{equation*}
and
\begin{equation*}\label{num11}
\theta^{''}(0)=\lim\limits_{z\rightarrow0}\theta^{''}(z)=2\theta^3(0)\phi^2(0)+\theta^2(0)\phi^\prime(0)=\dfrac{2}{(c-1)^3}\,.
\end{equation*}
Hence,
\begin{eqnarray*}
\boldsymbol{\xi}^\prime\tbV_n^+\boldsymbol{\theta}&\stackrel{a.s.}{\longrightarrow}& \dfrac{1}{c(c-1)}\bxi^\prime\bSigma_n^{-1}\btheta ~~\text{for}~p/n\rightarrow c>1~~\text{as}~n\rightarrow\infty,\\
\boldsymbol{\xi}^\prime(\tbV_n^+)^{2}\boldsymbol{\theta}&\stackrel{a.s.}{\longrightarrow}& \dfrac{1}{(c-1)^3}\bxi^\prime\bSigma_n^{-1}\btheta~~\text{for}~p/n\rightarrow c>1~~\text{as}~n\rightarrow\infty\,.
\end{eqnarray*}

Taking into account that
\begin{equation*}
\sx_n^\prime\tbV_n^{+}\sx_n = \dfrac{1}{n}\bi_n^\prime\bx^\prime_n\bx_n(\bx_n^\prime\bx_n)^{-2}\bx_n^\prime\bx_n\bi_n = \dfrac{1}{n}\bi_n^\prime\bi_n = 1\,.
\end{equation*}
we get \eqref{5_3}. Similarly, using
\begin{eqnarray*}
 \sx_n^\prime(\tbV_n^+)^i\sx_n&=& 1/n\bi_n^\prime(1/n\bx_n^\prime\bx_n)^{-(i-1)}\bi_n ~~ \text{for} ~~ i=2,3,4
\end{eqnarray*}
we get
\begin{eqnarray*}
 1/n\bi_n^\prime(1/n\bx_n^\prime\bx_n)^{-1}\bi_n&=&\lim\limits_{z\rightarrow0} \text{tr}[(1/n\bx_n^\prime\bx_n-z\bI)^{-1}\bTheta_n]\stackrel{a.s.}{\longrightarrow} \underline{m}(0) \,, \\
 1/n\bi_n^\prime(1/n\bx_n^\prime\bx_n)^{-2}\bi_n &=&\lim\limits_{z\rightarrow0}\dfrac{\partial}{\partial z} \text{tr}[(1/n\bx_n^\prime\bx_n-z\bI)^{-1}\bTheta_n]\stackrel{a.s.}{\longrightarrow}\underline{m}^\prime(0) \,,\\
 1/n\bi_n^\prime(1/n\bx_n^\prime\bx_n)^{-3}\bi_n &=&\frac{1}{2}\lim\limits_{z\rightarrow0}\dfrac{\partial^2}{\partial z^2} \text{tr}[(1/n\bx_n^\prime\bx_n-z\bI)^{-2}\bTheta_n]\stackrel{a.s.}{\longrightarrow}\frac{1}{2}\underline{m}^{\prime\prime}(0)
\end{eqnarray*}
for $p/n\rightarrow c>1$ as $n\rightarrow\infty$, where $\bTheta_n= 1/n\bi_n\bi_n^\prime$.

Using that the elements of $\bx_n$ are independent and identically distributed and the fact that $n<p$ from Lemma \ref{lem1} and Theorem \ref{weierstrass} we get that
\begin{equation*}
\underline{m}^\prime(z)=\frac{1}{\underline{x}(z)-z}~~\text{with}~~\underline{x}(z)=\dfrac{1}{2}\left(1-c^{-1}+z+\sqrt{(1-c^{-1}+z)^2-4z}\right).
\end{equation*}

Thus,
\[\underline{m}(0)=\frac{1}{\underline{x}(0)}=\frac{1}{1-c^{-1}}=\frac{c}{c-1},\]
which proves \eqref{5_4}. Furthermore, we get
\[\underline{m}^\prime(z)=-\frac{\underline{x}^\prime(z)-1}{(\underline{x}(z)-z)^2}\]
with
\[\underline{x}^\prime(z)=\dfrac{1}{2}\left(1+\frac{1}{2}\frac{2(1-c^{-1}+z)-4}{\sqrt{(1-c^{-1}+z)^2-4z}}\right)
=\dfrac{1}{2}\left(1+\frac{-1-c^{-1}+z}{\sqrt{(1-c^{-1}+z)^2-4z}}\right),
\]
and, consequently,
\[\underline{m}^\prime(0)=-\frac{\underline{x}^\prime(0)-1}{\underline{x}(0)^2}=\frac{1}{(1-c^{-1})^3}=\frac{c^3}{(c-1)^3}.\]

In order to prove \eqref{5_6}, we compute
\[\underline{m}^{\prime\prime}(z)=-\frac{\underline{x}^{\prime\prime}(z)}{(\underline{x}(z)-z)^2}+2\frac{(\underline{x}^\prime(z)-1)^2}{(\underline{x}(z)-z)^3}\]
where
\[\underline{x}^{\prime\prime}(z)=\dfrac{1}{2}\left(\frac{1}{\sqrt{(1-c^{-1}+z)^2-4z}}-\frac{(-1-c^{-1}+z)^2}{((1-c^{-1}+z)^2-4z)^{3/2}}
\right),
\]
Hence,
\[\underline{m}^{\prime\prime}(0)=
-\frac{\underline{x}^{\prime\prime}(0)}{\underline{x}(0)^2}+2\frac{(\underline{x}^\prime(0)-1)^2}{\underline{x}(0)^3}
=2\frac{c^{-1}+1}{(1-c^{-1})^5}=2\frac{c^4+c^5}{(c-1)^5}.\]

For \eqref{5_7} we consider
\begin{eqnarray*}
\sx^\prime\tbV_n^+\boldsymbol{\theta} &=& \text{tr}\left[\bI_p+z\left(\frac{1}{n}\bx_n\bx_n^\prime-z\bI_p\right)^{-1}\btheta\sx^\prime \right] = \sx^\prime\btheta +
z\sx^\prime\left(\tbV_n-z\bI_p\right)^{-1}\btheta \,.
\end{eqnarray*}
Because of \eqref{3}, it holds that $\sx^\prime\left(\tbV_n-z\bI_p\right)^{-1}\btheta$ is uniformly bounded as $z \longrightarrow 0$. Moreover, $\sx^\prime\btheta \stackrel{a.s.}{\longrightarrow}  0$ as $p \longrightarrow\infty$ following Kolmogorov's strong law of large numbers (c.f., Sen and Singer (1993, Theorem 2.3.10), since $\btheta$ has a bounded Euclidean norm. Hence, $\sx^\prime\tbV_n^+\boldsymbol{\theta} \stackrel{a.s.}{\longrightarrow}  0$ $\text{for}~ p/n\longrightarrow c \in (1,+\infty) ~\text{as} ~n\rightarrow\infty$.

Finally, in the case of \eqref{5_8} and \eqref{5_9}, we get
\begin{eqnarray*}
\sx^\prime(\tbV_n^+)^{2}\boldsymbol{\theta} &=&\left.\dfrac{\partial}{\partial z}z\text{tr}\left[\left(\frac{1}{n}\bx_n\bx_n^\prime-z\bI_p\right)^{-1}\btheta\sx^\prime \right]\right|_{z=0} \stackrel{a.s.}{\longrightarrow}  0  \,,\\
\sx^\prime(\tbV_n^+)^{3}\boldsymbol{\theta} &=& \frac{1}{2}\left.\dfrac{\partial^2}{\partial z^2}z\text{tr}\left[\left(\frac{1}{n}\bx_n\bx_n^\prime-z\bI_p\right)^{-1}\btheta\sx^\prime \right]\right|_{z=0}\stackrel{a.s.}{\longrightarrow}  0.
\end{eqnarray*}
$\text{for}~ p/n\longrightarrow c \in (1,+\infty) ~\text{as} ~n\rightarrow\infty$.
\end{proof}

\vspace{0.5cm}
\begin{lemma}\label{lem6}
Assume (A2). Let $\boldsymbol{\theta}$ and $\boldsymbol{\xi}$ be universal nonrandom vectors with bounded Euclidean norms. Then it holds that
\begin{eqnarray}
\boldsymbol{\xi}^\prime\bV_n^{+}\boldsymbol{\theta} &\stackrel{a.s.}{\longrightarrow}& c^{-1}(c-1)^{-1} \boldsymbol{\xi}^\prime \boldsymbol{\theta}  \,,\label{6_1}\\
   \sx_n^\prime\bV_n^{+}\sx_n &\stackrel{a.s.}{\longrightarrow}& \frac{1}{c-1}  \,,\label{6_2}\\
 \sx_n^\prime\bV_n^{+}\boldsymbol{\theta}&\stackrel{a.s.}{\longrightarrow}&0  \,,\label{6_3}\\
\boldsymbol{\xi}^\prime(\bV_n^+)^{2}\boldsymbol{\theta} &\stackrel{a.s.}{\longrightarrow}& (c-1)^{-3} \boldsymbol{\xi}^\prime \boldsymbol{\theta} \,,\label{6_4}\\
   \sx_n^\prime(\bV_n^+)^{2}\sx_n &\stackrel{a.s.}{\longrightarrow}& \frac{c^2}{(c-1)^3}   \,,\label{6_5}\\
 \sx_n^\prime(\bV_n^+)^{2}\boldsymbol{\theta}&\stackrel{a.s.}{\longrightarrow}&  0\label{6_6}
\end{eqnarray}
$\text{for}~ p/n\longrightarrow c \in (1,+\infty) ~\text{as} ~n\rightarrow\infty$.
\end{lemma}

\begin{proof}[Proof of Lemma \ref{lem6}:]
From \eqref{bVn_MPinv} we get
\begin{eqnarray*}
\bxi^\prime\bV_n^{+}\btheta&=& \bxi^\prime\tbV_n^+ \btheta
-\dfrac{\bxi^\prime\tbV_n^+\sx_n\sx_n^\prime(\tbV_n^+)^2\btheta+\bxi^\prime(\tbV_n^+)^2\sx_n\sx_n^\prime
\tbV_n^+\btheta}{\sx_n^\prime(\tbV_n^+)^2\sx_n}\\
 &+&\dfrac{\sx_n^\prime(\tbV_n^+)^3\sx_n}
{(\sx_n^\prime(\tbV_n^+)^2\sx_n)^2}\bxi^\prime\tbV_n^+\sx_n\sx_n^\prime\tbV_n^+\btheta
\stackrel{a.s.}{\longrightarrow} c^{-1}(c-1)^{-1} \boldsymbol{\xi}^\prime \boldsymbol{\theta}
\end{eqnarray*}
$\text{for}~ p/n\longrightarrow c \in (1,+\infty) ~\text{as} ~n\rightarrow\infty$ following \eqref{5_1}-\eqref{5_3}. Similarly, we get
\begin{eqnarray*}
\sx_n^\prime\bV_n^{+}\sx_n&=& \sx_n^\prime\tbV_n^+ \sx_n
-\dfrac{\sx_n^\prime\tbV_n^+\sx_n\sx_n^\prime(\tbV_n^+)^2\sx_n+\sx_n^\prime(\tbV_n^+)^2\sx_n\sx_n^\prime
\tbV_n^+\sx_n}{\sx_n^\prime(\tbV_n^+)^2\sx_n}\\
 &+&\dfrac{\sx_n^\prime(\tbV_n^+)^3\sx_n}
{(\sx_n^\prime(\tbV_n^+)^2\sx_n)^2}\sx_n^\prime\tbV_n^+\sx_n\sx_n^\prime\tbV_n^+\sx_n
\stackrel{a.s.}{\longrightarrow} \frac{1}{c-1}
\end{eqnarray*}
and
\begin{eqnarray*}
\sx_n^\prime\bV_n^{+}\btheta&=& \sx_n^\prime\tbV_n^+ \btheta
-\dfrac{\sx_n^\prime\tbV_n^+\sx_n\sx_n^\prime(\tbV_n^+)^2\btheta+\sx_n^\prime(\tbV_n^+)^2\sx_n\sx_n^\prime
\tbV_n^+\btheta}{\sx_n^\prime(\tbV_n^+)^2\sx_n}\\
 &+&\dfrac{\sx_n^\prime(\tbV_n^+)^3\sx_n}
{(\sx_n^\prime(\tbV_n^+)^2\sx_n)^2}\sx_n^\prime\tbV_n^+\sx_n\sx_n^\prime\tbV_n^+\btheta
\stackrel{a.s.}{\longrightarrow} 0
\end{eqnarray*}
$\text{for}~ p/n\longrightarrow c \in (1,+\infty) ~\text{as} ~n\rightarrow\infty$.

Now, we consider the equality
{\footnotesize
\begin{eqnarray*}
\left[(\tbV_n-\sx_n\sx_n^\prime)^+\right]^2 &=&  \left(\tbV_n^+-\dfrac{\tbV_n^+\sx_n\sx_n^\prime(\tbV_n^+)^2+(\tbV_n^+)^2\sx_n\sx_n^\prime(\tbV_n^+)}{\sx_n^\prime(\tbV_n^+)^2\sx_n} +\dfrac{\sx_n^\prime(\tbV_n^+)^3\sx_n}{(\sx_n^\prime(\tbV_n^+)^2\sx_n)^2}\tbV_n^+\sx_n\sx_n^\prime\tbV_n^+\right)^2\nonumber\\
&=& (\tbV_n^+)^2+ \left[ \dfrac{\sx_n^\prime(\tbV_n^+)^3\sx_n}{(\sx_n^\prime(\tbV_n^+)^2\sx_n)^2}\tbV_n^+\sx_n\sx_n^\prime\tbV_n^+ - \dfrac{(\tbV_n^+)^2\sx_n\sx_n^\prime\tbV_n^++\tbV_n^+\sx_n\sx_n^\prime(\tbV_n^+)^2}{\sx_n^\prime(\tbV_n^+)^2\sx_n}\right]^2\\
&-&\dfrac{2(\tbV_n^+)^2\sx_n\sx_n^\prime(\tbV_n^+)^2+\tbV_n^+\sx_n\sx_n^\prime(\tbV_n^+)^3+(\tbV_n^+)^3\sx_n\sx_n^\prime\tbV_n^+}{\sx_n^\prime(\tbV_n^+)^2\sx_n}\\
&+&\dfrac{\sx_n^\prime(\tbV_n^+)^3\sx_n}{(\sx_n^\prime(\tbV_n^+)^2\sx_n)^2}(\tbV_n^+\sx_n\sx_n^\prime(\tbV_n^+)^2+(\tbV_n^+)^2\sx_n\sx_n^\prime\tbV_n^+)\\
&=& (\tbV_n^+)^2 + \dfrac{\sx_n^\prime(\tbV_n^+)^4\sx_n}{(\sx_n^\prime(\tbV_n^+)^2\sx_n)^2}\tbV_n^+\sx_n\sx_n^\prime\tbV_n^+-\dfrac{(\sx_n^\prime(\tbV_n^+)^3\sx_n)^2}{(\sx_n^\prime(\tbV_n^+)^2\sx_n)^3}\tbV_n^+\sx_n\sx_n^\prime\tbV_n^+\\
&+& \dfrac{\sx_n^\prime(\tbV_n^+)^3\sx_n}{(\sx_n^\prime(\tbV_n^+)^2\sx_n)^2}\left[(\tbV_n^+)^2\sx_n\sx_n^\prime\tbV_n^++\tbV_n^+\sx_n\sx_n^\prime(\tbV_n^+)^2\right]\\
&-& \dfrac{(\tbV_n^+)^2\sx_n\sx_n^\prime(\tbV_n^+)^2+\tbV_n^+\sx_n\sx_n^\prime(\tbV_n^+)^3+(\tbV_n^+)^3\sx_n\sx_n^\prime\tbV_n^+}{\sx_n^\prime(\tbV_n^+)^2\sx_n}
\end{eqnarray*}
}

Hence,
\begin{eqnarray*}
\bxi^\prime(\bV_n^{+})^2\btheta&=& \bxi^\prime(\tbV_n^+)^2 \btheta
+\dfrac{\sx_n^\prime(\tbV_n^+)^4\sx_n}{(\sx_n^\prime(\tbV_n^+)^2\sx_n)^2} \bxi^\prime\tbV_n^+\sx_n\sx_n^\prime\tbV_n^+\btheta
-\dfrac{(\sx_n^\prime(\tbV_n^+)^3\sx_n)^2}{(\sx_n^\prime(\tbV_n^+)^2\sx_n)^3} \bxi^\prime\tbV_n^+\sx_n\sx_n^\prime\tbV_n^+\btheta\\
&+& \dfrac{\sx_n^\prime(\tbV_n^+)^3\sx_n}{(\sx_n^\prime(\tbV_n^+)^2\sx_n)^2} \left[\bxi^\prime(\tbV_n^+)^2\sx_n\sx_n^\prime\tbV_n^+\btheta
+\bxi^\prime\tbV_n^+\sx_n\sx_n^\prime(\tbV_n^+)^2\btheta\right]\\
&-& \dfrac{\bxi^\prime(\tbV_n^+)^2\sx_n\sx_n^\prime(\tbV_n^+)^2\btheta +\bxi^\prime\tbV_n^+\sx_n\sx_n^\prime(\tbV_n^+)^3 \btheta +\bxi^\prime(\tbV_n^+)^3\sx_n\sx_n^\prime\tbV_n^+\btheta}{\sx_n^\prime(\tbV_n^+)^2\sx_n}
\stackrel{a.s.}{\longrightarrow} (c-1)^{-3} \boldsymbol{\xi}^\prime \boldsymbol{\theta} \,,
\end{eqnarray*}
\begin{eqnarray*}
 \sx_n^\prime(\bV_n^+)^2\sx_n& =& \sx_n^\prime(\tbV_n^+)^2\sx_n+\dfrac{\sx_n^\prime(\tbV_n^+)^4\sx_n}{(\sx_n^\prime(\tbV_n^+)^2\sx_n)^2}-\dfrac{(\sx_n^\prime(\tbV_n^+)^3\sx_n)^2}{(\sx_n^\prime(\tbV_n^+)^2\sx_n)^3}+2\dfrac{\sx_n^\prime(\tbV_n^+)^3\sx_n}{\sx_n^\prime(\tbV_n^+)^2\sx_n}\nonumber\\
 &-& \sx_n^\prime(\tbV_n^+)^2\sx_n-2\dfrac{\sx_n^\prime(\tbV_n^+)^3\sx_n}{\sx_n^\prime(\tbV_n^+)^2\sx_n}\nonumber\\
 &=& \dfrac{\sx_n^\prime(\tbV_n^+)^4\sx_n}{(\sx_n^\prime(\tbV_n^+)^2\sx_n)^2}-\dfrac{(\sx_n^\prime(\tbV_n^+)^3\sx_n)^2}{(\sx_n^\prime(\tbV_n^+)^2\sx_n)^3}
 \stackrel{a.s.}{\longrightarrow} \frac{c^2}{(c-1)^3}\,,
\end{eqnarray*}
and
\begin{eqnarray*}
\sx_n^\prime(\bV_n^{+})^2\btheta&=& \sx_n^\prime(\tbV_n^+)^2 \btheta
+\dfrac{\sx_n^\prime(\tbV_n^+)^4\sx_n}{(\sx_n^\prime(\tbV_n^+)^2\sx_n)^2} \sx_n^\prime\tbV_n^+\btheta
-\dfrac{(\sx_n^\prime(\tbV_n^+)^3\sx_n)^2}{(\sx_n^\prime(\tbV_n^+)^2\sx_n)^3} \sx_n^\prime\tbV_n^+\btheta\\
&+& \dfrac{\sx_n^\prime(\tbV_n^+)^3\sx_n}{(\sx_n^\prime(\tbV_n^+)^2\sx_n)^2} \left[\sx_n^\prime(\tbV_n^+)^2\sx_n\sx_n^\prime\tbV_n^+\btheta
+\sx_n^\prime(\tbV_n^+)^2\btheta\right]\\
&-& \dfrac{\sx_n^\prime(\tbV_n^+)^2\sx_n\sx_n^\prime(\tbV_n^+)^2\btheta +\sx_n^\prime(\tbV_n^+)^3 \btheta +\sx_n^\prime(\tbV_n^+)^3\sx_n\sx_n^\prime\tbV_n^+\btheta}{\sx_n^\prime(\tbV_n^+)^2\sx_n}
\stackrel{a.s.}{\longrightarrow} 0
\end{eqnarray*}
$\text{for}~ p/n\longrightarrow c \in (1,+\infty) ~\text{as} ~n\rightarrow\infty$.
\end{proof}

\vspace{0.5cm}
\begin{lemma}\label{lem7}
Assume (A2). Let $\boldsymbol{\theta}$ and $\boldsymbol{\xi}$ be universal nonrandom vectors with bounded Euclidean norms and let $\bP_n^+=\bV_n^{+}-\frac{\bV_n^{+}\boldsymbol{\eta}\boldsymbol{\eta}^\prime\bV_n^{+}}{\boldsymbol{\eta}^\prime\bV_n^{+}\boldsymbol{\eta}}$ where $\boldsymbol{\eta}$ is a universal nonrandom vectors with bounded Euclidean norm. Then it holds that
\begin{eqnarray}
  \boldsymbol{\xi}^\prime\bP_n^+\boldsymbol{\theta} &\stackrel{a.s.}{\longrightarrow}& c^{-1}(c-1)^{-1}\left(\boldsymbol{\xi}^\prime \boldsymbol{\theta}-\frac{\bxi^\prime \boldsymbol{\eta}\boldsymbol{\eta}^\prime \btheta}{\boldsymbol{\eta}^\prime \boldsymbol{\eta}}\right) \,,\label{7_1}\\
   \sx_n^\prime\bP_n^+\sx_n &\stackrel{a.s.}{\longrightarrow}& \frac{1}{c-1}  \,,\label{7_2}\\
 \sx_n^\prime\bP_n^+\boldsymbol{\theta}&\stackrel{a.s.}{\longrightarrow}&0  \,,\label{7_3}\\
  \boldsymbol{\xi}^\prime(\bP_n^+)^{2}\boldsymbol{\theta} &\stackrel{a.s.}{\longrightarrow}&  (c-1)^{-3} \left(\boldsymbol{\xi}^\prime \boldsymbol{\theta}-\frac{\bxi^\prime \boldsymbol{\eta}\boldsymbol{\eta}^\prime \btheta}{\boldsymbol{\eta}^\prime \boldsymbol{\eta}}\right) \,,\label{7_4}\\
   \sx_n^\prime(\bP_n^+)^{2}\sx_n &\stackrel{a.s.}{\longrightarrow}& \frac{c^2}{(c-1)^3}   \,,\label{7_5}\\
 \sx_n^\prime(\bP_n^+)^2\boldsymbol{\theta}&\stackrel{a.s.}{\longrightarrow}&0 \label{7_6}
\end{eqnarray}
$\text{for}~ p/n\longrightarrow c \in (1,+\infty) ~\text{as} ~n\rightarrow\infty$.
\end{lemma}

\begin{proof}[Proof of Lemma \ref{lem7}:]
It holds that
\begin{eqnarray*}
\boldsymbol{\xi}^\prime\bP_n^+\boldsymbol{\theta}&=&\bxi^\prime\bV_n^{+}\btheta-\frac{\bxi^\prime\bV_n^{+}\boldsymbol{\eta}\boldsymbol{\eta}^\prime\bV_n^{+}\btheta}
{\boldsymbol{\eta}^\prime\bV_n^{+}\boldsymbol{\eta}}
\stackrel{a.s.}{\longrightarrow} c^{-1}(c-1)^{-1} \left(\boldsymbol{\xi}^\prime \boldsymbol{\theta}-\frac{\bxi^\prime \boldsymbol{\eta}\boldsymbol{\eta}^\prime \btheta}{\boldsymbol{\eta}^\prime \boldsymbol{\eta}}\right)
\end{eqnarray*}
$\text{for}~ p/n\longrightarrow c \in (1,+\infty) ~\text{as} ~n\rightarrow\infty$ following \eqref{5_1}. Similarly, we get
\begin{eqnarray*}
\sx_n^\prime\bP_n^+\sx_n&=&\sx_n^\prime\bV_n^{+}\sx_n-\frac{\sx_n^\prime\bV_n^{+}\boldsymbol{\eta}\boldsymbol{\eta}^\prime\bV_n^{+}\sx_n}
{\boldsymbol{\eta}^\prime\bV_n^{+}\boldsymbol{\eta}}
\stackrel{a.s.}{\longrightarrow} \frac{1}{c-1}
\end{eqnarray*}
and
\begin{eqnarray*}
\sx_n^\prime\bP_n^+\boldsymbol{\theta}&=&\sx_n^\prime\bV_n^{+}\btheta-\frac{\sx_n^\prime\bV_n^{+}\boldsymbol{\eta}\boldsymbol{\eta}^\prime\bV_n^{+}\btheta}
{\boldsymbol{\eta}^\prime\bV_n^{+}\boldsymbol{\eta}}
\stackrel{a.s.}{\longrightarrow}0
\end{eqnarray*}
$\text{for}~ p/n\longrightarrow c \in (1,+\infty) ~\text{as} ~n\rightarrow\infty$.

The rest of the proof follows from the equality
\[(\bP_n^+)^2=(\bV_n^+)^{2}-\frac{(\bV_n^+)^{2}\boldsymbol{\eta}\boldsymbol{\eta}^\prime\bV_n^{+}}{\boldsymbol{\eta}^\prime\bV_n^{+}\boldsymbol{\eta}}
-\frac{\bV_n^{+}\boldsymbol{\eta}\boldsymbol{\eta}^\prime(\bV_n^+)^{2}}{\boldsymbol{\eta}^\prime\bV_n^+\boldsymbol{\eta}}
+\boldsymbol{\eta}^\prime(\bV_n^+)^{2}\boldsymbol{\eta}\frac{\bV_n^{+}\boldsymbol{\eta}\boldsymbol{\eta}^\prime\bV_n^{+}}{(\boldsymbol{\eta}^\prime\bV_n^{+}\boldsymbol{\eta})^2}
\]
and Lemma \ref{lem6}.
\end{proof}

\vspace{0.5cm}

\begin{proof}[Proof of Theorem 2.2:]
Let $q_n=\max\{\bm_n^\prime\bSigma_n^{-1}\bm_n,\mathbf{b}^\prime\bSigma_n\mathbf{b}\}$. From Assumption (A3) get that $q_n>0$ uniformly in $p$, since $\bm_n^\prime\bSigma_n^{-1}\bm_n \ge s$ and $s>0$ uniformly in $p$.

In case of $c>1$, the optimal shrinkage intensity is given by
{\footnotesize
 \begin{eqnarray*}\label{alfa_app1}
  \alpha_n^+&=& \beta^{-1}\dfrac{\hat{\bw}^\prime_{S^*}(\bm_n-\beta\bSigma_n\mathbf{b})-\mathbf{b}^\prime(\bm_n-\beta\bSigma_n\mathbf{b})}
  {\hat{\bw}_{S^*}^\prime\bSigma_n\hat{\bw}_{S^*}-2\mathbf{b}^\prime\bSigma_n\hat{\bw}_{S^*}+\mathbf{b}^\prime\bSigma_n\mathbf{b}}\\
  &=& \beta^{-1}\dfrac{\dfrac{\bi^\prime\bS_n^{*}(\bm_n-\beta\bSigma_n\mathbf{b})}{\bi^\prime\bS_n^{*}\bi}+\gamma^{-1}\sy_n^\prime\hat{\bQ}_n^*(\bm_n-\beta\bSigma_n\mathbf{b})
  -\mathbf{b}^\prime(\bm_n-\beta\bSigma_n\mathbf{b})}{\dfrac{\bi^\prime\bS_n^{*}\bSigma_n\bS_n^{*}\bi}{(\bi^\prime\bS_n^{*}\bi)^2}
  +2\gamma^{-1}\dfrac{\sy_n^\prime\hat{\bQ}_n^*\bSigma_n\bS_n^{*}\bi}{\bi^\prime\bS_n^{*}\bi}+\gamma^{-2}\sy_n^\prime\hat{\bQ}_n^*\bSigma_n\hat{\bQ}_n^*\sy_n
  -2\dfrac{\mathbf{b}^\prime\bSigma_n\bS_n^{*}\bi}{\bi^\prime\bS_n^{*}\bi}-2\gamma^{-1}\mathbf{b}^\prime\bSigma_n\hat{\bQ}_n^*\sy_n+\mathbf{b}^\prime\bSigma_n\mathbf{b}}\\
 &=&\beta^{-1}\frac{A_n^+}{B_n^+},
 \end{eqnarray*}
}
where
\begin{eqnarray*}
  A_n^+ &=&\frac{1}{q_n}\left(\dfrac{\bi^\prime\bS_n^{*}(\bm_n-\beta\bSigma_n\mathbf{b})}{\bi^\prime\bS_n^{*}\bi}+\gamma^{-1}\sy_n^\prime\hat{\bQ}_n^*(\bm_n-\beta\bSigma_n\mathbf{b})
  -\mathbf{b}^\prime(\bm_n-\beta\bSigma_n\mathbf{b})\right)\\
  &=&\frac{(\bi^\prime\bSigma_n^{-1}\bi)^{-1/2}\sqrt{\bm_n^\prime\bSigma_n^{-1}\bm_n}}{q_n} \frac{(\bi^\prime\bSigma_n^{-1}\bi)^{-1/2}(\bm_n^\prime\bSigma_n^{-1}\bm_n)^{-1/2}\bi^\prime\bS_n^{*}\bm_n}
  {(\bi^\prime\bSigma_n^{-1}\bi)^{-1}\bi^\prime\bS_n^{*}\bi}\\
  &-&\beta \frac{(\bi^\prime\bSigma_n^{-1}\bi)^{-1/2}\sqrt{\mathbf{b}^\prime\bSigma_n\mathbf{b}}}{q_n} \frac{(\bi^\prime\bSigma_n^{-1}\bi)^{-1/2}(\mathbf{b}^\prime\bSigma_n\mathbf{b})^{-1/2}\bi^\prime\bS_n^{*}\bSigma_n\mathbf{b}}
  {(\bi^\prime\bSigma_n^{-1}\bi)^{-1}\bi^\prime\bS_n^{*}\bi}\\
  &+&\gamma^{-1}\frac{\bm_n^\prime\bSigma_n^{-1}\bm_n}{q_n}\frac{\sy_n^\prime\hat{\bQ}_n^{*}\bm_n}{\bm_n^\prime\bSigma_n^{-1}\bm_n}
  -\beta\gamma^{-1}\frac{\sqrt{\bm_n^\prime\bSigma_n^{-1}\bm_n}\sqrt{\mathbf{b}^\prime\bSigma_n\mathbf{b}}}{q_n}
  \frac{\sy_n^\prime\hat{\bQ}_n^{*}\bSigma_n\mathbf{b}}{\sqrt{\bm_n^\prime\bSigma_n^{-1}\bm_n}\sqrt{\mathbf{b}^\prime\bSigma_n\mathbf{b}}}\\
  &-&\frac{\sqrt{\bm_n^\prime\bSigma_n^{-1}\bm_n}\sqrt{\mathbf{b}^\prime\bSigma_n\mathbf{b}}}{q_n}\frac{\mathbf{b}^\prime\bm_n}
  {\sqrt{\bm_n^\prime\bSigma_n^{-1}\bm_n}\sqrt{\mathbf{b}^\prime\bSigma_n\mathbf{b}}} +\beta\frac{\mathbf{b}^\prime\bSigma_n\mathbf{b}}{q_n}
\end{eqnarray*}
and
\begin{eqnarray*}
   B_n^+&=&\frac{1}{q_n}\Bigg(\dfrac{\bi^\prime\bS_n^{*}\bSigma_n\bS_n^{*}\bi}{(\bi^\prime\bS_n^{*}\bi)^2}
  +2\gamma^{-1}\dfrac{\sy_n^\prime\hat{\bQ}_n^*\bSigma_n\bS_n^{*}\bi}{\bi^\prime\bS_n^{*}\bi}+\gamma^{-2}\sy_n^\prime\hat{\bQ}_n^*\bSigma_n\hat{\bQ}_n^*\sy_n\\
  &-&2\dfrac{\mathbf{b}^\prime\bSigma_n\bS_n^{*}\bi}{\bi^\prime\bS_n^{*}\bi}-2\gamma^{-1}\mathbf{b}^\prime\bSigma_n\hat{\bQ}_n^*\sy_n+\mathbf{b}^\prime\bSigma_n\mathbf{b} \Bigg)\\
   &=&\frac{(\bi^\prime\bSigma_n^{-1}\bi)^{-1}}{q_n}
   \frac{(\bi^\prime\bSigma_n^{-1}\bi)^{-1}\bi^\prime\bS_n^{*}\bSigma_n\bS_n^{*}\bi}{(\bi^\prime\bSigma_n^{-1}\bi)^{-2}(\bi^\prime\bS_n^{*}\bi)^2}\\
   &+&2\gamma^{-1}\frac{(\bi^\prime\bSigma_n^{-1}\bi)^{-1/2}\sqrt{\bm_n^\prime\bSigma_n^{-1}\bm_n}}{q_n}
  \frac{(\bi^\prime\bSigma_n^{-1}\bi)^{-1/2}(\bm_n^\prime\bSigma_n^{-1}\bm_n)^{-1/2}\sy_n^\prime\hat{\bQ}_n^{*}\bSigma_n\bS_n^{*}\bi}
  {(\bi^\prime\bSigma_n^{-1}\bi)^{-1}\bi^\prime\bS_n^{*}\bi}\\
  &+&\gamma^{-2}\frac{\bm_n^\prime\bSigma_n^{-1}\bm_n}{q_n}\frac{\sy_n^\prime\hat{\bQ}_n^{*}\bSigma_n\hat{\bQ}_n^{*}\sy_n}{\bm_n^\prime\bSigma_n^{-1}\bm_n}\\
  &-&2\frac{(\bi^\prime\bSigma_n^{-1}\bi)^{-1/2}\sqrt{\mathbf{b}^\prime\bSigma_n\mathbf{b}}}{q_n}
  \frac{(\bi^\prime\bSigma_n^{-1}\bi)^{-1/2}(\mathbf{b}^\prime\bSigma_n\mathbf{b})^{-1/2}\mathbf{b}^\prime\bSigma_n\bS_n^{*}\bi}
  {(\bi^\prime\bSigma_n^{-1}\bi)^{-1}\bi^\prime\bS^{*}_n\bi}\\
  &-&2\gamma^{-1}\frac{\sqrt{\bm_n^\prime\bSigma_n^{-1}\bm_n}\sqrt{\mathbf{b}^\prime\bSigma_n\mathbf{b}}}{q_n}
  \frac{\mathbf{b}^\prime\bSigma_n\hat{\bQ}_n^{*}\sy_n}{\sqrt{\bm_n^\prime\bSigma_n^{-1}\bm_n}\sqrt{\mathbf{b}^\prime\bSigma_n\mathbf{b}}}
  +\frac{\mathbf{b}^\prime\bSigma_n\mathbf{b}}{q_n}.
 \end{eqnarray*}

Using the equalities
\[\sy_n=\bm_n+\bSigma_n^{1/2}\sx_n$ and $\bS_n^{*}=\bSigma_n^{-1/2}\bV_n^{+}\bSigma_n^{-1/2},\]
the facts that $q_n^{-1}\bm_n^\prime\bSigma_n^{-1}\bm_n\le 1$, $q_n^{-1}\mathbf{b}^\prime\bSigma_n\mathbf{b}\le 1$, $q_n^{-1}(\bi^\prime\bSigma_n^{-1}\bi)^{-1} \le 1$ and that the Euclidean norms of the following vectors
\[\frac{\bSigma_n^{-1/2}\bi}{\sqrt{\bi^\prime\bSigma_n^{-1}\bi}},\quad \frac{\bSigma_n^{-1/2}\bm_n}{\sqrt{\bm_n^\prime\bSigma_n^{-1}\bm_n}}
\quad \text{and} \quad
\frac{\bSigma_n^{1/2}\mathbf{b}}{\sqrt{\mathbf{b}^\prime\bSigma_n\mathbf{b}}}\]
are equal to one, the application of Lemma \ref{lem6} yields
\begin{eqnarray*}
\frac{\bi^\prime\bS_n^{*}\bi}{\bi^\prime\bSigma_n^{-1}\bi}&=&\frac{\bi^\prime\bSigma_n^{-1/2}\bV_n^{+}\bSigma_n^{-1/2}\bi}{\bi^\prime\bSigma_n^{-1}\bi}
\stackrel{a.s.}{\longrightarrow}c^{-1}(c-1)^{-1},\\
\frac{\bi^\prime\bS_n^{*}\bm_n}{\sqrt{\bi^\prime\bSigma_n^{-1}\bi}\sqrt{\bm_n^\prime\bSigma_n^{-1}\bm_n}}&=&
\frac{\bi^\prime\bSigma_n^{-1/2}\bV_n^{+}\bSigma_n^{-1/2}\bm_n}{\sqrt{\bi^\prime\bSigma_n^{-1}\bi}\sqrt{\bm_n^\prime\bSigma_n^{-1}\bm_n}} \stackrel{a.s.}{\longrightarrow}c^{-1}(c-1)^{-1}  \frac{\bi^\prime\bSigma_n^{-1}\bm_n}{\sqrt{\bi^\prime\bSigma_n^{-1}\bi}\sqrt{\bm_n^\prime\bSigma_n^{-1}\bm_n}},\\
\frac{\bi^\prime\bS_n^{*}\bSigma_n\mathbf{b}}{\sqrt{\bi^\prime\bSigma_n^{-1}\bi}\sqrt{\mathbf{b}^\prime\bSigma_n\mathbf{b}}}&=&
\frac{\bi^\prime\bSigma_n^{-1/2}\bV_n^{+}\bSigma_n^{1/2}\mathbf{b}}{\sqrt{\bi^\prime\bSigma_n^{-1}\bi}\sqrt{\mathbf{b}^\prime\bSigma_n\mathbf{b}}}
\stackrel{a.s.}{\longrightarrow}c^{-1}(c-1)^{-1}  \frac{1}{\sqrt{\bi^\prime\bSigma_n^{-1}\bi}\sqrt{\mathbf{b}^\prime\bSigma_n\mathbf{b}}},\\
\frac{\bi^\prime\bS_n^{*}\bSigma_n\bS_n^{*}\bi}{\bi^\prime\bSigma_n^{-1}\bi}&=&\frac{\bi^\prime\bSigma_n^{-1/2}(\bV_n^{+})^2\bSigma_n^{-1/2}\bi}{\bi^\prime\bSigma_n^{-1}\bi} \stackrel{a.s.}{\longrightarrow}(c-1)^{-3}
\end{eqnarray*}
$\text{for}~ p/n\longrightarrow c \in (1,+\infty) ~\text{as} ~n\rightarrow\infty$.

Finally, from Lemma \ref{lem6} and \ref{lem7} as well as by using the equalities
\[\hat{\bQ}_n^*=\bSigma_n^{-1/2}\bP_n^+\bSigma_n^{-1/2} ~~\text{and}~~\bP_n^+\bV_n^{+}=(\bV_n^+)^{2}-\frac{\bV_n^{+}\boldsymbol{\eta}\boldsymbol{\eta} ^\prime(\bV_n^+)^{2}}
{\boldsymbol{\eta}^\prime\bV_n^{+}\boldsymbol{\eta}}\]
with $\boldsymbol{\eta}=\bSigma_n^{-1/2}\bi/\sqrt{\bi^\prime\bSigma_n^{-1/2}\bi}$ we obtain
\begin{eqnarray*}
\frac{\sy_n^\prime\hat{\bQ}_n^*\bm_n}{\bm_n^\prime\bSigma_n^{-1}\bm_n}&=&\frac{\bm_n^\prime \bSigma_n^{-1/2}\bP_n^+\bSigma_n^{-1/2}\bm_n}{\bm_n^\prime\bSigma_n^{-1}\bm_n}
+\frac{\sx_n^\prime\bP_n^+\bSigma_n^{-1/2}\bm_n}{\bm_n^\prime\bSigma_n^{-1}\bm_n}\\
&&\stackrel{a.s.}{\longrightarrow}c^{-1}(c-1)^{-1}\frac{\bm_n^\prime \bQ_n \bm_n}{\bm_n^\prime\bSigma_n^{-1}\bm_n},\\
\frac{\sy_n^\prime\hat{\bQ}_n^*\bSigma_n\mathbf{b}}{\sqrt{\bm_n^\prime\bSigma_n^{-1}\bm_n}\sqrt{\mathbf{b}^\prime\bSigma_n\mathbf{b}}}&=&\frac{\bm_n^\prime \bSigma_n^{-1/2}\bP_n^+\bSigma_n^{1/2}\mathbf{b}}{\sqrt{\bm_n^\prime\bSigma_n^{-1}\bm_n}\sqrt{\mathbf{b}^\prime\bSigma_n\mathbf{b}}}
+\frac{\sx_n^\prime\bP_n^+\bSigma_n^{1/2}\mathbf{b}}{\sqrt{\bm_n^\prime\bSigma_n^{-1}\bm_n}\sqrt{\mathbf{b}^\prime\bSigma_n\mathbf{b}}}\\
&&\stackrel{a.s.}{\longrightarrow}c^{-1}(c-1)^{-1}\frac{\bm_n^\prime \bQ_n\bSigma_n\mathbf{b}}{\sqrt{\bm_n^\prime\bSigma_n^{-1}\bm_n}\sqrt{\mathbf{b}^\prime\bSigma_n\mathbf{b}}},\\
\frac{\sy_n^\prime\hat{\bQ}_n^*\bSigma_n\bS_n^{*}\bi}{\sqrt{\bm_n^\prime\bSigma_n^{-1}\bm_n}\sqrt{\bi^\prime\bSigma_n^{-1}\bi}} &=&\frac{\bm_n^\prime\bSigma_n^{-1/2}\bP_n^+\bV_n^{+}\bSigma_n^{-1/2}\bi}{\sqrt{\bm_n^\prime\bSigma_n^{-1}\bm_n}\sqrt{\bi^\prime\bSigma_n^{-1}\bi}}
+\frac{\sx_n^\prime\bP_n^+\bV_n^{+}\bSigma_n^{-1/2}\bi}{\sqrt{\bm_n^\prime\bSigma_n^{-1}\bm_n}\sqrt{\bi^\prime\bSigma_n^{-1}\bi}} \\
&&\stackrel{a.s.}{\longrightarrow}(c-1)^{-3} \frac{\bm_n^\prime\bQ_n\bi}{\sqrt{\bm_n^\prime\bSigma_n^{-1}\bm_n}\sqrt{\bi^\prime\bSigma_n^{-1}\bi}}=0,\\
\frac{\sy_n^\prime\hat{\bQ}_n^*\bSigma_n\hat{\bQ}_n^*\sy_n}{\bm_n^\prime\bSigma_n^{-1}\bm_n}&=&
\frac{\sx_n^\prime(\bP_n^+)^2\sx_n}{\bm_n^\prime\bSigma_n^{-1}\bm_n}+2\frac{\bm_n^\prime\bSigma_n^{-1/2}(\bP_n^+)^2\sx_n}{\bm_n^\prime\bSigma_n^{-1}\bm_n}
+\frac{\bm_n^\prime\bSigma_n^{-1/2}(\bP_n^+)^2\bSigma_n^{-1/2}\bm_n}{\bm_n^\prime\bSigma_n^{-1}\bm_n}\\
&&\stackrel{a.s.}{\longrightarrow}\frac{(c-1)^{-3}c^2 + (c-1)^{-3} \bm_n^\prime\bQ_n\bm_n}{\bm_n^\prime\bSigma_n^{-1}\bm_n}
\end{eqnarray*}
$\text{for}~ p/n\longrightarrow c \in (1,+\infty) ~\text{as} ~n\rightarrow\infty$.

Substituting the above results into the expressions of $A_n^*$ and $B_n^*$, we get that
\[|A_n^*-A^*|\stackrel{a.s.}{\longrightarrow}0
\quad \text{and} \quad
|B_n^*-B^*|\stackrel{a.s.}{\longrightarrow}0\]
with
{\footnotesize
\begin{eqnarray*}
A^+ &=& \frac{(\bi^\prime\bSigma_n^{-1}\bi)^{-1/2}\sqrt{\bm_n^\prime\bSigma_n^{-1}\bm_n}}{q_n} \frac{(\bi^\prime\bSigma_n^{-1}\bi)^{-1/2}(\bm_n^\prime\bSigma_n^{-1}\bm_n)^{-1/2}c^{-1}(c-1)^{-1}\bi^\prime\bSigma_n^{-1}\bm_n}
  {c^{-1}(c-1)^{-1}}\\
  &-&\beta \frac{(\bi^\prime\bSigma_n^{-1}\bi)^{-1/2}\sqrt{\mathbf{b}^\prime\bSigma_n\mathbf{b}}}{q_n} \frac{(\bi^\prime\bSigma_n^{-1}\bi)^{-1/2}(\mathbf{b}^\prime\bSigma_n\mathbf{b})^{-1/2}c^{-1}(c-1)^{-1}}
  {c^{-1}(c-1)^{-1}}\\
  &+&\gamma^{-1}\frac{\bm_n^\prime\bSigma_n^{-1}\bm_n}{q_n}\frac{c^{-1}(c-1)^{-1}\bm_n^\prime\bQ_n\bm_n}{\bm_n^\prime\bSigma_n^{-1}\bm_n}
  -\beta\gamma^{-1}\frac{\sqrt{\bm_n^\prime\bSigma_n^{-1}\bm_n}\sqrt{\mathbf{b}^\prime\bSigma_n\mathbf{b}}}{q_n}
  \frac{c^{-1}(c-1)^{-1}\bm_n^\prime\bQ_n\bSigma_n\mathbf{b}}{\sqrt{\bm_n^\prime\bSigma_n^{-1}\bm_n}\sqrt{\mathbf{b}^\prime\bSigma_n\mathbf{b}}}\\
  &-&\frac{\sqrt{\bm_n^\prime\bSigma_n^{-1}\bm_n}\sqrt{\mathbf{b}^\prime\bSigma_n\mathbf{b}}}{q_n}\frac{\mathbf{b}^\prime\bm_n}
  {\sqrt{\bm_n^\prime\bSigma_n^{-1}\bm_n}\sqrt{\mathbf{b}^\prime\bSigma_n\mathbf{b}}} +\beta\frac{\mathbf{b}^\prime\bSigma_n\mathbf{b}}{q_n}\\
&=&\frac{1}{q_n}\left(\frac{\bi^\prime\bSigma_n^{-1}\bm_n}{\bi^\prime\bSigma_n^{-1}\bi}-\beta\frac{1}{\bi^\prime\bSigma_n^{-1}\bi}
+\frac{\gamma^{-1}}{c(c-1)}\bm_n^\prime\bQ_n\bm_n-\frac{\gamma^{-1}\beta}{c(c-1)}\bm_n^\prime\bQ_n\bSigma_n\mathbf{b} - \mathbf{b}^\prime\bm_n+\beta\mathbf{b}^\prime\bSigma_n\mathbf{b}\right)
\end{eqnarray*}
and
\begin{eqnarray*}
B^+ &=& \frac{(\bi^\prime\bSigma_n^{-1}\bi)^{-1}}{q_n}
   \frac{(c-1)^{-3}}{c^{-2}(c-1)^{-2}}+2\gamma^{-1}\frac{(\bi^\prime\bSigma_n^{-1}\bi)^{-1/2}\sqrt{\bm_n^\prime\bSigma_n^{-1}\bm_n}}{q_n}
  \frac{0}{c^{-1}(c-1)^{-1}}\\
  &+&\gamma^{-2}\frac{\bm_n^\prime\bSigma_n^{-1}\bm_n}{q_n}\frac{(c-1)^{-3}c^2 + (c-1)^{-3} \bm_n^\prime\bQ_n\bm_n}{\bm_n^\prime\bSigma_n^{-1}\bm_n}\\
  &-&2\frac{(\bi^\prime\bSigma_n^{-1}\bi)^{-1/2}\sqrt{\mathbf{b}^\prime\bSigma_n\mathbf{b}}}{q_n}
  \frac{(\bi^\prime\bSigma_n^{-1}\bi)^{-1/2}(\mathbf{b}^\prime\bSigma_n\mathbf{b})^{-1/2}c^{-1}(c-1)^{-1}}
  {c^{-1}(c-1)^{-1}}\\
  &-&2\gamma^{-1}\frac{\sqrt{\bm_n^\prime\bSigma_n^{-1}\bm_n}\sqrt{\mathbf{b}^\prime\bSigma_n\mathbf{b}}}{q_n}
  \frac{c^{-1}(c-1)^{-1}\mathbf{b}^\prime\bSigma_n\bQ_n\bm_n}{\sqrt{\bm_n^\prime\bSigma_n^{-1}\bm_n}\sqrt{\mathbf{b}^\prime\bSigma_n\mathbf{b}}}
  +\frac{\mathbf{b}^\prime\bSigma_n\mathbf{b}}{q_n}\\
&=&\frac{1}{q_n}\Bigg(\frac{c^2}{c-1}\frac{1}{\bi^\prime\bSigma_n^{-1}\bi}+\frac{\gamma^{-2}}{(c-1)^{3}}\left(c^2 + \bm_n^\prime\bQ_n\bm_n\right)
- 2\frac{1}{\bi^\prime\bSigma_n^{-1}\bi}\\
&-&2\frac{\gamma^{-1}}{c(c-1)}\left(\mathbf{b}^\prime\bm_n
 -\frac{\bi^\prime\bSigma_n^{-1}\bm_n}{\bi^\prime\bSigma_n^{-1}\bi}\right)+\mathbf{b}^\prime\bSigma_n\mathbf{b}\Bigg).
\end{eqnarray*}
}

Let $\alpha^+=A^+/B^+$. Hence,
\[|\alpha_n^+-\alpha^+|\le \left|\frac{1}{B_n^+}(A_n^+-A^+)\right|+\left|\frac{A^+}{B^+B_n^+}(B^+-B_n^+)\right|\stackrel{a.s.}{\longrightarrow}0.\]
This completes the proof of Theorem 2.2.
\end{proof}

\begin{proof}[Proof of Corollary 2.2:]
\noindent (a)
With $q_n=\max\{\bm_n^\prime\bSigma_n^{-1}\bm_n,\mathbf{b}^\prime\bSigma_n\mathbf{b}\}$ we get
\begin{eqnarray*}
\frac{1}{q_n}U_S&=& \frac{1}{q_n}\hat{\bw}_{S^*}^\prime\bm_n-\dfrac{\gamma}{2}\frac{1}{q_n}\hat{\bw}_{S^*}^\prime\bSigma_n\hat{\bw}_{S^*},
\end{eqnarray*}
where from the proof of Theorem 2.2 we have
\begin{eqnarray*}
\frac{1}{q_n}\hat{\bw}_{S^*}^\prime\bm_n&=&\frac{1}{q_n}\frac{\bi^\prime\bS_n^{*}\bm_n}{\bi^\prime\bS_n^{*}\bi}+\gamma^{-1}\frac{1}{q_n}\sy_n^\prime\hat{\bQ}_n^*\bm_n\\
&=& \frac{(\bi^\prime\bSigma_n^{-1}\bi)^{-1/2}\sqrt{\bm_n^\prime\bSigma_n^{-1}\bm_n}}{q_n} \frac{(\bi^\prime\bSigma_n^{-1}\bi)^{-1/2}(\bm_n^\prime\bSigma_n^{-1}\bm_n)^{-1/2}\bi^\prime\bS_n^{*}\bm_n}
  {(\bi^\prime\bSigma_n^{-1}\bi)^{-1}\bi^\prime\bS_n^{*}\bi}\\
  &+&\gamma^{-1} \frac{\bm_n^\prime\bSigma_n^{-1}\bm_n}{q_n}\frac{\sy_n^\prime\hat{\bQ}_n^*\bm_n}{\bm_n^\prime\bSigma_n^{-1}\bm_n} \\
&\stackrel{a.s.}{\longrightarrow}&\frac{(\bi^\prime\bSigma_n^{-1}\bi)^{-1/2}\sqrt{\bm_n^\prime\bSigma_n^{-1}\bm_n}}{q_n} \frac{(\bi^\prime\bSigma_n^{-1}\bi)^{-1/2}(\bm_n^\prime\bSigma_n^{-1}\bm_n)^{-1/2}c^{-1}(c-1)^{-1}\bi^\prime\bSigma_n^{-1}\bm_n}
  {c^{-1}(c-1)^{-1}}\\
  &+&\gamma^{-1}c^{-1}(c-1)^{-1}\frac{\bm_n^\prime\bSigma_n^{-1}\bm_n}{q_n}\frac{\bm_n^\prime \bQ_n \bm_n}{\bm_n^\prime\bSigma_n^{-1}\bm_n}
  =\frac{1}{q_n}\left(R_{GMV}+\gamma^{-1}c^{-1}(c-1)^{-1}s\right)
\end{eqnarray*}
and
\begin{eqnarray*}
&&\hat{\bw}_{S^*}^\prime\bSigma_n\hat{\bw}_{S^*}=\frac{\bi^\prime\bS_n^{*}\bSigma_n\bS_n^{*}\bi}{(\bi^\prime\bS_n^{*}\bi)^2}
+2\gamma^{-1}\frac{\bi^\prime\bS_n^{*}\bSigma_n\hat{\bQ}_n^*\sy_n}{\bi^\prime\bS_n^{*}\bi}+\gamma^{-2}\sy_n^\prime\hat{\bQ}_n^*\bSigma_n\hat{\bQ}_n^*\sy_n\\
&=&\frac{(\bi^\prime\bSigma_n^{-1}\bi)^{-1}}{q_n}
   \frac{(\bi^\prime\bSigma_n^{-1}\bi)^{-1}\bi^\prime\bS_n^{*}\bSigma_n\bS_n^{*}\bi}{(\bi^\prime\bSigma_n^{-1}\bi)^{-2}(\bi^\prime\bS_n^{*}\bi)^2}
   +\gamma^{-2}\frac{\bm_n^\prime\bSigma_n^{-1}\bm_n}{q_n}\frac{\sy_n^\prime\hat{\bQ}_n^*\bSigma_n\hat{\bQ}_n^*\sy_n}{\bm_n^\prime\bSigma_n^{-1}\bm_n}\\
&+& 2\gamma^{-1} \frac{(\bi^\prime\bSigma_n^{-1}\bi)^{-1/2}\sqrt{\bm_n^\prime\bSigma_n^{-1}\bm_n}}{q_n} \frac{(\bi^\prime\bSigma_n^{-1}\bi)^{-1/2}(\bm_n^\prime\bSigma_n^{-1}\bm_n)^{-1/2}\bi^\prime\bS_n^{*}\bSigma_n\hat{\bQ}_n^*\sy_n}
{(\bi^\prime\bSigma_n^{-1}\bi)^{-1}\bi^\prime\bS_n^{*}\bi} \\
&\stackrel{a.s.}{\longrightarrow}&\frac{(\bi^\prime\bSigma_n^{-1}\bi)^{-1}}{q_n}\frac{(c-1)^{-3}}{c^{-2}(c-1)^{-2}}
+\gamma^{-2}\frac{\bm_n^\prime\bSigma_n^{-1}\bm_n}{q_n}\frac{(c-1)^{-3}c^2 + (c-1)^{-3} \bm_n^\prime\bQ_n\bm_n}{\bm_n^\prime\bSigma_n^{-1}\bm_n}\\
&+&2\gamma^{-1}\frac{(\bi^\prime\bSigma_n^{-1}\bi)^{-1/2}\sqrt{\bm_n^\prime\bSigma_n^{-1}\bm_n}}{q_n} \frac{0}{c^{-1}(c-1)^{-1}}\\
&=&\frac{1}{q_n}\left(c^{2}(c-1)^{-1}V_{GMV}+\gamma^{-2}\frac{c^2}{(c-1)^{3}}  +\gamma^{-2} (c-1)^{-3} s\right)
\end{eqnarray*}
$\text{for}~ p/n\longrightarrow c \in (1,+\infty) ~\text{as} ~n\rightarrow\infty$.

Finally, using that
\[\frac{1}{q_n}U_{EU}=\frac{1}{q_n}R_{GMV}+\frac{1}{2}\gamma^{-1} \frac{1}{q_n}s - \frac{\gamma}{2} \frac{1}{q_n}V_{GMV},\]
we obtain the statement of the first part of the corollary.

\noindent (b) It holds that
\begin{eqnarray*}
\frac{1}{q_n}U_{GSE} &=& \alpha^+ \frac{1}{q_n}\hat{\bw}_{S^*}^\prime\bm_n+(1-\alpha^+)\frac{1}{q_n}\mathbf{b}^\prime\bm_n\\
&-&\dfrac{\gamma}{2}\left((\alpha^+)^2\frac{1}{q_n}\hat{\bw}_{S^*}^\prime\bSigma_n\hat{\bw}_{S^*}
+2\alpha^+(1-\alpha^+)\frac{1}{q_n}\mathbf{b}^\prime\bSigma_n\hat{\bw}_{S^*}+(1-\alpha^+)^2\frac{1}{q_n}\mathbf{b}^\prime\bSigma_n\mathbf{b}\right),
\end{eqnarray*}
where
\begin{eqnarray*}
\frac{1}{q_n}\mathbf{b}^\prime\bSigma_n\hat{\bw}_{S^*}&=&
\frac{1}{q_n}\frac{\mathbf{b}^\prime\bSigma_n\bS_n^{*}\bi}{\bi^\prime\bS_n^{*}\bi}
+\gamma^{-1}\frac{1}{q_n}\mathbf{b}^\prime\bSigma_n\hat{\bQ}_n^*\sy_n\\
&=&\frac{(\bi^\prime\bSigma_n^{-1}\bi)^{-1/2}\sqrt{\mathbf{b}^\prime\bSigma_n\mathbf{b}}}{q_n} \frac{(\bi^\prime\bSigma_n^{-1}\bi)^{-1/2}(\mathbf{b}^\prime\bSigma_n\mathbf{b})^{-1/2}\bi^\prime\bS_n^{*}\bSigma_n\mathbf{b}}
  {(\bi^\prime\bSigma_n^{-1}\bi)^{-1}\bi^\prime\bS_n^{*}\bi}\\
&+& \gamma^{-1}\frac{\sqrt{\bm_n^\prime\bSigma_n^{-1}\bm_n}\sqrt{\mathbf{b}^\prime\bSigma_n\mathbf{b}}}{q_n}
  \frac{\mathbf{b}^\prime\bSigma_n\hat{\bQ}_n^*\sy_n}{\sqrt{\bm_n^\prime\bSigma_n^{-1}\bm_n}\sqrt{\mathbf{b}^\prime\bSigma_n\mathbf{b}}} \\
&\stackrel{a.s.}{\longrightarrow}&
\frac{(\bi^\prime\bSigma_n^{-1}\bi)^{-1/2}\sqrt{\mathbf{b}^\prime\bSigma_n\mathbf{b}}}{q_n} \frac{(\bi^\prime\bSigma_n^{-1}\bi)^{-1/2}(\mathbf{b}^\prime\bSigma_n\mathbf{b})^{-1/2}c^{-1}(c-1)^{-1}}
  {c^{-1}(c-1)^{-1}}\\
&+&\gamma^{-1}\frac{\sqrt{\bm_n^\prime\bSigma_n^{-1}\bm_n}\sqrt{\mathbf{b}^\prime\bSigma_n\mathbf{b}}}{q_n}
  \frac{c^{-1}(c-1)^{-1}\left(\mathbf{b}^\prime\bm_n-\dfrac{\bi^\prime\bSigma_n^{-1}\bm_n}{\bi^\prime\bSigma_n^{-1}\bi}\right)}
  {\sqrt{\bm_n^\prime\bSigma_n^{-1}\bm_n}\sqrt{\mathbf{b}^\prime\bSigma_n\mathbf{b}}}
\\
&=&\frac{1}{q_n}\left(V_{GMV}+\gamma^{-1}c^{-1}(c-1)^{-1} (R_b-R_{GMV})\right)
\end{eqnarray*}
$\text{for}~ p/n\longrightarrow c \in (1,+\infty) ~\text{as} ~n\rightarrow\infty$.

Hence,
{\footnotesize
\begin{eqnarray*}
&&\frac{1}{q_n}(U_{EU}-U_{GSE})= (\alpha^+)^2\frac{1}{q_n}U_{EU}+(1-\alpha^+)^2\frac{1}{q_n}U_{EU}+2\alpha^+(1-\alpha^+)\frac{1}{q_n}U_{EU}\\
&-&(\alpha^+)^2\frac{1}{q_n}U_{S}-\alpha^+(1-\alpha^+)\frac{1}{q_n}\hat{\bw}_{S^*}^\prime\bm_n
-(1-\alpha^+)^2\frac{1}{q_n}U_{\bb}-\alpha^+(1-\alpha^+)\frac{1}{q_n}\mathbf{b}^\prime\bm_n
+\gamma\alpha^+(1-\alpha^+)\frac{1}{q_n}\mathbf{b}^\prime\bSigma_n\hat{\bw}_{S^*}\\
&\stackrel{a.s.}{\longrightarrow}&(\alpha^+)^2\frac{1}{q_n}(U_{EU}-U_S)+(1-\alpha^+)^2\frac{1}{q_n}(U_{EU}-U_{\bb})\\
&+&\alpha^+(1-\alpha^+)\frac{1}{q_n}\left(2R_{GMV}+\gamma^{-1} s - \gamma V_{GMV}-R_{GMV}-\frac{\gamma^{-1}c^{-1}}{c-1} s-R_b+\gamma V_{GMV}+\frac{R_b-R_{GMV}}{c(c-1)}\right)\\
&=&(\alpha^+)^2\frac{1}{q_n}(U_{EU}-U_S)+(1-\alpha^+)^2\frac{1}{q_n}(U_{EU}-U_{\bb})
+\alpha^+(1-\alpha^+)\frac{1+c-c^2}{c(c-1)}\frac{1}{q_n}(R_b-R_{GMV}-\gamma^{-1} s).
\end{eqnarray*}
}
$\text{for}~ p/n\longrightarrow c \in (0,1) ~\text{as} ~n\rightarrow\infty$.
\end{proof}

\setlength{\bibsep}{4pt}
\begin{small}
\bibliography{MV}
\end{small}


\end{document}